\newif\ifdraft\draftfalse  
\newif\ifanon\anonfalse    
\newif\ifappendices\appendicestrue

\PassOptionsToPackage{usenames,dvipsnames,svgnames,table}{xcolor}
\documentclass[acmsmall,screen]{acmart}
\setcopyright{rightsretained}
\acmPrice{}
\acmDOI{10.1145/3341699}
\acmYear{2019}
\copyrightyear{2019}
\acmJournal{PACMPL}
\acmVolume{3}
\acmNumber{ICFP}
\acmArticle{95}
\acmMonth{8}

\usepackage[usenames,dvipsnames,svgnames,table]{xcolor}
\usepackage[capitalise]{cleveref}
\usepackage{amsmath}
\usepackage{makecell}
\usepackage{nccmath}
\usepackage{mathtools}
\usepackage{bussproofs}
\usepackage{varwidth}
\usepackage{amsthm}
\usepackage{csvsimple}
\usepackage{thmtools,thm-restate}
\usepackage{changepage}
\usepackage{booktabs}
\usepackage{amssymb}
\usepackage{enumitem}
\usepackage{multirow,bigdelim}
\usepackage{multicol}
\usepackage{siunitx}
\usepackage{listings}
\usepackage{letltxmacro}
\usepackage{sansmath}
\usepackage{url}
\usepackage{flushend}
\usepackage{microtype}
\usepackage[utf8]{inputenc}
\usepackage{mathpartir}
\usepackage{empheq}
\usepackage{array}
\usepackage{pgfplots}
\usepackage{stmaryrd}
\usepackage{courier}
\usepackage{qtree}
\usepackage[normalem]{ulem}
\usepackage{relsize}
\usepackage{tikz}
\usepackage{algorithm}
\usepackage[noend]{algpseudocode}
\usepackage{graphicx}
\usepackage{subcaption}
\usepackage{textcomp}
\usepackage{tabularx}
\usepackage{stackengine}
\usepackage{caption}
\usepackage{wrapfig}
\usepackage{remreset}
\usepackage{tabulary}
\usepackage{xspace}
\usepackage{bbm}

\newtheorem*{theorem*}{Theorem}
\newenvironment{centermath}
 {\begin{center}$\displaystyle}
 {$\end{center}}
\setcellgapes{4pt}

\makeatletter
     \let\lst@oldvisiblespace\lst@visiblespace
     \def\lst@visiblespace{\,\lst@oldvisiblespace\,}
\makeatother

\usetikzlibrary{
  er,
  matrix,
  shapes,
  arrows,
  positioning,
  fit,
  calc,
  pgfplots.groupplots,
  arrows.meta
}
\tikzset{>={Latex}}

\definecolor{dkblue}{rgb}{0,0.1,0.5}
\definecolor{dkgreen}{rgb}{0,0.6,0}
\definecolor{dkred}{rgb}{0.6,0,0}
\definecolor{dkpurple}{rgb}{0.2,0,0.4}
\definecolor{olive}{rgb}{0.4, 0.4, 0.0}
\definecolor{teal}{rgb}{0.0,0.5,0.5}
\definecolor{orange}{rgb}{0.9,0.6,0.2}
\definecolor{lightyellow}{RGB}{255, 255, 240}
\definecolor{vlightyellow}{RGB}{255, 255, 245}
\definecolor{lightgreen}{RGB}{235, 255, 240}
\definecolor{vlightgreen}{RGB}{247, 255, 250}
\definecolor{lightblue}{RGB}{245, 255, 255}
\definecolor{vlightblue}{RGB}{250, 255, 255}
\definecolor{teal}{RGB}{141,211,199}
\definecolor{darkbrown}{RGB}{121,37,0}

\makeatletter
\def\renewtheorem#1{%
  \expandafter\let\csname#1\endcsname\relax
  \expandafter\let\csname c@#1\endcsname\relax
  \gdef\renewtheorem@envname{#1}
  \renewtheorem@secpar
}
\def\renewtheorem@secpar{\@ifnextchar[{\renewtheorem@numberedlike}{\renewtheorem@nonumberedlike}}
\def\renewtheorem@numberedlike[#1]#2{\newtheorem{\renewtheorem@envname}[#1]{#2}}
\def\renewtheorem@nonumberedlike#1{  
\def\renewtheorem@caption{#1}
\edef\renewtheorem@nowithin{\noexpand\newtheorem{\renewtheorem@envname}{\renewtheorem@caption}}
\renewtheorem@thirdpar
}
\def\renewtheorem@thirdpar{\@ifnextchar[{\renewtheorem@within}{\renewtheorem@nowithin}}
\def\renewtheorem@within[#1]{\renewtheorem@nowithin[#1]}
\makeatother

\pgfplotsset{
    every non boxed x axis/.style={} 
}

\newenvironment{mathprooftree}
  {\varwidth{.9\textwidth}\centering\leavevmode}
  {\DisplayProof\endvarwidth}

\newcommand{\IE}{\emph{i.e.}\xspace}
\newcommand{\EG}{\emph{e.g.}\xspace}

\newcommand{\ETAL}{\emph{et al.}\xspace}

\theoremstyle{definition}
\renewtheorem{theorem}{Theorem}
\newtheorem{mylemma}{Lemma}
\renewtheorem{corollary}{Corollary}
\renewtheorem{conjecture}{Conjecture}
\renewtheorem{definition}{Definition}
\newtheorem{property}{Property}
\theoremstyle{plain}
\theoremstyle{remark}

\theoremstyle{remark}
\newtheorem{case}{Case}
\makeatletter
\@addtoreset{subcase}{case}
\@addtoreset{case}{mylemma}
\@addtoreset{case}{theorem}
\@addtoreset{case}{corollary}
\@addtoreset{case}{definition}
\@addtoreset{case}{definition}
\@removefromreset{theorem}{section}
\makeatother

\algnewcommand\algorithmicswitch{\textbf{switch}}
\algnewcommand\algorithmicmatch{\textbf{match}}
\algnewcommand\algorithmiccase{\textbf{case}}
\algnewcommand\algorithmicwith{\textbf{with}}
\algnewcommand\algorithmicforeach{\textbf{foreach}}
\algnewcommand\Assert[1]{\State \algorithmicassert(#1)}%
\algdef{SE}[SWITCH]{Switch}{EndSwitch}[1]{\algorithmicmatch\ #1\ \algorithmicwith}{\algorithmicend\ \algorithmicswitch}%
\algdef{SE}[CASE]{Case}{EndCase}[1]{$|~$ #1 $\rightarrow$}{\algorithmicend\ \algorithmiccase}%
\algdef{SE}[FOREACH]{ForEach}{EndForEach}[2]{\algorithmicforeach\ #1 $\in$ #2}{\algorithmicend\ \algorithmicforeach}%
\algdef{SE}[CaseTwo]{CaseTwo}{EndCaseTwo}[2]{$|~$ #1 $\rightarrow$ #2}{\algorithmicend\ \algorithmiccase}%
\algtext*{EndSwitch}%
\algtext*{EndCase}%
\algtext*{EndCaseTwo}%
\algtext*{EndSecondCase}%
\algtext*{EndForEach}%

\newcommand{\CF}[1]{\ensuremath{\mathsf{#1}}}         
\newcommand{\get}{\ensuremath{\CF{get}}\xspace}
\newcommand{\pput}{\ensuremath{\CF{put}}\xspace}
\newcommand{\create}{\ensuremath{\CF{create}}\xspace}

\newcommand{\PCF}[1]{\textproc{#1}}

\newcommand{\KW}[1]{\CF{\textcolor{dkpurple}{#1}}}
\newcommand{\RegexNS}{\ensuremath{\mathit{S}}}
\newcommand{\Regex}{\RegexNS\xspace}         

\newcommand{\BooleanAnd}{\ensuremath{~\wedge~}}
\newcommand{\BooleanOr}{\ensuremath{\vee}}
\newcommand{\BooleanImplies}{\ensuremath{\Rightarrow}}
\newcommand{\Rewrite}{\ensuremath{\rightarrow}}

\newcommand{\ConcatDNF}{\ensuremath{\odot}}
\newcommand{\ConcatDNFOf}[2]{\ensuremath{#1\ConcatDNF#2}}

\newcommand{\ConcatSequence}{\ensuremath{\odot_{\Sequence}}}
\newcommand{\ConcatSequenceOf}[2]{\ensuremath{#1\ConcatSequence#2}}

\newcommand{\OrDNF}{\ensuremath{\oplus}}
\newcommand{\OrDNFOf}[3]{\ensuremath{#1\OrDNF_{#3}#2}}

\newcommand{\PutRSym}{\ensuremath{\mathit{putr}}}
\newcommand{\PutLSym}{\ensuremath{\mathit{putl}}}
\newcommand{\PutRSymOf}[1]{\ensuremath{\PutRSym \App #1}}
\newcommand{\PutLSymOf}[1]{\ensuremath{\PutLSym \App #1}}
\newcommand{\RegexAltNS}{\ensuremath{\mathit{T}}}         
\newcommand{\RegexAlt}{\RegexAltNS\xspace}         
\newcommand{\RegexAltAltNS}{\ensuremath{\mathit{U}}}         
\newcommand{\Or}{\ensuremath{~|~}}
\newcommand{\RegexOr}[2]{\ensuremath{#1\Or#2}}

\newcommand{\RegexConcat}[2]{\ensuremath{#1\cdot#2}}
\newcommand{\EmptyString}{\ensuremath{\lstinline{""}}}

\newcommand{\Concat}{\ensuremath{\cdot}}

\newcommand{\SetOf}[1]{\ensuremath{\{#1\}}}

\newcommand{\Atom}{\ensuremath{\mathit{A}}}          
\newcommand{\AtomAlt}{\ensuremath{\mathit{B}}}
\newcommand{\AtomType}{\ensuremath{\mathit{Atom}}}
\newcommand{\App}{\ensuremath{\,}}
\newcommand{\Sequence}{\ensuremath{\mathit{SQ}}}
\newcommand{\SequenceType}{\ensuremath{\mathit{Sequence}}}
\newcommand{\LetIn}[2]{\ensuremath{\text{let } #1 = #2\text{ in }}}

\newcommand{\SequenceAlt}{\ensuremath{\mathit{TQ}}}
\newcommand{\DNFRegex}{\ensuremath{\mathit{DS}}}         
\newcommand{\DNFRegexAlt}{\ensuremath{\mathit{DT}}}    
\newcommand{\DNFRegexType}{\ensuremath{\mathit{DNF}}}

\newcommand{\String}{\ensuremath{\mathit{s}}\xspace}        
\newcommand{\StringAlt}{\ensuremath{\mathit{t}}}        
\newcommand{\StringAltAlt}{\ensuremath{\mathit{u}}}        

\newcommand{\LanguageOf}[1]{\ensuremath{\mathcal{L}(#1)}}

\newcommand{\RangeIncInc}[2]{\ensuremath{[#1,#2]}}

\newcommand{\Lens}{\ensuremath{\mathit{\ell}}\xspace}
\newcommand{\AtomLens}{\ensuremath{\mathit{al}}}

\newcommand{\SequenceLens}{\ensuremath{\mathit{sql}}}

\newcommand{\Star}{\ensuremath{^*}}
\newcommand{\StarOf}[1]{\ensuremath{{#1}\Star}}

\newcommand{\ConcatLens}{\ensuremath{\KW{concat}}\xspace}
\newcommand{\ConcatLensOf}[2]{\ensuremath{\ConcatLens(#1,#2)}}

\newcommand{\SwapLens}{\ensuremath{\KW{swap}}\xspace}
\newcommand{\SwapLensOf}[2]{\ensuremath{\SwapLens(#1,#2)}}

\newcommand{\OrLens}{\ensuremath{\KW{or}}\xspace}
\newcommand{\OrLensOf}[2]{\ensuremath{\OrLens(#1,#2)}}
\newcommand{\Skip}{\ensuremath{\KW{skip}}\xspace}
\newcommand{\SkipOf}[1]{\ensuremath{\Skip(#1)}\xspace}
\newcommand{\SRequire}{\ensuremath{\KW{require}}\xspace}
\newcommand{\SRequireOf}[1]{\ensuremath{\SRequire(#1)}\xspace}
\newcommand{\IdentityLens}{\ensuremath{\KW{id}}}
\newcommand{\IdentityLensOf}[1]{\ensuremath{\IdentityLens(#1)}}

\newcommand{\IterateLens}{\ensuremath{\KW{iterate}}\xspace}
\newcommand{\IterateLensOf}[1]{\ensuremath{\mathit{\IterateLens(#1)}}}

\newcommand{\ComposeLensOf}[2]{\ensuremath{#1\mathrel{;}#2}}
\newcommand{\Disconnect}{\ensuremath{\KW{disconnect}}\xspace}
\newcommand{\DisconnectOf}[4]{\ensuremath{\Disconnect(#1,#2,#3,#4)}}
\newcommand{\MergeL}{\ensuremath{\KW{merge\_left}}\xspace}
\newcommand{\MergeLOf}[2]{\ensuremath{\MergeL(#1,#2)}}
\newcommand{\MergeR}{\ensuremath{\KW{merge\_right}}\xspace}
\newcommand{\MergeROf}[2]{\ensuremath{\MergeR(#1,#2)}}
\newcommand{\Invert}{\ensuremath{\KW{invert}}}
\newcommand{\InvertOf}[1]{\ensuremath{\Invert(#1)}}

\newcommand{\GEq}{\ensuremath{::=~}}

\newcommand{\Nats}{\ensuremath{\mathbb{N}}}

\newcommand{\OfType}{\ensuremath{:}}

\newcommand{\ArrowTypeOf}[2]{\ensuremath{#1 \rightarrow #2}}

\newcommand{\ToDNFRegex}{\ensuremath{\Downarrow}\xspace}
\newcommand{\ToDNFRegexOf}[1]{\ensuremath{\ToDNFRegex\mkern-4mu #1}}

\newcommand{\SuchThat}{\ensuremath{~|~}}

\newcommand{\Given}{\ensuremath{~|~}}

\newcommand{\SequenceLeft}{\ensuremath{[}}
\newcommand{\SequenceRight}{\ensuremath{]}}
\newcommand{\SequenceOf}[1]{\ensuremath{\SequenceLeft#1\SequenceRight}}
\newcommand{\SeqSep}{\ensuremath{\mkern-1mu\Concat\mkern-1mu}}
\newcommand{\DNFLeft}{\ensuremath{\langle}}
\newcommand{\DNFRight}{\ensuremath{\rangle}}
\newcommand{\DNFOf}[1]{\ensuremath{\DNFLeft#1\DNFRight}}
\newcommand{\DNFSep}{\ensuremath{\Or}}
\newcommand{\SequenceLensLeft}{\ensuremath{[}}
\newcommand{\SequenceLensRight}{\ensuremath{]}}

\newcommand{\SeqLSep}{\ensuremath{\mkern-1mu\Concat\mkern-1mu}}
\newcommand{\DNFLensLeft}{\ensuremath{\langle}}
\newcommand{\DNFLensRight}{\ensuremath{\rangle}}

\newcommand{\DNFLSep}{\ensuremath{\Or}}

\newcommand{\LensBuilder}{\ensuremath{\mathit{lb}}\xspace}
\newcommand{\SeqLensBuilder}{\ensuremath{\mathit{slb}}\xspace}

\newcommand{\SubLeft}{\textsubscript{L}}
\newcommand{\SubRight}{\textsubscript{R}}

\newcommand{\OrIdentityRule}{\textit{\Or{} Ident}}
\newcommand{\EmptyProjectionRightRule}{0-Proj\SubRight{}}
\newcommand{\EmptyProjectionLeftRule}{0-Proj\SubLeft{}}
\newcommand{\ConcatAssocRule}{\textit{\Concat{} Assoc}}
\newcommand{\OrAssociativityRule}{\textit{\Or{} Assoc}}
\newcommand{\OrCommutativityRule}{\textit{\Or{} Comm}}
\newcommand{\DistributivityLeftRule}{\textit{Dist\SubRight{}}}
\newcommand{\DistributivityRightRule}{\textit{Dist\SubLeft{}}}
\newcommand{\ConcatIdentityLeftRule}{\textit{\Concat{} Ident\SubLeft{}}}
\newcommand{\ConcatIdentityRightRule}{\textit{\Concat{} Ident\SubRight{}}}

\newcommand{\UnrollstarLeftRule}{\textit{Unroll\SubLeft{}}}
\newcommand{\UnrollstarRightRule}{\textit{Unroll\SubRight{}}}

\newcolumntype{q}{>{$}l<{$}}
\newcolumntype{v}{>{$}r<{$}}

\renewcommand{\subsubsection}[1]{\paragraph{{#1}}}

\newcommand{\Examples}{\ensuremath{\mathit{exs}}}

\newcommand{\ListLeft}{\ensuremath{[}}
\newcommand{\ListRight}{\ensuremath{]}}
\newcommand{\ListOf}[1]{\ensuremath{\ListLeft #1 \ListRight}}

\newcommand{\NormalizedStarOf}[1]{\ensuremath{\NormalizedStarOf{#1}_n}}

\newcommand{\SSREquiv}{\mathrel{\equiv}}

\newcommand{\AtomToDNF}{\ensuremath{\mathcal{D}}}
\newcommand{\AtomToDNFOf}[1]{\ensuremath{\AtomToDNF(#1)}}

\newcommand{\SynthSymLens}{\PCF{SynthSymLens}\xspace}
\newcommand{\RXSearchState}{\ensuremath{\mathit{pq}}}
\newcommand{\ToStochastic}{\PCF{ToStochastic}\xspace}
\newcommand{\Best}{\ensuremath{\mathit{best}}}

\newcommand{\Continue}{\PCF{Continue}\xspace}
\newcommand{\PQ}{\PCF{PQ}\xspace}

\newcommand{\GreedySynth}{\PCF{GreedySynth}\xspace}

\newcommand{\Expand}{\PCF{Expand}\xspace}

\newcommand{\ReturnVal}[1]{\ensuremath{\Return\,#1}}

\newcommand{\SSOpt}{\ensuremath{\mathbf{SS}}}
\newcommand{\SSNCOpt}{\ensuremath{\mathbf{SSNC}}}
\newcommand{\BSOpt}{\ensuremath{\mathbf{BS}}}

\newcommand{\AnyOpt}{\ensuremath{\mathbf{Any}}}
\newcommand{\FLOpt}{\ensuremath{\mathbf{FL}}}
\newcommand{\CCOpt}{\ensuremath{\mathbf{DC}}}
\newcommand{\NSOpt}{\textbf{NS}\xspace}
\newcommand{\NROpt}{\textbf{NR}\xspace}

\newcommand{\None}{\ensuremath{\mathit{None}}\xspace}

\newcommand{\Some}{\ensuremath{\mathit{Some}}}
\newcommand{\SomeOf}[1]{\ensuremath{\Some\,#1}}
\newcommand{\Option}{\ensuremath{\mathit{Option}}}
\newcommand{\OptionOf}[1]{\ensuremath{#1 \App \Option}}

\newcommand{\CreateR}{\KW{createR}\xspace}
\newcommand{\CreateL}{\KW{createL}\xspace}
\newcommand{\PutR}{\KW{putR}\xspace}
\newcommand{\PutL}{\KW{putL}\xspace}
\newcommand{\CreateROf}[1]{\CreateR \App #1}
\newcommand{\CreateLOf}[1]{\CreateL \App #1}
\newcommand{\PutROf}[2]{\PutR \App #1 \App #2}
\newcommand{\PutLOf}[2]{\PutL \App #1 \App #2}
\newcommand{\PutRL}{\PCF{PutRL}\xspace}
\newcommand{\PutLR}{\PCF{PutLR}\xspace}
\newcommand{\CreatePutRL}{\PCF{CreatePutRL}}
\newcommand{\CreatePutLR}{\PCF{CreatePutLR}}
\newcommand{\ForgetfulLR}{\PCF{ForgetfulLR}\xspace}
\newcommand{\ForgetfulRL}{\PCF{ForgetfulRL}\xspace}

\newcommand{\SDNFREabbrev}{SDNF RE\xspace}
\newcommand{\SDNFREabbrevs}{SDNF REs\xspace}

\newcommand{\SDNFLens}{\ensuremath{\mathit{sdl}}}
\newcommand{\SDNFLensOf}[1]{\ensuremath{\DNFLensLeft#1\DNFLensRight}}
\newcommand{\SSQLensOf}[1]{\ensuremath{\SequenceLensLeft#1\SequenceLensRight}}
\newcommand{\SSQLens}{\ensuremath{\mathit{ssql}}}
\newcommand{\SAtomLens}{\ensuremath{\mathit{sal}}}

\newcommand{\Entropy}{\ensuremath{\mathbb{H}}}
\newcommand{\EntropyOf}[1]{\ensuremath{\Entropy(#1)}}
\newcommand{\LEntropyOf}[1]{\ensuremath{\Entropy^\leftarrow(#1)}}
\newcommand{\REntropyOf}[1]{\ensuremath{\Entropy^\rightarrow(#1)}}

\newcommand{\Max}{\ensuremath{\text{max}}}
\newcommand{\MaxOf}[1]{\ensuremath{\Max(#1)}}

\newcommand{\PRegexOr}[3]{\ensuremath{#1~|_{#3}~#2}}
\newcommand{\PRegexConcat}[2]{{\ensuremath{\RegexConcat{#1}{#2}}}}
\newcommand{\PRegexStar}[2]{\ensuremath{#1^{*_{#2}}}}
\newcommand{\Probability}{\ensuremath{p}\xspace}
\newcommand{\ProbabilityAlt}{\ensuremath{q}}
\newcommand{\ProbabilityDistOf}[1]{P_{#1}}
\newcommand{\ProbabilityOf}[2]{P_{#1}(#2)}

\newcommand{\Fst}{\ensuremath{\mathit{fst}}}
\newcommand{\Snd}{\ensuremath{\mathit{snd}}}

\newcommand{\InL}{\ensuremath{\mathit{inl}}}
\newcommand{\InLOf}[1]{\ensuremath{\InL \App #1}}
\newcommand{\InR}{\ensuremath{\mathit{inr}}}
\newcommand{\InROf}[1]{\ensuremath{\InR \App #1}}

\newcommand{\SingleApp}{\ensuremath{\mathit{singleapp}}}
\newcommand{\Log}{\ensuremath{\mathit{log}}}

\newcommand{\sep}{\ensuremath{\; | \;}}

\newcommand{\Reals}{\ensuremath{\mathbb{R}}}
\newcommand{\SLS}{\ensuremath{\mathit{sls}}}
\newcommand{\ALS}{\ensuremath{\mathit{als}}}
\newcommand{\CartesianMap}{\PCF{CartesianMap}\xspace}
\newcommand{\GreedySeqSynth}{\PCF{GreedySeqSynth}\xspace}
\newcommand{\AtomSynth}{\PCF{GreedyAtomSynth}\xspace}
\newcommand{\GreedyState}{\ensuremath{\mathit{pq}}}

\newcommand{\myoverline}[1]{
  \overline{\raisebox{0pt}[1.3ex][0pt]{\ensuremath{\mkern-1.5mu#1}}}
}
\newcommand{\BRegex}{\ensuremath{\myoverline{\RegexNS}}\xspace}
\newcommand{\BRegexAlt}{\ensuremath{\myoverline{\RegexAltNS}}\xspace}
\newcommand{\BRegexAltAlt}{\ensuremath{\myoverline{\RegexAltAltNS}}\xspace}

\newcommand*{\SavedLstInline}{}
\LetLtxMacro\SavedLstInline\lstinline
\DeclareRobustCommand*{\lstinline}{%
  \ifmmode
    \let\SavedBGroup\bgroup
    \def\bgroup{%
      \let\bgroup\SavedBGroup
      \hbox\bgroup
    }%
  \fi
  \SavedLstInline
}

\begin{document}








\title{Synthesizing Symmetric Lenses}

\begin{abstract}

{\em Lenses} are programs that can be run both ``front to back'' and ``back
to front,'' allowing updates to either their source or their target data to be
transferred in both directions. 
Since their introduction by Foster \ETAL, lenses have been extensively
studied, extended, and applied.  Recent work has also demonstrated how
techniques from {\em type-directed program synthesis} can be used to
efficiently synthesize a simple class of lenses---so-called {\em
  bijective lenses} over string data---given a pair of types (regular
expressions) and a small number of examples.

\hspace*{1em}We extend this synthesis algorithm to a much broader class of
lenses, called {\em simple symmetric lenses}, including all
bijective lenses, all of the popular category of ``asymmetric''
lenses, and a rich subset of the more powerful
``symmetric lenses'' proposed by Hofmann \ETAL  
Intuitively, simple symmetric lenses allow some information to be present
on one side but not the other and vice versa.
They are of independent theoretical interest, being
the largest class of symmetric lenses that do not rely on persistent
internal state.  

Synthesizing simple symmetric lenses is substantially more challenging than
synthesizing bijective lenses: Since some of the information on each side
can be ``disconnected'' from the other side, there will, in general, be {\em
  many} lenses that agree with a given example. To guide the search process,
we use {\em stochastic regular expressions} and ideas from information
theory to estimate the amount of information propagated by a candidate
lens, generally preferring lenses that propagate more information, 
as well as user annotations marking parts of the source and
target data structures as either {\em irrelevant} or {\em essential}.

We describe an implementation of simple symmetric lenses and our synthesis
procedure as extensions to the Boomerang language. We evaluate its performance
on 48 benchmark examples drawn from Flash Fill, Augeas, the bidirectional
programming literature, and electronic file format synchronization tasks. Our
implementation can synthesize each of these lenses in under 30 seconds.
\end{abstract}

%

%



\author{Anders Miltner}
\affiliation{
  \institution{Princeton University}
  \country{USA}
}
\email{amiltner@cs.princeton.edu}

\author{Solomon Maina}
\affiliation{
  \institution{University of Pennsylvania}
  \country{USA}
}
\email{smaina@seas.upenn.edu}

\author{Kathleen Fisher}
\affiliation{
  \institution{Tufts University}
  \country{USA}
}
\email{kfisher@eecs.tufts.edu}

\author{Benjamin C. Pierce}
\affiliation{
  \institution{University of Pennsylvania}
  \country{USA}
}
\email{bcpierce@cis.upenn.edu}

\author{David Walker}
\affiliation{
  \institution{Princeton University}
  \country{USA}
}
\email{dpw@cs.princeton.edu}

\author{Steve Zdancewic}
\affiliation{
  \institution{University of Pennsylvania}
  \country{USA}
}
\email{stevez@cis.upenn.edu}


\maketitle
\renewcommand{\shortauthors}{A. Miltner, S. Maina, K. Fisher, B. C. Pierce, D. Walker, S. Zdancewic}

\ifanon\else
\fi

\section{Introduction}

\noindent
In today's data-rich world, similar information is often stored in different
places and in different formats. For instance, electronic calendars come in
iCalendar, vCalendar, and h-event formats~\cite{calendar-formats}; U.S. taxation
data can be transmitted via Tax XML (used by the U.S. government for e-filing) or
TXF (used by TurboTax)~\cite{tax-formats}; configuration files can be stored in
ad-hoc or structured file formats~\cite{augeas} (for example, scheduled commands
can be stored in a crontab or as Windows's Scheduled Tasks xml).

One convenient way to maintain consistency between such varied data formats
is to use a \emph{lens}~\cite{Focal2005-long}---a bidirectional program that
can transform updates to data represented in a format $S$ to another format
$T$ and vice versa, providing round-tripping guarantees that help ensure
information is not lost or corrupted as it is transformed back and forth
between different representations.
Domain-specific languages for writing lenses can facilitate generation of
principled data transformations.  For example, Boomerang~\cite{boomerang} is a
language for expressing bidirectional string transformations.

\textit{Synthesizing} lenses can make the task even easier. Indeed, lens
languages like Boomerang are excellent targets for search-based synthesis,
since their 
specialized, sub-Turing-complete syntax and their expressive type systems
drastically limit the space of possible programs.
The recent Optician synthesizer for Boomerang~\cite{optician} has
demonstrated the feasibility of synthesizing quite complex examples, albeit
only for a restricted class of lenses---so-called bijective lenses.
Given source and target formats plus a small number of examples of
corresponding data, Optician's synthesis algorithm is guaranteed to find a
bijective lens matching the examples, if one exists, often requiring no
examples at all. While specifying regular expressions can be cumbersome, doing so
is much easier than programming lenses -- regular expressions are widely
understood whereas few programmers are fluent in lenses, and writing lenses
requires users to reason about how to align two formats at once.
One reason that bijective lens synthesis can be so effective, even relative to
successful synthesis tools in other domains, is that there are typically not
very many bijections between a given pair of data formats. 
If the synthesis algorithm finds any bijection, it is fairly likely to be
the intended one. 

However, the set of real-world use cases for bijective lenses, where two
data formats contain different arrangements of precisely the same
information, is  
limited.  Oftentimes, two data
formats share just {\em some} of their information content.  For instance,
one ad-hoc system-configuration
file might include some format-specific metadata, such as a date or a
reference number, while the same configuration file on another operating
system does not.  Indeed among the benchmark problems described in
Section \ref{sec:evaluation}, all of the benchmarks taken from Flash
Fill~\cite{flashfill} and many of the benchmarks that synchronize between
two ad hoc file formats have this characteristic.

Can Optician's basic synthesis procedure be extended to a richer class of
lenses? Could we even imagine synthesizing all \emph{symmetric
  lenses}~\cite{symmetric-lenses}---a much larger class that includes
bijective lenses, ``asymmetric'' lenses (where the transformation from a
source structure to a target can throw away information that must
then be restored when transferring from target to source), and even more flexible
transformations that allow each side to throw away information?


%

One might first hope that extending the algorithms from \citet{optician} to
synthesize symmetric rather than bijective lenses would be relatively
straightforward: Simply replace the bijective combinators with symmetric
ones and search using similar heuristics.  However, this na\"ive approach
encounters two difficulties.

The first of these is pragmatic.  Symmetric lenses as presented
by~\citet{symmetric-lenses} operate
over three structures: a ``left'' structure $X$, a ``right'' structure
$Y$ and a ``complement'' 
$C$ that contains the information not
present in either $X$ or $Y$. These complements must be stored
and managed somehow.  More fundamentally, complements complicate
synthesis {\em specifications}---instead of just giving single examples of
source and target 
pairings, users would have to give longer ``interaction sequences'' to 
show how a lens should behave. To avoid these complexities, we define a
restricted variant 
of symmetric lenses,  called \emph{simple symmetric lenses}. 
%
Intuitively, simple symmetric lenses are symmetric
lenses that do not require external ``memory'' to recover data from past
instances of $X$ or $Y$ when making a round trip. They only need the most
recent instance.

Formally, we characterize the expressiveness of simple symmetric lenses by
proving that they are exactly the symmetric lenses
that satisfy an intuitive property called \emph{forgetfulness}.  We also show
they are expressive enough for many practical uses by adding simple
symmetric lenses to 
the Boomerang language~\cite{boomerang} and applying them to a range of
real-world applications.  This exercise also demonstrates that simple
symmetric lenses can coexist with, and be extended by, other advanced lens features
provided by Boomerang, including {\em quotient
  lenses}~\cite{quotientlenses,maina+:quotient-synthesis} and {\em matching
  lenses}~\cite{matchinglenses}.

This leaves us with the second difficulty in synthesizing symmetric
lenses: Whereas the number of bijective lenses between two given formats is
typically tiny, the number of simple symmetric lenses is typically enormous.
If a na\"ive search algorithm just selects the first simple symmetric lens
it finds, the returned lens will generally \emph{not} be the one the user wanted.  We
need a new principle for identifying ``more likely'' lenses and a more
sophisticated synthesis algorithm that uses this principle to search the
space more intelligently. 

For these, we turn to information theory. We consider ``likely'' lenses to
be ones that propagate ``a lot'' of information 
from the left data format to the right and vice versa. Conversely, ``unlikely''
lenses are ones that require a large amount of additional information to recover
one of the formats given the other. By default, our synthesis algorithm prefers
the lenses that propagate more information.
This preference is formalized using \emph{stochastic regular expressions}
(SREs)~\cite{stoch-rnn,stoch-def}, which simultaneously define a set of
strings and a probability distribution over those strings.  Using this
probability distribution, we can calculate the likelihood of a given lens.

\begin{figure}
  \includegraphics[width=.63\textwidth]{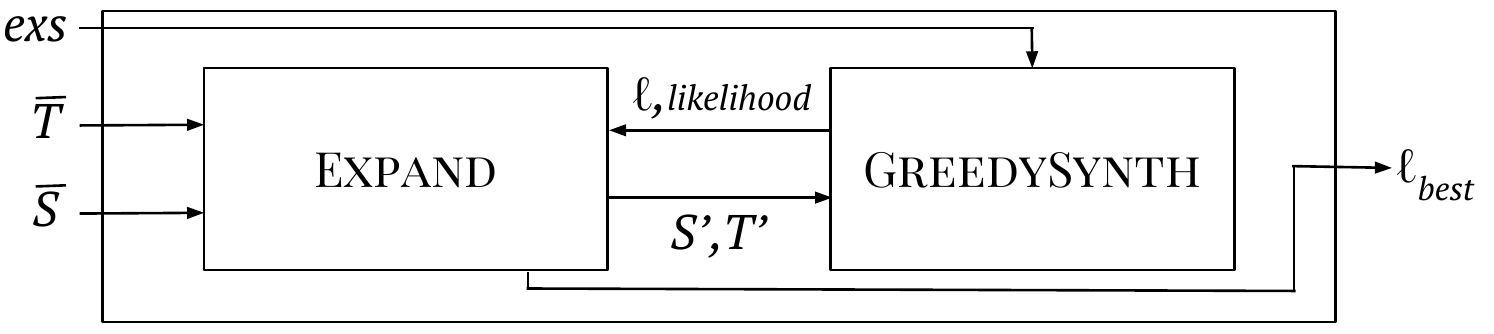}
  \vspace{-2ex}
  \caption{Schematic diagram for the simple symmetric lens synthesis algorithm.
    The user provides regular expressions \BRegex and \BRegexAlt and a set of
    examples \Examples{} as input. \Expand first converts \BRegex and \BRegexAlt
    to stochastic regular expressions \Regex and \RegexAlt with default
    probabilities. It then finds pairs of stochastic regular expressions
    equivalent to \Regex and \RegexAlt and iteratively proposes them to
    \GreedySynth. \GreedySynth finds a lens typed between the supplied SREs.
    When the algorithm finds a likely lens, it returns it.}
  \label{fig:high-level-algorithm}
\end{figure}

With simple symmetric lenses and this SRE-based likelihood measure in hand, we
propose a new algorithm for synthesizing likely lenses. At its core, the
algorithm performs a type-directed search between descriptions of the data
formats, measuring success using the likelihood measure.

Interesting complications arise from the need to deal with regular-expression
equivalences. There are infinitely many regular expressions equivalent to a
given one, and the lens returned by a type-directed search will in general
depend on which of the possible representations are chosen for its source and
target formats. Moreover, certain lenses may not be well-typed unless the
format representations are replaced by equivalent ones:
in general, the synthesis algorithm has to search through
equivalent regular expression types to find the most likely lens. To tame
this complexity, we divide the synthesis algorithm into two communicating search
procedures (Figure~\ref{fig:high-level-algorithm}), following~\citet{optician}.
The first, \Expand, uses rewriting rules to propose new pairs of
stochastic regular expressions equivalent to the original pair. The second,
\GreedySynth, uses a greedy, type-directed
algorithm to find a simple symmetric lens between input SRE pairs, returning the lens and
its likelihood score to \Expand. The whole synthesis algorithm heuristically
terminates when a sufficiently likely lens is found.



We added this synthesis procedure to the Boomerang system and explored its
effectiveness on a range of applications. Users set up a Boomerang synthesis
task by providing two regular expressions and optionally supplying input-output
examples. Users can also override the default mechanism for
calculating the information content of a SRE by
asserting that certain strings are \emph{essential} or \emph{irrelevant},
forcing certain data to either be retained or discarded during the
transformations.

We evaluate our implementation, the effects of optimizations, and our inclusion
of relevance annotations on a set of 48 lens synthesis benchmarks drawn from
data cleaning, view maintenance, and file synchronization tasks. We find the
system can synthesize simple symmetric lenses for all of the benchmarks in under
30 seconds (\S\ref{sec:evaluation}).

In summary, our contributions are as follows:
\begin{itemize}
\item We define {\em simple symmetric lenses}, a subclass of symmetric
lenses with no hidden state,
we propose a syntax for them---a collection
of combinators for building simple symmetric lenses---and we
describe how to calculate their likelihoods
(\S\ref{sec:overview}--\S\ref{sec:ssl}).  
\item We give an algorithm for synthesizing lenses built from these combinators,
  using a novel application of stochastic regular expressions to guide the
  search process (\S\ref{sec:synthesis}) and prove that our metrics correspond
  to the information loss of a given lens.
\item We extend the Boomerang implementation with simple symmetric lenses
and lens synthesis specifications, 
  including types, examples, and relevance annotations. We evaluate our
  implementation, the effects of
  optimizations, and our inclusion of relevance annotations on a set of
  48 lens synthesis benchmarks drawn from data cleaning, view maintenance, and
  file synchronization tasks. The system can synthesize simple symmetric lenses
  for all of the benchmarks in under 30 seconds (\S\ref{sec:evaluation}).
\item We prove that simple symmetric lenses fall between asymmetric
lenses~\cite{Focal2005-long} and full symmetric
lenses~\cite{symmetric-lenses} in expressiveness.  Indeed, the
class of simple lenses can be characterized semantically as the subset of
full symmetric lenses for which a simple ``forgetfulness'' property holds
(\S\ref{sec:related}).
\end{itemize}
Along the
way, we show that {\em star-semiring equivalences}~\cite{starsemiring} on regular
expressions 
remain valid when extended to stochastic regular expressions, preserving
probability distributions as well as languages (\S\ref{sec:sre}).
We close with related work (\S\ref{sec:related}) and
concluding thoughts (\S\ref{sec:conc}).
Elided proofs can be found
\ifappendices in the appendices\else in the full version of the paper~\cite{sslfull}\fi.


\noindent


\section{Overview}
\label{sec:overview}

\begin{figure}
  \setbox0=\hbox{%
    \begin{minipage}{3.5in}
\begin{lstlisting}[
stringstyle=\upshape\sffamily
]
Jane Doe: 38000
John Public: 37500
\end{lstlisting}
    \end{minipage}
  }
  \savestack{\listingA}{\box0}

  \setbox0=\hbox{%
    \begin{minipage}{3.5in}
\begin{lstlisting}[
stringstyle=\upshape\sffamily
]
FirstLast,Company
Jane Doe,Healthcare Inc.
John Public,Insurance Co.
\end{lstlisting}
    \end{minipage}
  }
  \savestack{\listingB}{\box0}

  \setbox0=\hbox{%
    \begin{minipage}{3.5in}
\begin{lstlisting}
let salary = number | "unk"
let emp_salary = name . " " . name . ": " salary
let emp_salaries = "" | emp_salary . ("\n" emp_salary)$^*$
\end{lstlisting}
    \end{minipage}
  }
  \savestack{\listingC}{\box0}

  \setbox0=\hbox{%
    \begin{minipage}{3.5in}
\begin{lstlisting}
let company = (co_name . ("Co." | "Inc." | "Ltd.")) | "UNK"
let emp_ins = name . " " . name "," company
let header = "FirstLast,Company"
let emp_insurance = header . ("\n" . emp_ins)$^*$
\end{lstlisting}
    \end{minipage}
  }
  \savestack{\listingD}{\box0}
  
  \centering
  \resizebox{\columnwidth}{!}{%
    \begin{tabular}{|>{\columncolor{vlightyellow}}c|>{\columncolor{vlightblue}}c|}
      \hline
      Management's data & HR's data \\
      \hline
      \listingA & \listingB \\
      \hline
      \cellcolor{lightyellow}Management's type & 
      \cellcolor{lightblue} HR's type\\
      \hline
      \cellcolor{lightyellow}\listingC & \cellcolor{lightblue}\listingD \\
      \hline
    \end{tabular}
  }
  \caption{Hypothetical example data files and corresponding regular expressions
  used by management and HR at a company to represent employee salaries and health
  insurance providers, respectively.
  }
  \label{fig:minimized-representations}
\end{figure}

We begin with an overview of simple symmetric lenses and our
synthesis algorithm, using a simplified example drawn from a hypothetical
company. In this company, management and human resources (HR) store
information about employees in separate text files: management stores the
names of 
employees and their salaries while HR stores the names of employees and their
health insurance providers. Figure~\ref{fig:minimized-representations} gives
examples of the two file formats and regular expressions describing them. 

The company uses a simple symmetric lens to keep these files
synchronized. When management adds a new employee, say ``Chris Roe: 32500'',
this lens 
adds the corresponding entry ``Chris Roe, UNK'' to HR's file. The default value
``UNK'' represents the fact that the employee's insurance company is currently
unknown. A similar update happens if HR adds a new employee before
management knows about them,
in which case the sentinel value ``unk'' reflects unknown salary information
on management's side.
Furthermore, if HR corrects an error in an employee's name, say changing ``John
Public'' to ``Jon Public'', the lens mirrors this change into management's
file. Not all the data is mirrored, however: management's file is not updated
in response to insurance changes, and HR's is oblivious to salary
changes. Simple symmetric lenses are appropriate for keeping these two files
synchronized because they contain a mix of shared and unshared information.

%

\subsection{Simple Symmetric Lenses}
Semantically, a simple symmetric lens (we will just say ``lens'' when clear
from context) between 
sets $X$ and $Y$ comprises four functions subject to four round-tripping
laws.
\vspace*{.7ex}
\begin{trivlist}
\item 
\begin{center}
\begin{minipage}[c]{.3\textwidth}
\begin{gather}
 \mbox{$\CreateR \OfType X \rightarrow Y$} \nonumber \\
 \mbox{$\CreateL \OfType Y \rightarrow X$} \nonumber \\
 \mbox{$\PutR \OfType X \rightarrow Y \rightarrow Y$} \nonumber \\
 \mbox{$\PutL \OfType Y \rightarrow X \rightarrow X$} \nonumber
\end{gather} 
\end{minipage} 
\hspace{.4in}
\begin{minipage}[c]{.5\textwidth}
\begin{align}
  \tag{\CreatePutRL}
  \PutLOf{(\CreateROf{x})}{x} = x\\
  \tag{\CreatePutLR}
  \PutROf{(\CreateLOf{y})}{y} = y\\
  \tag{\PutRL}
  \PutLOf{(\PutROf{x}{y})}{x} = x\\
  \tag{\PutLR}
  \PutROf{(\PutLOf{y}{x})}{y} = y
\end{align}
\end{minipage} 
\end{center}
\end{trivlist}
The two \KW{create} functions are used to fill 
in default values when introducing new data (\EG, on \KW{create} the ``unk''
salary entry is added alongside the name to the management file when HR inserts
a new employee). The two \KW{put} functions propagate edits from one format to
the other by combining a new value from one with an old value from the other.
The record projection notation $\Lens.\PutR$ extracts the \PutR function
from the 
lens $\Lens$. These four functions can be used to keep the
common information between two file types in sync. For example, if a new file
of the left-hand format is created, the \CreateR function will build a new
file in the 
right-hand format. If the right-hand file is then edited, the \PutL function can
update the left-hand file with the changed information.

The round-tripping laws guarantee that pushing an ``unedited'' value from one
side (the result of a put or create) back through a lens in the other direction
will get back to exactly where we began. For example, consider a lens
synchronizing the employee data formats in
Figure~\ref{fig:minimized-representations}. Applying the \CreateR function to a
salary file creates an insurance file with the same employee names in the file,
with
\lstinline{"UNK"} for each insurance company. The round-tripping
laws guarantees that applying the \PutL function to the generated insurance file
and the original salary file will return the original salary file, unmodified.

Simple symmetric lenses differ from ''classical'' symmetric lenses in that
they do not involve a \emph{complement}. We give a detailed comparison in
\S\ref{sec:relationship}.


\subsection{Simple Symmetric Lens Combinators}

In principle, programmers can define lenses just by writing
four arbitrary functions by hand and showing they satisfy the round-tripping
laws, but this can be tedious and error prone.  A better idea is to provide
a set of combinators---a domain-specific language---for building complex,
law-abiding 
lenses from simpler ones. Such bidirectional languages have been given for
many different kinds of data, including database
relations~\cite{BohannonPierceVaughan} algebraic data
types~\cite{symmetric-lenses}, graphs~\cite{tgg}, and more.
We focus here on combinators for defining lenses between string data formats
described by regular expressions.

If $\BRegex$ and $\BRegexAlt$ are regular expressions, then
\lstinline{$\Lens$ : $\BRegex \Leftrightarrow \BRegexAlt$} indicates that
$\Lens$ is a simple symmetric lens between $\LanguageOf{\BRegex}$ and
$\LanguageOf{\BRegexAlt}$. (We will use undecorated variables later for
stochastic regular expressions, so we mark plain REs with overbars.)
We illustrate some of these
combinators by defining lenses on subcomponents of the employee data formats.

The simplest combinator is the identity lens \IdentityLens, which takes as an
argument a regular expression $\BRegex$ and propagates data unchanged
in both directions.
\begin{lstlisting}
id(name)  :  name $\Leftrightarrow$ name
\end{lstlisting}
The identity lens moves data back and forth from source to target without
changing it. Both the \CreateR and \CreateL functions are the identity function
($\CreateR \App s = s$), and the put functions merely return the first argument
($\PutROf{\String}{\StringAlt} = \String$).

In contrast, $\DisconnectOf{\Regex}{\RegexAlt}{\String}{\StringAlt}$ does not propagate any data at all
from one format to the other.  The \Disconnect lens takes four arguments: two
regular expressions ($\Regex, \RegexAlt$) and two strings ($\String, \StringAlt$). The regular expressions specify the formats
on the two sides, while the strings provide default values. 
\begin{lstlisting}
disconnect(salary, "", "unk", "")  :  salary $\Leftrightarrow$ ""
\end{lstlisting}
On creates, the input values are thrown away, and default values are
returned ($\CreateLOf{\StringAlt} =
\lstinline{"unk"}$), and on puts, the second argument is used and the first is
thrown away ($\PutROf{\String}{\StringAlt} = \StringAlt$). For 
example, the \lstinline{salary} field is only present in management files,
so the disconnect lens can ensure salary edits do not cause updates to the
HR file.
With the above lens, $\PutL \App \lstinline{""} \App 60000$ will return 60000,
and \PutR will always return \lstinline{""}.

The insert lens \lstinline{ins} and the delete lens \lstinline{del} are
syntactic sugar for uses of the disconnect lens in which a string
constant is omitted entirely from the source or target format.
\begin{lstlisting}
 ins($t$) = disconnect("", $t$, "", $t$)
 del($s$) = disconnect($s$, "", $s$, "")
\end{lstlisting}
The \lstinline{ins} lens inserts a constant string when going from
left to right, while \lstinline{del} inserts a string when going from right
to left.

Finally, there are a number of combinators that construct
more complex lenses from simpler ones. These include concatenation
($\ConcatLensOf{\Lens_1}{\Lens_2}$, often written $\Lens_1 . \Lens_2$), variation
($\OrLensOf{\Lens_1}{\Lens_2}$), and iteration ($\IterateLensOf{\Lens}$).
For example, we could use the lens
\begin{lstlisting}
id(name) . id(" ") . id(name) . del(": ") . ins(",")
\end{lstlisting}
to transform ``\lstinline{Jane Doe: }'' to ``\lstinline{Jane Doe,}'' in
the left-to-right direction.

%
%
The \IterateLens lens is useful for synchronizing a series of items or rows in a
table. For example given a
lens \lstinline{employee_lens} that synchronizes data for a single employee, the
lens
\begin{lstlisting}
iterate(id("\n") . employee_lens)  :  ("\n" . emp_salary)$^*$ $\Leftrightarrow$ ("\n" . emp_ins)$^*$
\end{lstlisting}
transforms a list of employees in employees in the Management format to a list
of employees in the HR format and vice versa. These combinators
are combined in figure~\ref{fig:example_lens} to construct a complete lens
between the employee formats.

\begin{figure}
\begin{lstlisting}
let name_lens = id(name) . id(" ") . id(name) . del(": ") . ins(",")
let employee_lens = name_lens . disconnect(salary,"","unk","") 
                                                    . disconnect("",company,"","UNK")
let employees_lens = ins("\n") . employee_lens . iterate(id("\n") . employee_lens)
let full_lens  :  emp_salaries $\Leftrightarrow$ emp_insurance = ins(header) . employees_lens
\end{lstlisting}
  \caption{A lens that synchronizes management and HR employee files}
  \label{fig:example_lens}
\end{figure}


Lenses are typed by pairs of regular expressions and only a handful of
rules result in a given type.
Many of these rules provide tight relationships between terms and types, for example, the type of a
\ConcatLens lens is the concatenation of the regular expressions of the sub-lenses.
\[
  \inferrule*
  {
    \Lens_1 \OfType \BRegex_1 \Leftrightarrow \BRegexAlt_1\\
    \Lens_2 \OfType \BRegex_2 \Leftrightarrow \BRegexAlt_2
  }
  {
    \ConcatLensOf{\Lens_1}{\Lens_2} \OfType \BRegex_1 \Concat \BRegex_2
    \Leftrightarrow
    \BRegexAlt_1 \Concat \BRegexAlt_2
  }
\]
The typing  rules maintain structural similarity
between regular expression types and lens terms, with one exception: The
type equivalence rule permits a well-typed lens to be well-typed among any
equivalent regular expression pair.
\[
  \centering
  \inferrule*
  {
    \Lens \OfType \BRegex \Leftrightarrow \BRegexAlt\\
    \BRegex \SSREquiv \BRegex'\\
    \BRegexAlt \SSREquiv \BRegexAlt'
  }
  {
    \Lens \OfType \BRegex' \Leftrightarrow \BRegexAlt'
  }
\]
This rule is needed for lenses like the final one in
Figure~\ref{fig:example_lens} to be well-typed.


\subsection{Ranking Simple Symmetric Lenses}

For larger and more complex formats, it can become quite difficult to
write lenses by hand, as we did above.  Instead, we would like to synthesize
them from types and examples.  
That is, given a pair of regular expressions and a set of input-output
examples, we want 
to find a simple symmetric lens that is typed by the regular expression pair and
satisfies all the input/output examples. We call this a \emph{satisfying
  lens}. For our running example, suppose we wish to synthesize a lens between
\lstinline{emp_salaries} and \lstinline{emp_insurance}, using as sole
input-output example the data in Figure~\ref{fig:minimized-representations}.

One challenge is that the simple symmetric lens combinators permit many
well-typed lenses
between a given pair of regular expressions. For example, 
Figure~\ref{fig:example_lens} gives one possible lens between regular
expressions  \lstinline{emp_salaries} and \lstinline{emp_insurance},
but 
\begin{lstlisting}
disconnect(emp_salaries, emp_insurance, $s$, $t$)
\end{lstlisting}
is another well-typed lens that satisfies the example in
Figure~\ref{fig:minimized-representations} (where $s$ and $t$ represent
the provided example strings). In general, {\em many} examples
may be 
required to rule out all possible occurrences of \Disconnect lenses,
particularly in complex formats.
Instead of merely finding
\emph{any} satisfying lens, we wish to synthesize a satisfying lens that is
likely to please the user.

How can we identify such a ``likely'' lens? We propose the following heuristic:
A satisfying lens is ``more likely'' if it uses more data from one format to
construct the other.  For example, the
identity lens (which uses all the data) is more likely than the disconnect
lens (which uses none). Formally, we define the {\em likelihood} of a
satisfying lens as the 
expected number of bits required to recover data in one format from 
data in the other; higher likelihoods correspond to fewer bits. Two
strings $s$ and $t$ are \emph{synchronized according to lens
  $\Lens$} if $\Lens.\PutR\
s\ t = t$ and $\Lens.\PutL\ t\ s = s$. We can \emph{recover} $s$ from 
$t$ using bits $b$ and lens $\Lens$ if we can reconstruct $s$ from $t$, $b$, and
$\Lens$. For example, given the \IdentityLens{} lens, we can recover $s$ from $t$
using no bits because, in this case, $s$ is just $t$. In contrast, given the
\Disconnect{} lens, we need enough bits to fully encode $s$ in order to recover
it from $t$ because all the information in $t$ gets thrown away. Thus, if both
\IdentityLens{} and \Disconnect{} are satisfying lenses for a particular pair of
formats, \IdentityLens{} will be more likely.

The expected number of bits required to recover a piece of data corresponds to
the well-known information-theoretic concept of
\emph{entropy}~\cite{Shannon1948}. Calculating entropy requires a probability
distribution over the space of possible values for the data. Specifically, given
a set $S$ and a probability distribution $P \OfType S \rightarrow \Reals$ over
$S$, the entropy $\EntropyOf{S,P}$ is $-\Sigma_{s\in S}P(s)\Log_2P(s)$. In
information theory, the \emph{information content} of each element of $s \in
S$ is 
the number of bits required to specify $s$ in a perfect encoding scheme
($-\Log_2P(s)$). The entropy of $S$ is the expected number of bits
required to specify an element drawn from $S$---the
probability of each element times its information content. Entropy captures
the intuition that, if a data source contains many possible elements and
none have 
significantly higher probability than others, it will have high entropy. Data
sources with just a few high probability elements have low
entropy. When $P$ is clear from context, we will
often use the shorthand of $\EntropyOf{S}$ for $\EntropyOf{S,P}$.

In the present setting, we already have a way of expressing sets of data:
regular 
expressions. To calculate entropy, what we need is a way to express probability
distributions over those sets. To that end, we adopt \emph{stochastic regular
  expressions}~\cite{stoch-def,stoch-rnn} (SREs), which are regular expressions
in which each operator is annotated with a probability (see \S\ref{sec:sre}). A
stochastic regular expression thus specifies both a set of strings at the
same time as a probability distribution over that set.

We use entropy to gauge the relative likelihood of lenses in our synthesis
algorithm (see \S\ref{subsec:lenscosts}). For any lens $\Lens$, we can calculate
the expected number of bits required to recover a string $t$ in
$\LanguageOf{\RegexAlt}$ from a synchronized string $s$ in
$\LanguageOf{\Regex}$. This expectation is the \emph{conditional entropy of
  $\Regex$ given $\RegexAlt$ and $\Lens$}, formally $\sum_{s \in
  \Regex}\ProbabilityOf{\Regex}{s}\cdot\EntropyOf{\SetOf{t \SuchThat
    \Lens.\PutROf{s}{t} = t}}$. The likelihood we assign to $\Lens$ is the sum
of the conditional entropy of $\Regex$ given $\RegexAlt$ and $\Lens$ and the
conditional entropy of $\RegexAlt$ given $\Regex$ and $\Lens$. This metric
assigns higher entropy (or lower likelihood) to lenses where knowing the string on one side 
provides little information about the string on the other side. It assigns zero
entropy to bijections because given a string $s \in \Regex$, the bijection
exactly determines the corresponding string in $\RegexAlt$.


To obtain SREs from the plain regular expressions that users write, 
we use a default heuristic that attempts to assign probability annotations
giving
each string in the language equal probability (not prioritizing one piece of
information over another).

Sometimes users already know that certain data should or should not be
used to construct the other format. We introduce relevance annotations to SREs
that enable users to specify whether a piece of data should be used to construct
the other format (with \SRequire) or not (with \Skip). In our example,
\lstinline{salary} and \lstinline{company} could safely be skipped (as they are
disconnected), where \lstinline{name} in \lstinline{emp_salary} and
\lstinline{emp_ins} could be annotated as required (as they are converted with the
identity lens). Doing so both constrains the problem (as \lstinline{name}s must
be synchronized) and makes it easier (as the algorithm does not waste time
looking for \lstinline{salary} information in the insurance format). In this
way, users can tweak the
lens likelihoods with their external knowledge.

\subsection{Searching for Likely Lenses}
Now that we have a way to calculate likelihood, we need a way to search for
likely lenses. Previous work~\cite{augustsson-2004, osera+:pldi15,
  feser-pldi-2015, frankle+:popl16} has shown that types can be used to
dramatically restrict the search space and improve the effectiveness of
synthesis. In our setting, types are pairs of stochastic regular
expressions, and the fact that each regular expression is semantically
equivalent to infinitely many others complicates the problem.

Following existing work on bidirectional program
synthesis~\cite{maina+:quotient-synthesis}, we 
split our algorithm into two communicating
search procedures. The first,
\Expand, navigates the space of semantically equivalent regular expressions
by applying rewrite rules that preserve both semantics and
probability distributions. This
algorithm ranks pairs of stochastic regular expressions by the number of
rewrite rule applications required to obtain each pair from the one given as
input. It passes the pairs off to the second 
search procedure, \GreedySynth, in rank order with the smallest first.

\GreedySynth looks for highest-likelihood lenses between a given
pair of stochastic regular expressions $\Regex$ and $\RegexAlt$ by performing a
type-directed search. It first converts the stochastic
regular expressions provided by \Expand into \emph{stochastic DNF regular
  expressions}---a constrained representation of stochastic regular expressions
with disjunctions distributed over concatenations and with concatenations and
disjunctions normalized to operating over lists.
Then it uses the syntax of these n-ary DNF regular expressions to find
normalized lenses in a form we call \emph{simple symmetric n-ary DNF lenses}.
These involve neither a composition operator nor a type equivalence rule.
These restrictions mean that there are comparatively few simple symmetric n-ary DNF
lenses that are well typed between a given pair of stochastic n-ary DNF regular
expressions, so \GreedySynth's search space for a given pair of regular
expressions is finite. Finally, \GreedySynth yields a simple symmetric lens by
converting the n-ary syntax back to the binary forms provided in the surface
language.

This architecture of communicating synthesizers gives us a way to enumerate
pairs of stochastic regular expressions of increasing rank and to efficiently
search through them, but it poses a problem: when should \Expand stop proposing
new SRE pairs? We might have found a promising lens between a pair of stochastic
regular expressions, but a different pair we haven't yet discovered may give
rise to an even better lens. The search algorithm must resolve a tension between
the quality of the inferred lens and the amount of time it takes to return a
result. For example, if the algorithm has already found the lens in
Figure~\ref{fig:example_lens}, we don't want to spend a lot of time searching
for an even better lens. To resolve this tension, the algorithm uses heuristics
to judge whether to return the current best satisfying lens to the user or to
pass the next set of equivalent SREs to \GreedySynth. The heuristics favor stopping
if the current best satisfying lens is very likely, indicating the lens is very
promising (for example, if the satisfying lens loses no information, the
algorithm should terminate for no other lens will be more likely). The
heuristics also favor stopping if \Expand has delivered to \GreedySynth all
pairs of stochastic regular expressions at a given rank and there is a large
number of pairs at the next rank, because searching through all such pairs will
take a long time (\IE, if a satisfying lens loses only a little
information
and searching for a better one will take a long time, the discovered lens is
returned).
With this approach, the algorithm can quickly return a
satisfying lens with relatively high likelihood. If the user is unhappy with the
result, they can either refine the search by supplying additional examples,
which serve both to rule out previously proposed lenses and to reduce the
size of the search space by cutting down on the number of satisfying lenses, or
they can supply annotations on the source and target SREs indicating that
certain information either must be retained or must be discarded by the lens
(see \S\ref{subsec:relevanceannotations}).

\section{Stochastic Regular Expressions}
\label{sec:sre}
To characterize likely lenses, we must compute the expected number of bits
needed to recover a string in one data source from a synchronized string in the
other data source. To do this, we first develop a probabilistic 
model for our language using \emph{stochastic regular
  expressions} (SREs)---regular expressions annotated with probability
information~\cite{stoch-rnn,stoch-def}---that
jointly express a language and a probability distribution over that language. 
$$\Regex{},\RegexAlt{} \quad\GEq{}\quad  \String  \quad|\quad \emptyset \quad|\quad \PRegexStar{\Regex}{p} \quad|\quad \RegexConcat{\Regex_1}{\Regex_2} \quad|\quad \PRegexOr{\Regex_1}{\Regex_2}{p}$$
(Lowercase $\String$ ranges over constant strings and $\Probability$ ranges
over real numbers between 0 and 1, exclusive.)
The semantics of an SRE $S$ is a probability distribution $P_S$ defined as
follows.
\begin{center}
  \begin{tabular}{rcl}
    $\ProbabilityOf{\String}{\String''}$
    & =
    & $\begin{cases*}1 & if $\String = \String''$\\ 0 & otherwise\end{cases*}$ \\
     
    $\ProbabilityOf{\emptyset}{\String}$
    & =
    & $0$ \\
    
    $\ProbabilityOf{\RegexConcat{\Regex_1}{\Regex_2}}{s}$

    & =
    & $\Sigma_{\String = \String_1\String_2}\ProbabilityOf{\Regex_1}{\String_1}\ProbabilityOf{\Regex_2}{\String_2}$ \\
    
    $\ProbabilityOf{\PRegexOr{\Regex_1}{\Regex_2}{\Probability}}{\String}$
    & =
    & $\Probability\ProbabilityOf{\Regex_1}{\String} +
      (1-\Probability)\ProbabilityOf{\Regex_2}{\String}$\\
    
    $\ProbabilityOf{\PRegexStar{\Regex}{\Probability}}{\String}$
    & =
    & $\Sigma_n \Sigma_{\String = \String_1 \ldots \String_n}\Probability^n(1-\Probability)\Pi_{i=1}^n\ProbabilityOf{\Regex}{\String_i}$\\
  \end{tabular}
\end{center}
One may think of SREs as string generators. Under this interpretation, the
constant SRE $\String$ always generates the string $\String$ and never any other
string. The SRE $\PRegexOr{\Regex_1}{\Regex_2}{\Probability}$ generates a string
from $\Regex_1$ with probability $\Probability$ and generates a string from
$\Regex_2$ with probability $1-\Probability$. The SRE
$\PRegexStar{\Regex}{\Probability}$ generates strings in $\Regex^n$ with
probability $\Probability^n(1-\Probability)$. For example,
$\PRegexStar{\Regex}{\Probability}$ will generate a string in $\Regex$ with
probability $\Probability(1-\Probability)$ and a string in
$\RegexConcat{\Regex}{\Regex}$ with probability $\Probability^2(1-\Probability)$.

\begin{figure}
  \begin{tabular}{@{}r@{\hspace{1em}}c@{\hspace{1em}}l@{}r@{\hspace*{3em}}r@{\hspace{1em}}c@{\hspace{1em}}l@{}r}
    $\RegexConcat{\emptyset}{\Regex}$ & $\SSREquiv$ & $\emptyset$ & \EmptyProjectionLeftRule{} &
    $\RegexConcat{\Regex}{\emptyset}$ & $\SSREquiv$ & $\emptyset$ & \EmptyProjectionRightRule{} \\
    \RegexConcat{\EmptyString{}}{\Regex{}} & $\SSREquiv$ & \Regex{} & \ConcatIdentityLeftRule{} &
    \RegexConcat{\Regex{}}{\EmptyString{}} & $\SSREquiv$ & \Regex{} & \ConcatIdentityRightRule{} \\
    \PRegexOr{\Regex}{\emptyset}{1} & $\SSREquiv$ & \Regex{} & \OrIdentityRule{} &
    \PRegexOr{\Regex{}}{\RegexAlt{}}{\Probability} & $\SSREquiv$ & \PRegexOr{\RegexAlt{}}{\Regex{}}{1-\Probability} & \OrCommutativityRule{}\\
    \RegexConcat{\Regex{}}{(\PRegexOr{\Regex{}'}{\Regex{}''}{\Probability})} & $\SSREquiv$ & \PRegexOr{(\RegexConcat{\Regex{}}{\Regex{}'})}{(\RegexConcat{\Regex{}}{\Regex{}''})}{\Probability} & \DistributivityLeftRule{} &
    \RegexConcat{(\PRegexOr{\Regex{}'}{\Regex{}''}{\Probability})}{\Regex{}} & $\SSREquiv$ & \PRegexOr{(\RegexConcat{\Regex{}'}{\Regex{}})}{(\RegexConcat{\Regex{}''}{\Regex{}})}{\Probability} & \DistributivityRightRule{} \\
    \PRegexStar{\Regex{}}{\Probability} & $\SSREquiv$ & \PRegexOr{\EmptyString{}}{(\RegexConcat{\Regex{}}{\PRegexStar{\Regex{}}{\Probability}})}{1-\Probability} & \UnrollstarLeftRule{} &
    \PRegexStar{\Regex{}}{\Probability} & $\SSREquiv$ & \PRegexOr{\EmptyString{}}{(\RegexConcat{\PRegexStar{{\Regex{}}}{\Probability}}{\Regex{}})}{1-\Probability} & \UnrollstarRightRule{} 
  \end{tabular}
  \begin{tabular}{@{}r@{\hspace{1em}}c@{\hspace{1em}}l@{\hspace{1em}}r@{}}
    \RegexConcat{(\RegexConcat{\Regex{}}{\Regex'})}{\Regex''} & $\SSREquiv$ & \RegexConcat{\Regex{}}{(\RegexConcat{\Regex'}{\Regex''})} & \ConcatAssocRule{}  \\
    \PRegexOr{(\PRegexOr{\Regex}{\Regex'}{\Probability_1})}{\Regex''}{\Probability_2} & $\SSREquiv$ & \PRegexOr{\Regex}{(\PRegexOr{\Regex'}{\Regex''}{\frac{(1-\Probability_1)\Probability_2}{1 - \Probability_1\Probability_2}})}{\Probability_1\Probability_2} & \OrAssociativityRule{}  \\
  \end{tabular}
  \caption{Stochastic Regular Expression Star-Semiring Equivalence}
  \label{fig:ssr-sre}
\end{figure}

\subsection{Stochastic Regular Expression Equivalences}
The $\Expand$ algorithm enumerates SREs that are
``equivalent'' to a given one. However, existing work does not define any notion
of stochastic regular expression equivalence. Figure~\ref{fig:ssr-sre} shows how
we extend the \emph{star-semiring}~\cite{starsemiring} equivalences (a
\emph{finer} notion of equivalence than semantic equivalence) to
SREs.  

\begin{theorem}
  If $\Regex \SSREquiv \RegexAlt$, then
  $\ProbabilityDistOf{\Regex} = \ProbabilityDistOf{\RegexAlt}$.
\end{theorem}

\ifappendices(For the proof, see the appendix 
Theorem~\ref{thm:semantics-correctness}.) \else\fi  This theorem will come in handy as we traverse
(with \Expand) or normalize across (as part of \GreedySynth) star-semiring
equivalences.

\subsection{Stochastic Regular Expression Entropy}
The entropy of a data source $S$ is the expected number of bits required to
describe an element drawn from $S$ (formally $-\Sigma_{s\in S}\,P(s)\cdot
\Log_2P(s)$). The entropy of a SRE can be computed directly from its syntax,
when each string is uniquely parsable (\IE{} the SRE is unambiguous) and contains
no empty subcomponents.
\begin{center}
  \begin{tabular}{rcl}
    $\EntropyOf{s}$
    & =
    & 0\\
    
    $\EntropyOf{\PRegexStar{\Regex}{\Probability}}$
    & =
    & $\frac{\Probability}{1-\Probability}(\EntropyOf{\Regex} - \Log_2\Probability)
      - \Log_2(1-\Probability)$\\
    
    $\EntropyOf{\PRegexConcat{\Regex}{\RegexAlt}}$
    & =
    & \EntropyOf{\Regex} + \EntropyOf{\RegexAlt}\\
    
    $\EntropyOf{\PRegexOr{\Regex}{\RegexAlt}{\Probability}}$
    & =
    & $\Probability(\EntropyOf{\Regex} - \Log_2(\Probability)) + (1-\Probability)(\EntropyOf{\RegexAlt} - \Log_2(1-\Probability))$\\
  \end{tabular}
\end{center}

For example, the entropy of $\Regex =
\PRegexOr{\lstinline{"a"}}{\lstinline{"b"}}{.5}$ is 1. The best encoding of a
stream of elements from $\Regex$ will use, on average, 1 bit per element to
determine whether that element is \lstinline{"a"} or \lstinline{"b"}. As an
additional example, a fixed string has no information content, and so has no
entropy.

\begin{theorem}
  \label{thm:correct_entropy}
  If $\Regex$ is unambiguous and does not contain $\emptyset$ as a subterm,
  $\EntropyOf{\Regex}$ is the entropy of $\Regex$.
\end{theorem}

To understand the difficulties caused by ambiguity, consider the SRE
$\PRegexOr{\lstinline{"a"}}{\lstinline{"a"}}{.5}$. The formula above defines the
entropy to be 1, but the true entropy is 0 (\IE, no bits are needed to know the
generated string will be \lstinline{"a"}). Similar issues occur when trying to
find the entropy when $\emptyset$ is a subterm (what should the entropy of
$\PRegexStar{\emptyset}{.5}$ be?), so we do not define entropy on~$\emptyset$.

Fortunately, we already require unambiguous regular expressions as input to our
synthesis procedure to guarantee the synthesized lenses are well-typed, and we
can easily preprocess empty subexpressions out of SREs that are themselves
nonempty using the star-semiring equivalences (\EG,
$\PRegexOr{\emptyset}{\lstinline{"s"}}{.5} \SSREquiv \lstinline{"s"}$).



\section{Simple Symmetric String Lenses}
\label{sec:ssl}
Simple symmetric string lenses are bidirectional programs that convert string
data between formats. These lenses operate over regular expressions, not SREs.
If $\Regex$ is a stochastic regular expression, $\BRegex$ is the regular
expression obtained by removing probability annotations.

\begin{align*}
\ell ::= &\IdentityLensOf{\BRegex} \; | \; \DisconnectOf{\BRegex}{\BRegexAlt}{\String}{\StringAlt} \; | \; \IterateLensOf{\Lens} \; | \; \ConcatLensOf{\Lens_1}{\Lens_2} \; | \; \SwapLensOf{\Lens_1}{\Lens_2} \; | \; \OrLensOf{\Lens_1}{\Lens_2} \; | \; \ComposeLensOf{\Lens_1}{\Lens_2} \; |\\
&\MergeROf{\Lens_1}{\Lens_2} \; | \; \MergeLOf{\Lens_1}{\Lens_2} \; | \; \InvertOf{\Lens}
\end{align*}

We have already described a number of these combinators in \S\ref{sec:overview}.
Briefly (full definitions are below), the lens $\Lens_1 ; \Lens_2$
composes two lenses sequentially, it 
transforms from one data format to the other by first transforming to an
intermediate format, shared as the target of $\Lens_1$ and the source of
$\Lens_2$. The \MergeR lens combines two sublenses that operate on disjoint
source formats and share the same target.
The lens that gets used depends on what data is in the left-hand format, if the data is
in the domain of $\Lens_1$ then $\Lens_1$ gets used, and similarly for
$\Lens_2$.  The \MergeL lens is symmetric to \MergeR, the two sublenses must
share the same source. Finally, $\InvertOf{\Lens}$ inverts a lens, \CreateR
becomes \CreateL and \PutR becomes \PutL, and vice-versa.

We now give typing rules and semantics for (some of) these combinators. We write
$\Lens \OfType \BRegex \Leftrightarrow \BRegexAlt$ to indicate that $\Lens$ is a
well typed lens between the language of \BRegex and the language of \BRegexAlt.
We present just \CreateR and \PutR in cases where \CreateL and \PutL are
symmetric; full details can be found in\ifappendices{} appendix
\S\ref{sec:appendixlenses} \else{} the full version of the paper\fi. For
simplicity of presentation, we use unambiguous regular expressions everywhere,
even though the identity and disconnect lenses can actually be defined over
ambiguous regular expressions; relaxing this restriction is not problematic.

\begin{centermath}
  \begin{tabular}{m{1cm}@{\hspace*{6em}}m{1cm}@{\hspace*{3em}}}
    $
  \inferrule*
  {
  }
  {
    \IdentityLensOf{\BRegex} \OfType \BRegex \Leftrightarrow \BRegex
  }$
&
$\begin{array}[b]{r@{\ }c@{\ }l}
    \CreateR{} \App s & = & s\\
    \PutR{} \App s \App t & = & s
  \end{array}$
\end{tabular}
\end{centermath}
Note that the identity lens ignores the second argument in the put functions.
Because the two formats are fully synchronized, no knowledge of the prior data
is needed.

\begin{centermath}
  \begin{tabular}{m{5cm}m{2cm}}
    $
    \inferrule*
    {
    \String \in \LanguageOf{\BRegex}\\
    \StringAlt \in \LanguageOf{\BRegexAlt}
    }
    {
    \DisconnectOf{\BRegex}{\BRegexAlt}{\String}{\StringAlt}
    \OfType \BRegex \Leftrightarrow \BRegexAlt
    }$
    &
      $
  \begin{array}{@{}r@{\ }c@{\ }l@{}}
    $\CreateR{} \App \String'$ & = & $\StringAlt$\\
    $\PutR{} \App \String' \App \StringAlt'$ & = & $\StringAlt'$
  \end{array}$
\end{tabular}
\end{centermath}

Just as the identity lens ignores the second argument in puts, disconnect lenses
ignore the first argument in both puts and creates.  The data is unsynchronized in these
two formats, information from one format doesn't impact the other.

\begin{centermath}
\begin{tabular}[b]{l@{\qquad \; \qquad}l}
$
  \inferrule*
  {
    \Lens_1 \OfType \BRegex_1 \Leftrightarrow \BRegexAlt_1\\
    \Lens_2 \OfType \BRegex_2 \Leftrightarrow \BRegexAlt_2
  }
  {
    \ConcatLensOf{\Lens_1}{\Lens_2} \OfType \BRegex_1 \Concat \BRegex_2
    \Leftrightarrow
    \BRegexAlt_1 \Concat \BRegexAlt_2
  }
$
&
$
  \inferrule*
  {
    \Lens_1 \OfType \BRegex_1 \Leftrightarrow \BRegexAlt_1\\
    \Lens_2 \OfType \BRegex_2 \Leftrightarrow \BRegexAlt_2
  }
  {
    \SwapLensOf{\Lens_1}{\Lens_2} \OfType \BRegex_1 \Concat \BRegex_2
    \Leftrightarrow
    \BRegexAlt_2 \Concat \BRegexAlt_1
  }
$
\end{tabular}
\end{centermath}

Concat is similar to concatenation in existing string lens languages like
Boomerang.  For such terms, we do not provide the semantics, and merely refer readers to existing work. The swap combinator is similar to concat, though the second regular expression
is swapped.
\begin{centermath}
  \begin{tabular}{m{5cm}m{7cm}}
$
  \inferrule*
  {
    \Lens_1 \OfType \BRegex_1 \Leftrightarrow \BRegexAlt_1\\
    \Lens_2 \OfType \BRegex_2 \Leftrightarrow \BRegexAlt_2
  }
  {
    \OrLensOf{\Lens_1}{\Lens_2} \OfType
    \RegexOr{\BRegex_1}{\BRegex_2}
    \Leftrightarrow
    \RegexOr{\BRegexAlt_1}{\BRegexAlt_2}
  }
  $
&
  \begin{tabular}{@{}r@{\ }c@{\ }l@{}}
    $\CreateR{} \App \String$
    & =
    & $\begin{cases*}
      \Lens_1.\CreateROf{\String} & if $\String\in\LanguageOf{\BRegex_1}$\\
      \Lens_2.\CreateROf{\String} & if $\String\in\LanguageOf{\BRegex_2}$
      \end{cases*}$\\
    
    $\PutR{} \App \String \App \StringAlt$
    & =
    & $\begin{cases*}
        \Lens_1.\PutROf{\String}{\StringAlt} & if $\String\in\LanguageOf{\BRegex_1} \BooleanAnd \StringAlt\in\LanguageOf{\BRegexAlt_1}$\\
        \Lens_2.\PutROf{\String}{\StringAlt} & if $\String\in\LanguageOf{\BRegex_2} \BooleanAnd \StringAlt\in\LanguageOf{\BRegexAlt_2}$\\
        \Lens_1.\CreateROf{\String} & if $\String\in\LanguageOf{\BRegex_1} \BooleanAnd \StringAlt\in\LanguageOf{\BRegexAlt_2}$\\
        \Lens_2.\CreateROf{\String} & if $\String\in\LanguageOf{\BRegex_2} \BooleanAnd \StringAlt\in\LanguageOf{\BRegexAlt_1}$
      \end{cases*}$\\
  \end{tabular}
\end{tabular}
\end{centermath}
The \OrLens lens deals with data that can come in one form or another. If the
data gets changed from one format to the other, information in the old format is
lost. This differs from the \OrLens lens of classical symmetric
lenses, as those can retain such information in the complement (see
\S\ref{sec:related}).
%
%
\begin{centermath}
  \begin{tabular}{m{5cm}m{6cm}}
$
  \centering
  \inferrule*
  {
    \Lens_1 \OfType \BRegex_1 \Leftrightarrow \BRegexAlt\\
    \Lens_2 \OfType \BRegex_2 \Leftrightarrow \BRegexAlt
  }
  {
    \MergeROf{\Lens_1}{\Lens_2} \OfType
    \RegexOr{\BRegex_1}{\BRegex_2}
    \Leftrightarrow
    \BRegexAlt
  }
$
&
  \begin{tabular}{@{}r@{\ }c@{\ }l@{}}
    $\CreateR{} \App \String$
    & =
    & $\begin{cases*}
      \Lens_1.\CreateROf{\String} & if $\String\in\LanguageOf{\BRegex_1}$\\
      \Lens_2.\CreateROf{\String} & if $\String\in\LanguageOf{\BRegex_2}$
      \end{cases*}$\\
    
    $\CreateL{} \App \StringAlt$
    & =
    & $\Lens_1.\CreateLOf{\StringAlt}$\\
    
    $\PutR{} \App \String \App \StringAlt$
    & =
    & $\begin{cases*}
      \Lens_1.\PutROf{\String}{\StringAlt} & if $\String\in\LanguageOf{\BRegex_1}$\\
      \Lens_2.\PutROf{\String}{\StringAlt} & if $\String\in\LanguageOf{\BRegex_2}$
    \end{cases*}$\\
    
    $\PutL{} \App \StringAlt \App \String$
    & =
    & $\begin{cases*}
        \Lens_1.\PutLOf{\StringAlt}{\String} & if $\String\in\LanguageOf{\BRegex_1}$\\
        \Lens_2.\PutLOf{\StringAlt}{\String} & if $\String\in\LanguageOf{\BRegex_2}$
      \end{cases*}$\\
  \end{tabular}
\end{tabular}
\end{centermath}

The \MergeR lens is interesting because it merges data where one piece of data can be in
two formats, and one data has only one format. In previous
work~\cite{boomerang}, this was combined into \OrLens{}, where
\OrLens{} could have ambiguous types, but we find it more clear to have explicit
merge operators: it is easier to see what lens the synthesis algorithm is
creating.

\begin{centermath}
  \inferrule*
  {
    \Lens_1 \OfType \BRegex \Leftrightarrow \BRegexAlt_1\\
    \Lens_2 \OfType \BRegex \Leftrightarrow \BRegexAlt_2
  }
  {
    \MergeLOf{\Lens_1}{\Lens_2} \OfType
    \BRegex
    \Leftrightarrow
    \RegexOr{\BRegexAlt_1}{\BRegexAlt_2}
  }
\end{centermath}
The \MergeL lens is symmetric to \MergeR.

\begin{centermath}
  \begin{tabular}{m{5cm}m{6cm}}
  $
  \inferrule*
  {
    \Lens_1 \OfType \BRegex \Leftrightarrow \BRegexAlt\\
    \Lens_2 \OfType \BRegexAlt \Leftrightarrow \BRegexAltAlt\\
  }
  {
    \ComposeLensOf{\Lens_1}{\Lens_2} \OfType
    \BRegex \Leftrightarrow \BRegexAltAlt
  }
$
&
  \begin{tabular}{@{}r@{\ }c@{\ }l@{}}
    $\CreateR{} \App \String$ & = & $\Lens_2.\CreateROf{(\Lens_1.\CreateROf{\String})}$\\
    $\CreateL{} \App \StringAlt$ & = & $\Lens_1.\CreateLOf{(\Lens_2.\CreateLOf{\StringAlt})}$\\
    $\PutR{} \App \String \App \StringAltAlt$ & = & $\Lens_2.\PutROf{(\Lens_1.\PutROf{\String}{(\Lens_2.\CreateLOf{\StringAltAlt})})}{\StringAltAlt}$\\
    $\PutL{} \App \StringAltAlt \App \String$ & = & $\Lens_1.\PutLOf{(\Lens_2.\PutLOf{\StringAltAlt}{(\Lens_2.\CreateROf{\String})})}{\String}$
  \end{tabular}
\end{tabular}
\end{centermath}
Composing is interesting in the put functions. Because puts require intermediary
data, we recreate that intermediary data with creates.

\begin{centermath}
\begin{tabular}[b]{l@{\qquad}l}
$
\inferrule*
  {
    \Lens \OfType \BRegex \Leftrightarrow \BRegexAlt
  }
  {
    \IterateLensOf{\Lens} \OfType
    \StarOf{\BRegex}
    \Leftrightarrow
    \StarOf{\BRegexAlt}
  }
  $
  &
  $
  \inferrule*
  {
    \Lens \OfType \BRegex \Leftrightarrow \BRegexAlt
  }
  {
    \InvertOf{\Lens} \OfType \BRegexAlt \Leftrightarrow \BRegex
  }
  $
\end{tabular}
\end{centermath}
The \IterateLens lens deals with iterated data, while inverting reverses the direction of a lens: creating on the right becomes creating on the left and vice versa, and putting on the right becomes putting on the left and vice versa. The invert combinator is particularly useful when chaining many compositions together, as it can be used to align the central types.
\[
  \centering
  \inferrule*
  {
    \Lens \OfType \BRegex \Leftrightarrow \BRegexAlt\\
    \BRegex \SSREquiv \BRegex'\\
    \BRegexAlt \SSREquiv \BRegexAlt'
  }
  {
    \Lens \OfType \BRegex' \Leftrightarrow \BRegexAlt'
  }
\]

Type equivalence enables a lens of type $S \Leftrightarrow T$ to be used as a
lens of type $S' \Leftrightarrow T'$ if $S$ is equivalent to $S'$ and $T$ is
equivalent to $T'$. Type equivalence is useful both for addressing type
annotations, and for making well-typed compositions.

\subsection{Lens Likelihoods}
\label{subsec:lenscosts}
Our likelihood metric is based on the expected amount of information required to
recover a string in one data format from the other. We use the function
$\REntropyOf{\RegexAlt \Given \Lens, \Regex}$ to calculate bounds on the
expected amount of information required to recover a string in $\RegexAlt$ from
a string in $\Regex$, synchronized by $\Lens$. Similarly, we use the function
$\LEntropyOf{\Regex \Given \Lens, \RegexAlt}$ to calculate bounds on the
expected amount of information required to recover a string in $\Regex$ from a
string in $\RegexAlt$, synchronized by $\Lens$. We write $a[b,c]$ to
mean $[ab,ac]$, and $[a,b]+[c,d]$ to mean $[a+c,b+d]$.
\begin{center}
  \begin{tabular}{rcl}
    $\REntropyOf{\RegexAlt \Given \IdentityLensOf{\RegexAlt}, \RegexAlt}$
    & =
    & [0,0]\\
    
    $\REntropyOf{\RegexAlt \Given \DisconnectOf{\Regex}{\RegexAlt}{\String}{\StringAlt}, \Regex}$
    & =
    & [\EntropyOf{\RegexAlt},\EntropyOf{\RegexAlt}]\\

    $\REntropyOf{\PRegexStar{\RegexAlt}{\ProbabilityAlt} \Given \IterateLensOf{\Lens}, \PRegexStar{\Regex}{\Probability}}$
    & =
    & $\frac{\Probability}{1-\Probability}\REntropyOf{\RegexAlt \Given \Lens, \Regex}$\\
    
    $\REntropyOf{\RegexAlt_1 \Concat \RegexAlt_2 \Given \ConcatLensOf{\Lens_1}{\Lens_2}, \Regex_1 \Concat \Regex_2}$
    & =
    & $\REntropyOf{\RegexAlt_1 \Given \Lens_1, \Regex_1} + \REntropyOf{\RegexAlt_2 \Given \Lens_2, \Regex_2}$\\
    
    $\REntropyOf{\RegexAlt_2 \Concat \RegexAlt_1 \Given \SwapLensOf{\Lens_1}{\Lens_2}, \Regex_1 \Concat \Regex_2}$
    & =
    & $\REntropyOf{\RegexAlt_2 \Given \Lens_1, \Regex_2} + \REntropyOf{\RegexAlt_1 \Given \Lens_1, \Regex_1}$\\
    
    $\REntropyOf{\PRegexOr{\RegexAlt_1}{\RegexAlt_2}{\ProbabilityAlt} \Given \OrLensOf{\Lens_1}{\Lens_2}, \PRegexOr{\Regex_1}{\Regex_2}{\Probability}}$
    & =
    & $\Probability\REntropyOf{\Regex_1 \Given \Lens_1, \RegexAlt_1} + (1-\Probability)\REntropyOf{\Regex_2 \Given \Lens_2, \RegexAlt_2}$\\
    
    $\REntropyOf{\RegexAlt \Given \MergeROf{\Lens_1}{\Lens_2}, \PRegexOr{\Regex_1}{\Regex_2}{\Probability}}$
    & =
    & $\Probability\REntropyOf{\RegexAlt \Given \Lens_1, \Regex_1} + (1-\Probability)\REntropyOf{\RegexAlt \Given \Lens_2, \Regex_2}$\\
    
    $\REntropyOf{\PRegexOr{\RegexAlt_1}{\RegexAlt_2}{\ProbabilityAlt} \Given \MergeLOf{\Lens_1}{\Lens_2}, \Regex}$
    & =
    & $[0,\REntropyOf{\RegexAlt_1 \Given \Lens_1, \Regex}+\REntropyOf{\RegexAlt_2 \Given \Lens_2, \Regex}+1]$\\
    
    $\REntropyOf{\Regex \Given \InvertOf{\Lens}, \RegexAlt}$
    & =
    & $\LEntropyOf{\Regex \Given \Lens, \RegexAlt}$\\
  \end{tabular}
\end{center}
$\LEntropyOf{\Regex \Given \Lens, \RegexAlt}$ is defined symmetrically. These
functions bound the expected number of bits to recover one data format from a
synchronized string in the other format. Note that we would be able to exactly
calculate the conditional entropy, were it not for \MergeL and \MergeR. If
$\MergeLOf{\Lens_1}{\Lens_2} \OfType \Regex \Leftrightarrow
\PRegexOr{\RegexAlt_1}{\RegexAlt_2}{\ProbabilityAlt}$, given a string in $s$, we
need to determine if the synchronized string is in $\RegexAlt_1$ or $\RegexAlt_2$.
However, this information content is dependent on how likely the
synchronized string is to be in $\RegexAlt_1$ or $\RegexAlt_2$. Nevertheless, we
typically calculate the conditional entropy exactly, as merges are relatively
uncommon in practice; only 2 of the lenses synthesized in our benchmarks
include merges.

The likelihood of a lens is the negative of its \emph{cost}. The cost of a lens
between two SREs $cost(\Lens, \Regex, \RegexAlt) = \MaxOf{\LEntropyOf{\Regex
    \Given \Lens, \RegexAlt}} + \MaxOf{\REntropyOf{\RegexAlt \Given \Lens,
    \Regex}}$ is the sum of the maximum expected number of bits to recover the
left format from the right, and the right from the left. We have proven theorems
demonstrating the calculated entropy corresponds to the actual conditional
entropy to recover the data.

\begin{theorem}
  Let $\Lens \OfType \Regex \Leftrightarrow \RegexAlt$, where $\Lens$ does not
  include composition, $\Regex$ and $\RegexAlt$ are unambiguous, and neither
  $\Regex$ nor $\RegexAlt$ contain any empty subcomponents.
  \begin{enumerate}
  \item $\REntropyOf{\RegexAlt \Given \Lens, \Regex}$ bounds the entropy of
    $\SetOf{t \SuchThat t \in \LanguageOf{\RegexAlt}}$, given $\SetOf{s
      \SuchThat s \in \LanguageOf{\Regex} \BooleanAnd \Lens.\PutROf{s}{t} = t}$
  \item $\LEntropyOf{\Regex \Given \Lens, \RegexAlt}$ bounds the entropy of
    $\SetOf{s \SuchThat s \in \LanguageOf{\Regex}}$, given $\SetOf{t \SuchThat t
      \in \LanguageOf{\RegexAlt} \BooleanAnd \Lens.\PutLOf{t}{s} = s}$
  \end{enumerate}
\end{theorem}


Note that our definition of $\Entropy^\rightarrow$ contains no case for sequential
composition $\Lens_1 ; \Lens_2$ and our theorem excludes lenses that contain
such compositions.  Defining the entropy of lenses involving
composition is challenging because $\Lens_1$ might, for instance, add some
information that is subsequently projected away in $\Lens_2$. Such operations
can cancel, leaving a zero-entropy bijection composed from two non-zero entropy
transformations. However, detecting such cancellations directly is complicated and
this property is difficult to determine merely from syntax. Fortunately, we are able to
sidestep such considerations by synthesizing \emph{DNF lenses}---simple symmetric
lenses that inhabit a disjunctive normal form that does not include
composition.


\section{Synthesis}
\label{sec:synthesis}
Algorithm~\ref{alg:synth-sym-lens} presents our synthesis algorithm at a high
level of abstraction. This algorithm searches for likely lenses in priority
order one ``class'' at a time using a \GreedySynth subroutine. Each class is the
set of lenses that can by typed by a pair of regular expressions, modulo a set
of simple axioms such as associativity, commutativity, and distributivity. The
\Expand subroutine generates new classes using the star-unrolling axioms
(\UnrollstarLeftRule{} and \UnrollstarRightRule).

To summarize, the input regular
expressions are first converted into stochastic regular expressions with
\ToStochastic. This pair of SREs is used to initialize a priority queue ($pq$).
The priority of a SRE pair is the number of rewrites needed to derive the pair
from the originals. Next, $\PCF{SynthSymLens}$ enters a loop that searches for
likely lenses. The loop terminates when the algorithm believes it is unlikely to
find a better lens than the best one it has found so far (a termination
condition defined by $\PCF{\Continue}$). Within each iteration of the loop, it:
\begin{itemize}
\item pops the next class ($S$, $T$) of lenses to
search off of the priority queue ($\PCF{\PQ.Pop}$),
\item executes $\GreedySynth$ to find a best lens in that class if
one exists ($\Lens$),
using the examples $\Examples$ to filter out potential lenses that do not satisfy
the specification,
\item replaces $best$ with $\Lens$, if $\Lens$ is more likely according
to our information-theoretic metric, and
\item adds the SREs derived from rewriting $S$ and $T$ ($\Expand(S,T)$)
  to the priority queue.
\end{itemize}
When the loop terminates, the search returns the globally best lens found
($best$).  Each subroutine of this algorithm will be explained in
further depth in the following subsections.

\begin{algorithm}[t]
  \caption{\SynthSymLens}
  \label{alg:synth-sym-lens}
  \begin{algorithmic}[1]
    \Function{SynthSymLens}{$\BRegex,\BRegexAlt,\Examples$}
    \State $\Regex \gets \Call{\ToStochastic}{\BRegex}$
    \State $\RegexAlt \gets \Call{\ToStochastic}{\BRegexAlt}$
    \State $\RXSearchState \gets \Call{\PQ.Create}{\Regex,\RegexAlt}$
    \State $\Best \gets \None$
    \While{$\Call{\Continue}{\RXSearchState,\Best}$}
    \State $(\Regex,\RegexAlt) \gets \Call{\PQ.Pop}{\RXSearchState}$
    \State $\Lens \gets \Call{\GreedySynth}{\Examples,\Regex,\RegexAlt}$
    \If{$\Call{Cost}{\Lens} < \Call{Cost}{\Best}$}
    \State $\Best \gets \Lens$
    \EndIf
    \State $\Call{\PQ.Push}{\RXSearchState,\Expand(\Regex,\RegexAlt)}$
    \EndWhile
    \State \ReturnVal{best}
    \EndFunction
  \end{algorithmic}
\end{algorithm}

\subsection{Searching for $(\Regex,\RegexAlt)$ Candidate Classes}
The first phase of the synthesis algorithm looks for pairs of SREs
$(\Regex,\RegexAlt)$ to drive the \GreedySynth algorithm.  These pairs are
generated using the star unrolling axioms:
\begin{center}
  \begin{tabular}{rcl}
    $\PRegexStar{\Regex}{\Probability}$
    & \Rewrite
    & \PRegexOr{\EmptyString{}}{(\RegexConcat{\Regex{}}{\PRegexStar{\Regex{}}{\Probability}})}{1-\Probability}\\

    $\PRegexStar{\Regex}{\Probability}$
    & \Rewrite
    & \PRegexOr{\EmptyString{}}{(\RegexConcat{\PRegexStar{\Regex{}}{\Probability}}{\Regex{}})}{1-\Probability}
  \end{tabular}
\end{center}
as well as the congruence rules that allow these rewrites to be applied on
subexpressions.
The priority queue yields stochastic regular expressions generated using fewer
rewrites first. Only when there are no more proposed regular expressions derived
from $n$ rewrites will \PCF{\PQ.Pop} propose regular expressions derived from
$n+1$ rewrites.

The procedure \PCF{\Continue} terminates the loop based on the how long the
search has been going, and how hard it expects the next class of problems to be.
In particular, if $\PCF{PQ.Peek}(pq) = (\Regex,\RegexAlt)$, that RE pair is at
distance $d$, the number of pairs in $pq$ at distance $d$ is $n$, and the
current best lens has cost $c$, then $\Call{Continue}{pq}$ continues the loop
while $c < d + \Log_2(n)$. This termination condition is based around two
primary principles: do not search overly deep ($d$), and do not tackle a
frontier that would take too long to process ($\Log_2(n)$). The log of the
frontier is included because the frontier grows much faster than lens costs
grow. Lens cost grows with the log of the expected number of choices in a given
lens (due to its information theoretic basis), so the frontier calculation does
too.

When $\Call{Continue}{pq}$ terminates the loop, the algorithm stops proposing
regular expression pairs, and instead returns to the user the best lens found
thus far. If the algorithm finds a bijective lens, which
has zero cost, it will immediately return.
\subsection{Stochastic DNF Regular Expressions}
\label{subsec:sdnfre}
The \GreedySynth subroutine converts input SREs into a temporary pseudo-normal
form, \emph{Stochastic DNF regular expressions} (\SDNFREabbrev{}s).
\SDNFREabbrev{}s normalize across many of the star-semiring equivalences -- if
two regular expressions are equivalent modulo differences in associativity,
commutativity, or distributivity, their corresponding SREs are syntactically
equal.

Syntactically, stochastic DNF regular expressions (\DNFRegex, \DNFRegexAlt) are
lists of stochastic sequences. Stochastic sequences (\Sequence, \SequenceAlt)
themselves are lists of interleaved strings and stochastic atoms. Stochastic
atoms (\Atom,\AtomAlt) are iterated stochastic DNF regular expressions.

\begin{center}
  \begin{tabular}{l@{\ }c@{\ }l@{\ }>{\itshape\/}r}
    \Atom{},\AtomAlt{} & \GEq{} & \PRegexStar{\DNFRegex{}}{\Probability}
\\
    \Sequence{},\SequenceAlt{} & \GEq{} &
                                                       $\SequenceOf{\String_0\SeqSep\Atom_1\SeqSep\ldots\SeqSep\Atom_n\SeqSep\String_n}$ 
\\
    \DNFRegex{},\DNFRegexAlt{} & \GEq{} & $\DNFOf{(\Sequence_1,\Probability_1)\DNFSep\ldots\DNFSep(\Sequence_n,\Probability_n)}$ 
  \end{tabular}
\end{center}

Intuitively, stochastic DNF regular expressions are stochastic
regular expressions with all concatenations fully distributed over all
disjunctions. As such, the language of a stochastic DNF regular expression is a
union of its subcomponents, the language of a stochastic sequence is the
concatenation of its subcomponents, and the language of a stochastic atom is the
iteration of its subcomponent. For
$\DNFOf{(\Sequence_1,\Probability_1)\DNFSep\ldots\DNFSep(\Sequence_n,\Probability_n)}$
to be a valid stochastic DNF regular expression, the probabilities must sum to
one ($\sum_{i=0}^n\Probability_i = 1$).

\begin{trivlist}
  \centering
\item 
  \begin{tabular}{@{\ }v@{\ }q}
    \LanguageOf{\PRegexStar{\DNFRegex{}}{\Probability}} \ =\  &
                                            \{\String_1\Concat\ldots\Concat\String_n
                                            \SuchThat \forall i,  \String_i\in\LanguageOf{\DNFRegex}\}\\
    \LanguageOf{\SequenceOf{\String_0\SeqSep\Atom_1\SeqSep\ldots\SeqSep\Atom_n\SeqSep\String_n}}\ =\  & 
    \{\String_0\Concat\StringAlt_1\cdots\StringAlt_n\Concat\String_n \SuchThat \StringAlt_i\in\LanguageOf{\Atom_i}\}
    \\
    \LanguageOf{\DNFOf{(\Sequence_1,\Probability_1)\DNFSep\ldots\DNFSep(\Sequence_n,\Probability_n)}}\ =\  &
    \{\String \SuchThat \String \in \LanguageOf{\Sequence_i} \text{\ and $i\in\RangeIncInc{1}{n}$}\}
  \end{tabular}
\end{trivlist}

As these DNF regular expressions are \emph{stochastic}, they are annotated with
probabilities to express a probability distribution, in addition to a
language.
\begin{center}
  \begin{tabular}{rcl}
    $\ProbabilityOf{\PRegexStar{\DNFRegex}{\Probability}}{\String}$
    & =
    & $\Sigma_n \Sigma_{\String = \String_1 \ldots \String_n}\Probability^n(1-\Probability)\Pi_{i=1}^n\ProbabilityOf{\DNFRegex}{\String_i}$\\
    
    $\ProbabilityOf{\SequenceOf{\String_0 \SeqSep \Atom_1 \SeqSep \ldots \SeqSep \Atom_n \SeqSep \String_n}}{\String'}$
    & =
    & $\Sigma_{\String' = \String_0\String_1'\ldots\String_n'\String_n}\Pi_{i = 1}^n\ProbabilityOf{\Atom_i}{\String_i'}$ \\
    
    $\ProbabilityOf{\DNFOf{(\Sequence_1,\Probability_1) \DNFSep \ldots \DNFSep (\Sequence_n,\Probability_n)}}{\String}$
    & =
    & $\Sigma_{i=1}^n\Probability_i\ProbabilityOf{\Sequence_i}{\String}$\\
  \end{tabular}
\end{center}

\begin{figure}
  \raggedright
  $\ConcatSequence{} \OfType{} \ArrowTypeOf{\SequenceType{}}{\ArrowTypeOf{\SequenceType{}}{\SequenceType{}}}$\\
  $\ConcatSequenceOf{[\String_0\SeqSep\Atom_1\SeqSep\ldots\SeqSep\Atom_n\SeqSep\String_n]}{[\StringAlt_0\SeqSep\AtomAlt_1\SeqSep\ldots\SeqSep\AtomAlt_m\SeqSep\StringAlt_m]}=
  [\String_0\SeqSep\Atom_1\SeqSep\ldots\SeqSep\Atom_n\SeqSep\String_n\Concat\StringAlt_0\SeqSep\AtomAlt_1\SeqSep\ldots\SeqSep\AtomAlt_m\SeqSep\StringAlt_m]$\\

  \medskip
  
  $\ConcatDNF{} \OfType{} \ArrowTypeOf{\DNFRegexType{}}{\ArrowTypeOf{\DNFRegexType{}}{\DNFRegexType{}}}$\\
  $\ConcatDNFOf{\DNFOf{(\Sequence_1,\Probability_1)\DNFSep\ldots\DNFSep(\Sequence_n,\Probability_n)}}{\DNFOf{(\SequenceAlt_1,\ProbabilityAlt_1)\DNFSep\ldots\DNFSep(\SequenceAlt_m,\ProbabilityAlt_m)}}=$\\
      $\DNFLeft (\ConcatSequenceOf{\Sequence_1}{\SequenceAlt_1},\Probability_1\ProbabilityAlt_1)\DNFSep \ldots
      \DNFSep
      (\ConcatSequenceOf{\Sequence_1}{\SequenceAlt_m},\Probability_1\ProbabilityAlt_m)\DNFSep
      \ldots$\\
      $\DNFSep
      (\ConcatSequenceOf{\Sequence_n}{\SequenceAlt_1},\Probability_n\ProbabilityAlt_1)\DNFSep
      \ldots \DNFSep
      (\ConcatSequenceOf{\Sequence_n}{\SequenceAlt_m},\Probability_n\ProbabilityAlt_m) \DNFRight$
  
  \medskip
  
  $\OrDNF{}_{\Probability} \OfType{}
  \ArrowTypeOf{\DNFRegexType{}}{\ArrowTypeOf{\DNFRegexType{}}{\DNFRegexType{}}
  }$ \\
  $\OrDNFOf{\DNFOf{(\Sequence_1,\Probability_1)\DNFSep\ldots\DNFSep(\Sequence_n,\Probability_n)}}{\DNFOf{(\SequenceAlt_1,\ProbabilityAlt_1)\DNFSep\ldots\DNFSep(\SequenceAlt_m,\ProbabilityAlt_m)}}{\Probability} =$\\
  $\DNFOf{(\Sequence_1,\Probability_1\Probability)\DNFSep\ldots\DNFSep(\Sequence_n,\Probability_n\Probability)\DNFSep(\SequenceAlt_1,\ProbabilityAlt_1(1-\Probability))\DNFSep\ldots\DNFSep(\SequenceAlt_m,\ProbabilityAlt_m(1-\Probability))}$
  
  \medskip
  
  \AtomToDNF{} \OfType
  \ArrowTypeOf{\AtomType{}}{\DNFRegexType{}}\\
  $\AtomToDNFOf{\Atom} = \DNFOf{(\SequenceOf{\EmptyString \SeqSep \Atom \SeqSep
      \EmptyString},1)}$
  \caption{Stochastic DNF Regular Expression Functions}
  \label{fig:dnf-regex-functions}
\end{figure}

The algorithm for converting a stochastic regular expressions $\Regex$ into its
corresponding \SDNFREabbrev form, written
$\ToDNFRegexOf{\Regex}$, is defined below. This conversion relies on operators
defined in Figure~\ref{fig:dnf-regex-functions}.
\[
  \begin{array}{rcl@{\hspace*{3em}}rcl}
    \ToDNFRegexOf{\String} 
    & = 
    & \DNFOf{(\SequenceOf{\String},1)}

    & \ToDNFRegexOf{(\RegexConcat{\Regex_1}{\Regex_2})} 
    & = 
    & \ToDNFRegexOf{\Regex_1} \ConcatDNF \ToDNFRegexOf{\Regex_2} \\

    \ToDNFRegexOf{\emptyset} 
    & = 
    & \DNFOf{}

    & \ToDNFRegexOf{(\PRegexOr{\Regex_1}{\Regex_2}{\Probability})} 
    & = 
    & \ToDNFRegexOf{\Regex_1} \OrDNF_{\Probability} \ToDNFRegexOf{\Regex_2} \\

    \ToDNFRegexOf{(\PRegexStar{\Regex}{\Probability})} & = & \AtomToDNFOf{\PRegexStar{(\ToDNFRegexOf{\Regex})}{\Probability}}\\
    
  \end{array}
\]
After this syntactic conversion has taken place, the sequences are ordered
(normalizing commutativity differences). This conversion respects languages and
probability distributions.

\begin{theorem}
  $\ProbabilityOf{\Regex}{\String} = \ProbabilityOf{\ToDNFRegex
    \Regex}{\String}$ and $\LanguageOf{\Regex} =
  \LanguageOf{\ToDNFRegexOf{\Regex}}$.
\end{theorem}

\paragraph*{Entropy}
We have developed a syntactic means for finding the entropy of a stochastic DNF
regular expression, like we have for stochastic regular expressions. This
enables us to efficiently find the entropy without first converting a \SDNFREabbrev to
a stochastic regular expression.

\begin{center}
  \begin{tabular}{rcl}
    $\EntropyOf{\PRegexStar{\DNFRegex}{\Probability}}$
    & =
    & $\frac{\Probability}{1-\Probability}(\EntropyOf{\DNFRegex} - \Log_2\Probability)
      - \Log_2(1-\Probability)
      $\\
    
    $\EntropyOf{\SequenceOf{\String_0 \SeqSep \Atom_1 \SeqSep \ldots \SeqSep \Atom_n \SeqSep \String_n}}$
    & =
    & $\Sigma_{i = 1}^n\EntropyOf{\Atom_i}$ \\
    
    $\EntropyOf{\DNFOf{(\Sequence_1,\Probability_1) \DNFSep \ldots \DNFSep (\Sequence_n,\Probability_n)}}$
    & =
    & $\Sigma_{i=1}^n\Probability_i(\EntropyOf{\Sequence_i}+\Log_2\Probability_i)$\\
  \end{tabular}
\end{center}

\begin{theorem}
  $\EntropyOf{\DNFRegex}$ is the entropy of $P_{\DNFRegex}$.
\end{theorem}

\paragraph*{\ToStochastic} With stochastic DNF regular expressions and
\ToDNFRegex defined, it is easier to explain \ToStochastic, the function that
converts regular expressions into stochastic regular expressions. If $\Regex$ is
a stochastic regular expression generated by \ToStochastic, then when put into
\SDNFREabbrev form, $\ToDNFRegexOf{\Regex} = \DNFOf{(\Sequence_1,\frac{1}{n})
  \DNFSep \ldots \DNFSep (\Sequence_n,\frac{1}{n})}$ for some sequences
$\Sequence_1\ldots\Sequence_n$, and every stochastic atom generated by
\ToStochastic is $\PRegexStar{\DNFRegex}{.8}$. In particular, $\ToDNFRegex$
generates regular expressions whose DNF form gives equal probability to all
sequence subcomponents of the \SDNFREabbrevs, and gives a .8 chance for
stars to continue iterating. In our experience generating random strings from regular
expressions, these probabilities provide good distributions of strings---stars
are iterated 4 times on average, and no individual choice in a series of
disjunctions is chosen disproportionately often.

\subsection{Relevance Annotations}
\label{subsec:relevanceannotations}

Even though our generated probability distribution works well in most
situations, it is not perfect. Consider synthesizing a lens between the formats
shown in Figure~\ref{fig:minimized-representations}. Because salary information is
present in \lstinline{emp_salaries}, but not in \lstinline{emp_insurance}, the
algorithm might spend a long time (fruitlessly) trying to construct lenses that
transform the salary to information present in \lstinline{emp_insurance} even
though that is impossible. In a similar but more elaborate example, such wasted
processing effort may cause synthesis to fail to terminate in any reasonable
amount of time.

One solution would be to cut off the search early. However, then one runs into
the opposite problem: In other scenarios, salary information may be present in
the other format, but it may take quite a bit of work to find a
transformation that connects the salary in \lstinline{emp_salaries} to the
salary in the second format. Hence, early termination
may cut off synthesis before the right lens is found.

We can solve both problems by allowing users to augment the specifications with
\emph{relevance annotations}. Sometimes, users have external knowledge that certain
information appears exclusively in one format or the other, or they may know the
information is present both formats. By communicating this knowledge to the synthesis algorithm
through relevance annotations, users can force the synthesis of lenses that
discard or retain certain information.
The first annotation, $\SkipOf{\Regex}$, says the information of
$\Regex$ appears only in $\Regex$, and can safely be
projected. The second annotation, $\SRequireOf{\Regex}$, says the information of
$\Regex$ appears in the other format and cannot be discarded.

For example, users can easily recognize that salary information is not present
in \lstinline{emp_insurance}. By annotating the \lstinline{salary} field as
\lstinline{skip(salary)}, users can add this knowledge to the specification to
optimize the search. In practice, we define the information content of
$\Skip(\Regex)$ to be zero.
\begin{gather*}
  \EntropyOf{\Skip(\Regex)} = 0
\end{gather*}

Similarly, users can recognize that employee names are present in both files. By
annotating instances of \lstinline{name} as \lstinline{require(name)}, users can
add this knowledge to the specification to optimize the search and force the
generated lens to retain \lstinline{name} information. In practice, we
make any lens that loses ``required'' information infinitely unlikely.
\begin{center}
  \begin{tabular}{rcl}
    $\REntropyOf{\SRequireOf{\RegexAlt} \Given \Lens, \Regex}$
    & =
    & $\begin{cases}
      \infty & \text{if }\REntropyOf{\SRequireOf{\RegexAlt} \Given \Lens, \Regex}
      \neq 0\\
    0 & \text{otherwise}
  \end{cases}$\\
    
    $\REntropyOf{\RegexAlt \Given \DisconnectOf{\Regex}{\RegexAlt}{\String}{\StringAlt}, \Regex}$
    & =
    & $\begin{cases}
      \infty & \vspace*{-.3em}\text{if $\RegexAlt$ contains $\SRequireOf{\RegexAlt'}$ as a}\\
      & \text{subexpression, where $\EntropyOf{\RegexAlt'} \neq 0$}\\
    [\EntropyOf{\RegexAlt},\EntropyOf{\RegexAlt}] & \text{otherwise}
  \end{cases}$\\
  \end{tabular}
\end{center}

\subsection{Symmetric DNF Lenses}
Symmetric DNF lenses are an intermediate synthesis target for \GreedySynth.
There are many fewer symmetric DNF lenses than symmetric regular lenses. In
fact, if one does not use the star-unrolling axioms, there are only finitely
many DNF lenses of a given type (though there are still many more symmetric DNF
lenses than DNF bijective lenses).

The structure of symmetric DNF lenses mirrors that of SDNF REs. A symmetric DNF
lens ($sdl$) is a union of symmetric sequence lenses, a symmetric sequence
lens ($ssql$) is a concatenation of symmetric atom lenses, and a
symmetric atom lens ($sal$) is an iteration of a symmetric DNF lens.

Just as we analyzed the information content of ordinary regular expressions, we
can analyze the information content of DNF regular expressions. As before, we
use $\REntropyOf{\DNFRegexAlt \Given \SDNFLens, \DNFRegex}$ to calculate bounds
on the expected amount of information required to recover a string in
$\DNFRegexAlt$ from a string in $\DNFRegex$, synchronized by $\SDNFLens$. We use
the function $\LEntropyOf{\DNFRegex \Given \SDNFLens, \DNFRegexAlt}$ to
calculate bounds on the expected amount of information required to recover a
string in $\DNFRegex$ from a string in $\DNFRegexAlt$, synchronized by
$\SDNFLens$.

The details of these definitions are syntactically tedious, but not
intellectually difficult. We elide them here but include them in
\ifappendices appendix \S\ref{sec:appendixdnf}.  \else{}the full version of the
paper.  \fi In the original Optician paper~\cite{optician}, DNF lenses were
proven equivalent in expressiveness to standard lenses. While we conjecture
symmetric DNF lenses are equivalent in expressivity to our standard symmetric
lenses, we have not proven this equivalence.

\subsection{\GreedySynth}
\label{subsec:greedy-synth}
The synthesis procedure comprises three algorithms: one that greedily finds
symmetric DNF lenses (\GreedySynth), one that greedily finds symmetric sequence
lenses (\GreedySeqSynth), and one that finds symmetric atom lenses
(\AtomSynth{}). These three algorithms are hierarchically structured:
\GreedySynth relies on \GreedySeqSynth, \GreedySeqSynth relies on \AtomSynth,
and \AtomSynth relies on \GreedySynth. The structure of the algorithms mirrors
the structure of symmetric DNF lenses and \SDNFREabbrevs.

\begin{algorithm}[t]
  \caption{\GreedySynth}
  \label{alg:greedysynth}
  \begin{algorithmic}[1]
    \Function{GreedySynth}{$\Examples,\DNFOf{(\Sequence_1,\Probability_1)\DNFSep\ldots\DNFSep(\Sequence_n,\Probability_n)},\DNFOf{(\SequenceAlt_1,\ProbabilityAlt_1)\DNFSep\ldots\DNFSep(\SequenceAlt_m,\ProbabilityAlt_m)}$}
    \If{$\Call{CannotMap}{\Examples,\DNFOf{(\Sequence_1,\Probability_1)\DNFSep\ldots\DNFSep(\Sequence_n,\Probability_n)},\DNFOf{(\SequenceAlt_1,\ProbabilityAlt_1)\DNFSep\ldots\DNFSep(\SequenceAlt_m,\ProbabilityAlt_m)}}$}
    \State $\ReturnVal{\None}$
    \EndIf
    \State $\SLS \gets
    \Call{\CartesianMap}{\GreedySeqSynth(\Examples),\ListOf{\Sequence_1;\ldots;\Sequence_n},\ListOf{\SequenceAlt_1;\ldots;\SequenceAlt_m}}$
    \State $\GreedyState \gets \Call{PQ.Create}{\SLS}$
    \State $\LensBuilder \gets \Call{LensBuilder.Empty}{}$
    \While{$\Call{PQ.IsNonempty}{\GreedyState}$}
    \State $\SSQLens \gets
    \Call{PQ.Pop}{\GreedyState}$
    \If{$\Call{LensBuilder.UsefulAdd}{\LensBuilder,\SSQLens,\Examples}$}
    \State $\LensBuilder \gets \Call{LensBuilder.AddSeq}{\LensBuilder,\SSQLens}$
    \EndIf
    \EndWhile
    \State \ReturnVal{\Call{LensBuilder.ToDNFLens}{\GreedyState}}
    \EndFunction
  \end{algorithmic}
\end{algorithm}

\paragraph*{Symmetric DNF Lenses} Algorithm~\ref{alg:greedysynth} presents
\GreedySynth, which synthesizes symmetric DNF lenses. Its inputs 
are a suite of input-output examples and a pair of stochastic DNF regular expressions.
First, \PCF{CannotMap} determines whether there is no lens satisfying the
examples, and if so, \GreedySynth returns \None immediately. Otherwise, \GreedySynth finds
the best lenses (given the examples that match them) between all sequence pairs
$(SQ_i, TQ_j)$ drawn from the left and right DNF regular expressions. (The
function \CartesianMap maps its argument across the cross product of the input
lists). A priority queue containing these sequence lenses, ordered by likelihood, is
then initialized with $\PCF{PQ.Create}$. The symmetric lens is then built up
iteratively from these sequence lenses, where the state of the partially
constructed lens is tracked in the lens builder, \LensBuilder.

\GreedySynth loops until there are no more sequence lenses in the priority
queue. Within this loop, a sequence lens is popped from the queue and, if it is
``useful,'' is included in the final DNF lens. The lens is considered to be
useful when its source (or target) is \textit{not} already the source (or
target) of an included sequence lens. If examples require that two
sequences have a lens between them, such lenses are considered useful. The
priorities of the sequence lenses update as the algorithm proceeds: if two
sequence lenses have the same source, the second one to be popped gets a higher
cost than it originally had; information must now be stored for including that
source of non-bijectivity.

As an example, consider searching for a lens between \lstinline{"" | name.name$^*$}
and \lstinline{"" | name}. \GreedySynth might first pop the
sequence lens between the
sequences \lstinline{""} and \lstinline{""}, because it is a bijective
sequence lens between. As neither \lstinline{""} is involved in a sequence
lens, this lens is considered useful. Next, the sequence lens between
\lstinline{name.name$^*$} and \lstinline{name} would be popped: while that lens
is not bijective it is still better than the alternatives. As all sequences are
now involved in sequence lenses, and there are no examples to make other lenses
useful, no more sequence lenses would be added to the lens builder.

Finally, after all sequences have been popped, the partial DNF lens \LensBuilder is
converted into a symmetric DNF lens. This is only possible if all sequences are
involved in some sequence lens: if they are not, \PCF{LensBuilder.ToDNFLens}
instead returns \None.

\paragraph*{Symmetric Sequence lenses} Algorithm~\ref{alg:greedyseqsynth}
presents \GreedySeqSynth, which synthesizes symmetric sequence lenses using an
algorithm whose structure is similar to \GreedySynth{}'s.  It calls \PCF{AtomSynth}, 
which synthesizes atom lenses by iterating a DNF lens between
its subcomponents. 

The inputs to \GreedySeqSynth are a suite of input-output examples and a pair of
lists of stochastic atoms. As in \GreedySynth, \GreedySeqSynth returns \None early if
there is no possible lens. Afterward, \GreedySeqSynth finds the best lenses
between each atom pair of the left and right sequences, and organizes them into
a priority queue ordered by likelihood with \PCF{PQ.Create}. The symmetric sequence
lens is built up iteratively from these atom lenses, where the state of the
partially built lens is tracked in the sequence lens builder, \SeqLensBuilder.

\GreedySeqSynth loops until there are no more atom lenses in the priority
queue. In the loop, a popped atom lens is considered 
``useful'' if adding it to the sequence will lower the
cost of the generated sequence lens, or if examples show that one of its atoms
must not be disconnected. Each atom can be part of only one lens at a time, so
the algorithm must sometimes remove a previously chosen atom lens in order to connect 
one that must not be disconnected.
The algorithm succeeds when all atoms that must not be disconnected
are involved in an atom lens; \PCF{SLensBuilder.ToDNFLens}
returns \None otherwise.

\begin{algorithm}[t]
  \caption{\GreedySeqSynth}
  \label{alg:greedyseqsynth}
  \begin{algorithmic}[1]
    \Function{AtomSynth}{$\Examples,\PRegexStar{\DNFRegex}{\Probability},\PRegexStar{\DNFRegexAlt}{\ProbabilityAlt}$}
    \If{$\Call{CannotMap}{\Examples,\PRegexStar{\DNFRegex}{\Probability},\PRegexStar{\DNFRegexAlt}{\ProbabilityAlt}}$}
    \State $\ReturnVal{\None}$
    \Else
    \State $\ReturnVal{\IterateLensOf{\GreedySynth(\Examples,\DNFRegex,\DNFRegexAlt)}}$
    \EndIf
    \EndFunction

    \Function{GreedySeqSynth}{$\Examples,\SequenceOf{\String_0\SeqSep\Atom_1\SeqSep\ldots\SeqSep\Atom_n\SeqSep\String_n},\SequenceOf{\StringAlt_0\SeqSep\AtomAlt_1\SeqSep\ldots\SeqSep\AtomAlt_m\SeqSep\String_m}$}
    \If{$\Call{CannotMap}{\Examples,\SequenceOf{\String_0\SeqSep\Atom_1\SeqSep\ldots\SeqSep\Atom_n\SeqSep\String_n},\SequenceOf{\StringAlt_0\SeqSep\AtomAlt_1\SeqSep\ldots\SeqSep\AtomAlt_m\SeqSep\String_m}}$}
    \State $\ReturnVal{\None}$
    \EndIf
    \State $\ALS \gets
    \Call{\CartesianMap}{\AtomSynth(\Examples),\ListOf{\Atom_1;\ldots;\Atom_n},\ListOf{\AtomAlt_1;\ldots;\AtomAlt_m}}$
    \State $\GreedyState \gets \Call{PQ.Create}{\ALS}$
    \State $\SeqLensBuilder \gets \Call{SLensBuilder.Empty}{}$
    \While{$\Call{PQ.IsNonempty}{\GreedyState}$}
    \State $\SAtomLens \gets
    \Call{PQ.Pop}{\GreedyState}$
    \If{$\Call{SLensBuilder.UsefulAdd}{\SeqLensBuilder,\SAtomLens,\Examples}$}
    \State $\GreedyState \gets \Call{SLensBuilder.AddAtom}{\GreedyState,\SAtomLens}$
    \EndIf
    \EndWhile
    \State \ReturnVal{\Call{SLensBuilder.ToDNFLens}{\GreedyState}}
    \EndFunction
  \end{algorithmic}
\end{algorithm}

\subsection{Optimizations}
\label{subsec:optimizations}

Our implementation includes a number of optimizations not described above:
annotations that guide DNF conversion; an expansion inference algorithm;
and compositional synthesis. The optimizations
make the system performant enough for interactive use.

\paragraph*{Open and Closed Regular Expressions} While \GreedySynth acts
relatively efficiently, it can suffer from an exponential blowup when converting
SREs to DNF form. This problem can be mitigated by avoiding the conversion
of some SREs to DNF form and by performing the conversion lazily when necessarily.
More specifically, \Expand labels some unconverted SREs as ``closed,'' which means the
type-directed \GreedySynth algorithm treats them as DNF atoms and does not dig into
them recursively.  In other words,
given a pair of closed SREs, \GreedySynth can either construct
the identity lens between them (it will do this if they are the same SRE), or it can
construct a disconnect lens between them.
Regular expressions that are not annotated as closed are considered ``open.''

Distinguishing between open and closed regular expressions improves the
efficiency of \GreedySynth, but forces \Expand to decide which closed
expressions to open. At the start of synthesis, all regular expressions are
closed, and \Expand rewrites selected closed regular expressions to open ones
(thereby triggering DNF normalization).

\paragraph*{Expansion Inference} These additional rewrites make the search
through possible regular expressions harder. Our algorithm identifies when
certain closed regular expressions can \emph{only} be involved in a disconnect
lens (unless opened). Such regular expressions will automatically be opened. The
full details of expansion inference are explained in \citet{optician}.

\paragraph*{Compositional Synthesis} Compositional synthesis allows \GreedySynth to use
previously defined (by users or synthesis) lenses. As the synthesizer processes a Boomerang file, it accumulates
lens definitions and types.  It tackles synthesis problems one after another
and there may be many such problems in a given program.
During synthesis, if an existing lens has the right type and agrees with the examples,
\GreedySynth will use it.  If a large synthesis problem cannot be solved all
at once, a user can generate subproblems for parts of a data format, and the
larger problem can subsequently use solutions to those subproblems. For example,
when given the synthesis task:
\begin{lstlisting}
let l_1 = synth R <=> S
let l_2 = synth R' <=> S'
\end{lstlisting}
the synthesis algorithm use the previously generated \lstinline{l_1} in the
synthesis of \lstinline{l_2}.
In our experience,
this is a very powerful tool that allows synthesis to tackle problems of arbitrary
complexity.

\section{Evaluation}
\label{sec:evaluation}
We implemented simple symmetric lenses as an extension to
Boomerang~\cite{boomerang}. In doing so, we reimplemented Boomerang's
asymmetric lens combinators using a combination of the symmetric
combinators presented in this paper and symmetric versions of asymmetric
extensions (like matching lenses~\cite{matchinglenses} and quotient
lenses~\cite{quotientlenses}) already present in Boomerang. We also integrated
our synthesis engine into Boomerang, allowing users to write synthesis tasks
alongside lens combinators, incorporate synthesis results into manually-written
lenses, and reference previously defined lenses during synthesis. All experiments
were performed on a 2.5 GHz Intel Core i7 processor with 16 GB of 1600 MHz DDR3
running macOS Mojave.

In this evaluation, we aim to answer four primary questions:
\begin{enumerate}
\item Can the algorithm (with suitable examples and annotations) find the correct lens?
\item Is the synthesis procedure efficient enough to be used in everyday
  development?
\item How much slower is our tool on bijective lens synthesis benchmarks than
  prior work customized for bijective lenses~\cite{optician}?
\item How effective is the information-theoretic search heuristic and how do our
  annotations affect the results?
\end{enumerate}

\subsection{Benchmark Suite}
Our benchmarks are drawn from three different sources.
\begin{enumerate}
\item We adapted 8 data cleaning benchmarks from Flash Fill~\cite{flashfill}. Flash Fill data
  cleaning tasks are either derived from online help forums or taken from the
  Excel product team. The 8 we chose are merely the first 8 found in the Flash
  Fill paper. Note that our tool produces bidirectional transformations
  rather than one-way transformers like Flash Fill. We ensure one direction of
  our bidirectional transformers performs the same unidirectional transformation
  as Flash Fill---the other direction is determined from the round-tripping
  laws. None of these benchmarks were bijective.
\item We adapted 29 benchmarks from Augeas~\cite{augeas}. Augeas is a utility
  that bidirectionally converts between Linux configuration files and an
  in-memory dictionaries. In our benchmarks, we synthesize lenses between Linux
  configuration files and serialized versions of Augeas dictionaries. We
  selected the first 28 benchmarks from Augeas in alphabetical order, and
  included one additional benchmark we considered particularly difficult,
  \texttt{xml\_to\_augeas.boom}. Because these benchmarks merely transformed the
  information into a structured form, synthesized lenses were all bijective.
\item We created 11 additional benchmarks 
  derived from real-world examples and/or the bidirectional programming
  literature. These tasks range from synchronizing REST and JSON web resource
  descriptions to synchronizing BibTeX and EndNote citation descriptions. Five
  of these benchmarks were not bijective lenses.
\end{enumerate}

Each task consists of a source and target regular expression and a set of
examples. We manually programmed the regular expressions. However, when writing
these regular expressions, we did not massage them to be amenable to synthesis.
The synthesizer implements the star-semiring regular expression equivalence
axioms, so that users do not have to worry about the way they write expressions
themselves. We are synthesizing transformations between real formats, so these
tasks are can be quite large -- on average, the size of each RE in the
specification is 1637 and the size of the synthesized lens is 3254. For each of
these tasks, we selected tasks until we did not feel additional examples would
be elucidating.

\subsection{Synthesizing Correct Lenses}

To determine whether the system can synthesize desired lenses, we ran it
interactively on all 48 tasks, working with the system to create sufficient
examples and provide useful relevance annotations.  In all cases, the desired lens
was obtained.
The majority of the tasks required only a single example and none required
more than three examples to synthesize the desired lens.\footnote{In
  one benchmark, we supplied a fourth example that was later discovered to 
  be unnecessary.}  

Providing relevance annotations was needed in only 8 of the 48 tasks. In
practice, we found that adding such annotations quite easy: if manual inspection
of the lens showed there were too few \IdentityLens{}s, and too many
\Disconnect{}s or merges, we would add \SRequire annotations. If synthesis took
too long, we would add \Skip annotations.
Section~\ref{subsec:effect-of-heuristics} studies the effects of removing such
annotations.

We verified that our default running mode (\SSOpt{}) generated the correct lenses the way programmers often
validate their programs: we manually inspected the code and ran unit tests on
the synthesized code. To determine whether the synthesis procedure generated the
correct lens when running in modes other than \SSOpt{}, we compared generated lens
to the lens synthesized by \SSOpt{}.


\subsection{Effectiveness of Compositional Synthesis}

Having determined appropriate examples and annotations for the 48
benchmarks, we evaluate the performance of the system by measuring the running
time of our algorithm in two modes:

\begin{tabulary}{\linewidth}{rL}
  \SSOpt{}: & Run the symmetric synthesis algorithm with all optimizations enabled.\\
  \SSNCOpt{}: & Run the symmetric synthesis algorithm, with no compositional synthesis enabled.\\
\end{tabulary}\\

Recall that compositional synthesis allows users to break a benchmark into a
series of smaller synthesis tasks, whose solutions are utilized in more complex
synthesis procedures. Compositional synthesis (\SSOpt{} mode)
allows our system to scale to arbitrarily large and complex formats; measuring
it shows the responsiveness of the system when used as
intended. \SSNCOpt{} mode, which synthesizes a complete lens all at once,
provides a useful experimental stress test for the system.

\begin{figure}
  \includegraphics{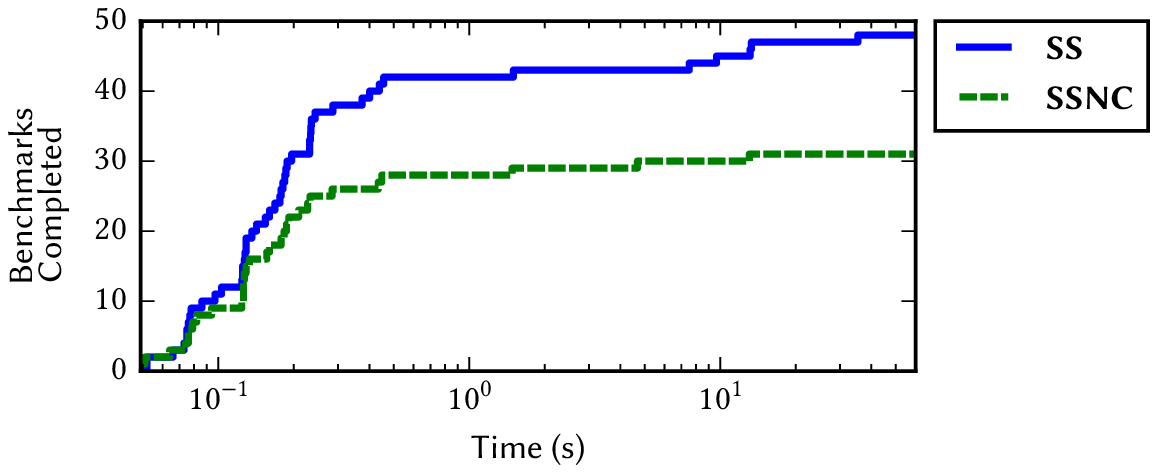}
  \vspace{-4ex}
  \caption{Number of
    benchmarks that can be solved by a given algorithm in a 
    given amount of time. \SSOpt{} is the full symmetric synthesis algorithm.
    \SSNCOpt{} is the symmetric synthesis algorithm without using a library of
    existing lenses. The symmetric synthesis algorithm is able to complete all
    benchmarks in under 30 seconds elapsed total time. Without compositional
    synthesis it is able to complete 31. Each benchmark specification includes
    source and target (potentially annotated) regular expressions, and between
    one and three sufficient examples.}
  \label{fig:times}
\end{figure}

For each benchmark in the suite and each mode, we ran the system with a timeout
of 60 seconds, averaging the result over 5 runs. Figure~\ref{fig:times}
summarizes the results of these tests. We find that our algorithm is able to
synthesize all of the benchmarks in under 30 seconds. Without compositional
synthesis, the synthesis algorithm is able to solve 31 out of 48 problem
instances. In total, 73 existing lenses were used in compositional synthesis
(about 1.5 per benchmark on average).


\subsection{Slowdown Compared to Bijective Synthesis}

To compare to the existing bijective synthesis algorithm, we run our symmetric
synthesis algorithm on the original Optician benchmarks, comprised of
39 bijective synthesis tasks.\footnote{We had to slightly alter four of these
benchmarks, either by providing additional examples or by adding in \SRequire annotations.
Without these alterations, symmetric synthesis yielded a lens that fit
the specification but that was undesired.}

To perform this comparison, we synthesized lenses in two modes:

\begin{tabulary}{\linewidth}{rL}
  \BSOpt{}: & The existing bijective synthesis
              algorithm with all optimizations enabled.\\
  \SSOpt{}: & The symmetric synthesis algorithm with all optimizations enabled.\\
\end{tabulary}\\

For each benchmark, we ran it in both modes with a
timeout of 60 seconds and averaged the result over 5 runs.
Figure~\ref{fig:times_bijective} summarizes the results of these tests. On
average, \SSOpt{} took 1.3 times (0.5 seconds) longer to complete than
\BSOpt{}. The slowest completed benchmark for both synthesis algorithms is
\texttt{xml\_to\_augeas.boom}, a benchmark that converts arbitrary XML up to
depth 3 into a serialized version of the structured dictionary representation
used in Augeas. This benchmark takes 18.9 seconds for the symmetric synthesis
algorithm to complete, and 9.3 seconds the bijective synthesis algorithm to
complete.

Both the bijective synthesis and the symmetric synthesis engines use a pair of collaborating
synthesizers that (1) search for a compatible pair of regular expressions and (2) search
for a lens given those regular expressions. Bijective synthesis is faster than symmetric synthesis
because part (2) is much faster. Specifically, a bijection must translate all data on the left into data
on the right, and this fact constrains the search. By contrast, a symmetric synthesis problem has a
choice of which data on the left to translate into data on the right. This choice gives rise to a large set of
other choices, and symmetric synthesis must consider all of them.

\begin{figure}[t]
  \includegraphics{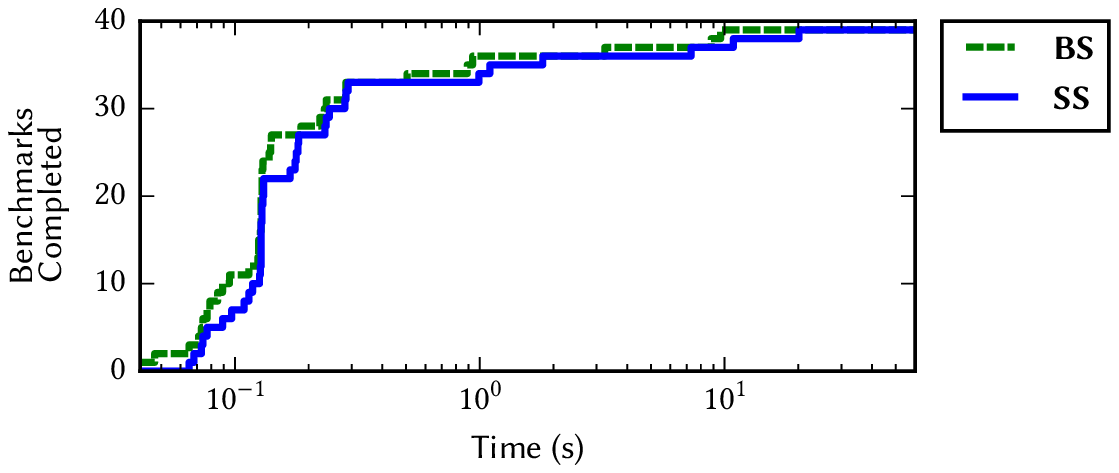}
  \vspace{-4ex}
  \caption{Number of benchmarks that can be solved by a given algorithm in a
    given amount of time. \SSOpt{} is the full symmetric synthesis algorithm.
    \BSOpt{} is the full bijective lens synthesis algorithm.}
  \label{fig:times_bijective}
\end{figure}

\subsection{The Effects of Heuristics and Relevance Annotations}
\label{subsec:effect-of-heuristics}
We evaluate the usefulness of (1) our information-theoretic metric, (2) our
termination heuristic and (3) our
relevance annotations.  To this end, we run our program in several different modes:\\[1ex]
\begin{tabulary}{\linewidth}{rL}
  \AnyOpt{}: & Ignore the information-theoretic preference metric (\IE, all valid
               lenses have cost 0). \\ 
  \FLOpt{}: &  Return the first highest ranked lens \GreedySynth returns (\IE,
              ignore the termination heuristic). \\
  \CCOpt{}: & Replace our information-theoretic cost metric with one where
              the cost of the lens is the number of disconnects plus the number
              of merges.\\
  \NSOpt{}: & Ignore all \Skip annotations in the SRE specifications. \\
  \NROpt{}: & Ignore all \SRequire annotations in the SRE specifications. \\
\end{tabulary}

We experimented with the \CCOpt{} mode to determine whether the
complexity of the information-theoretic measure is really
needed. Related work on string transformations has often used simpler
measures such as ``avoid constants'' that align with, but are simpler
than our measures. The \CCOpt{} mode is an example of such a simple
measure---it operates by counting disconnects, which put a complete
stop to information transfer, and merges, which
eliminate the information in a union. 

\begin{figure}[t]
  \includegraphics{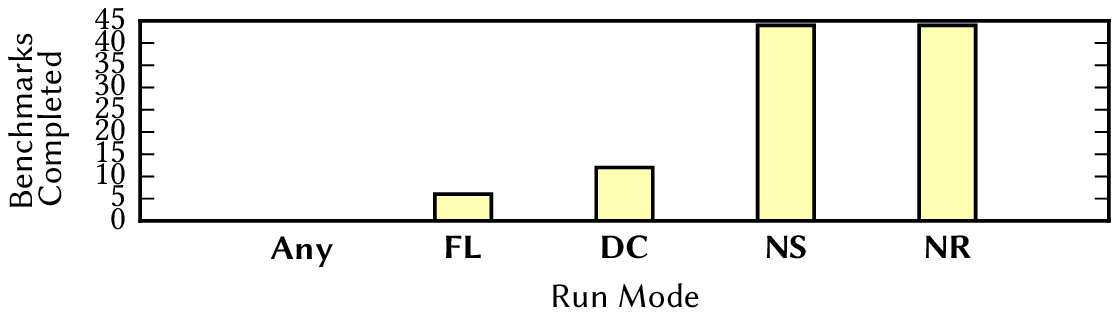}
  \vspace{-4ex}
  \caption{Number of benchmarks that synthesize the correct lens by a given
    algorithm. \AnyOpt{} provides no notion of cost, and merely returns the
    first lens it finds that satisfies the specification. \FLOpt{} provides a
    notion of cost to \GreedySynth, but once a satisfying lens is greedily
    found, that lens is returned. \CCOpt{} synthesizes lenses, where the cost of
    a lens is the number of disconnects plus the number of merges. \NSOpt
    ignores all \Skip annotations while running the algorithm. \NROpt ignores
    all \SRequire annotations while running the algorithm.}
  \label{fig:metric}
\end{figure}

Figure~\ref{fig:metric} summarizes the result of these experiments. The data
reveal that the information-theoretic metric is critical for finding the correct
lens: Only 10 of the benchmarks succeeded when running in \CCOpt{} mode. The
termination condition is also quite important. When running in \FLOpt{} mode,
the algorithm only discovers 5 lenses, which shows that the first class that
contains a satisfying lens is rarely the correct class. However, our algorithm
is not perfect and fails when either it is very difficult to find the desired
lens (necessitating \SRequire) or when a large amount of data is projected
(necessitating \Skip). Without any annotations, our algorithm finds the correct
lens for 40 of our 48 benchmarks; eight required relevance annotations to find
the correct lens. In total, there are 12 uses of \SRequire, and 4 uses of \Skip
in our benchmark suite.

\section{Related Work}
\label{sec:related}

In this paper, we have designed and implemented a new, pragmatic
formulation of symmetric lenses as well as new program synthesis
techniques.  In this section, we analyze the relationship
between our simple symmetric
lenses and two previous lens languages:  classical symmetric lenses
and asymmetric lenses. Proofs about these relationships are included in
\ifappendices appendix \S\ref{sec:appendixforget}. \else{}the full version of
the paper.  \fi We also comment on related
program synthesis techniques.

\subsection{Relationship with Classical Lens Languages}
\label{sec:relationship}

\paragraph*{Symmetric Lenses}
A classical symmetric lens\cite{symmetric-lenses} $\Lens$ between $X$ and $Y$ consists of 4 components:
a complement $C$, a designated element $init \in C$, and two functions,
$\PutRSym \OfType X \times C \rightarrow Y \times C$ and $\PutLSym \OfType Y
\times C \rightarrow X \times C$, that propagate changes in one format to the
other.

In this formulation, data unique to each side are stored in the complement. When
one format is edited, the \PutR or \PutL function stitches together the edited
data with data stored in the complement. The $init$ element is the initial value
of $C$ and specifies default behavior when data is missing. For instance, to
implement the scenario in Figure~\ref{fig:minimized-representations}, the
complement would consist of a list of pairs of salary and company name.
Classical symmetric lenses satisfy the following equational laws.
\begin{centermath}
\begin{array}[b]{l@{\qquad \; \qquad}l}
  \begin{mathprooftree}
    \AxiomC{$\PutRSymOf{(x,c)} = (y,c')$}
    \UnaryInfC{$\PutLSymOf{(y,c')} = (x,c')$}
  \end{mathprooftree}
&
  \begin{mathprooftree}
    \AxiomC{$\PutLSymOf{(y,c)} = (x,c')$}
    \UnaryInfC{$\PutRSymOf{(x,c')} = (y,c')$}
  \end{mathprooftree}
\end{array}
\end{centermath}

Two classical symmetric lenses are equivalent if they output the same formats
given any sequence of edits. Formally, given a lens $\Lens$ between $X$ and $Y$, an \emph{edit} for $\Lens$ is a member of $X + Y$. Consider the function
$apply$, which, given a lens and an element of that lens's complement, is a
function from sequences of edits to sequences of edits. If $apply(\Lens,c,es) =
es'$, then given complement $c$ and edit $es_i$, the lens $\Lens$ generates
$es'_i$.
\begin{centermath}
\begin{tabular}[b]{l@{\qquad \; \qquad}l}
$
  \inferrule
  {
  }
  {
    apply(\Lens,c,[]) = []
  }
$
&
$
  \inferrule
  {
    \Lens.putr(x,c) = (y,c')\\
    apply(\Lens,c',es) = es'
  }
  {
    apply(\Lens,c,(\InLOf{x})::es) = (\InROf{y})::es'
  }
$
\end{tabular}
\end{centermath}
\[
  \inferrule
  {
    \Lens.putl(y,c) = (x,c')\\
    apply(\Lens,c',es) = es'
  }
  {
    apply(\Lens,c,(\InROf{y})::es) = (\InLOf{x})::es'
  }
\]
\noindent
Two lenses, $\Lens_1$ and $\Lens_2$, are equivalent if
$apply(\Lens_1,\Lens_1.init,es) = apply(\Lens_2,\Lens_2.init,es)$ for all $es$.

Simple symmetric lenses are a strict subset of classical symmetric lenses, but
quite a useful one. For instance, all asymmetric lenses are expressible as
simple symmetric lenses. The primary loss is the loss of ``memory'' within the
complement. In classical symmetric lenses, disjunctive (\OrLens) lenses retain information
about both possible formats. If a user edits a format from one disjuncted format
to the other, the information contained in that first disjunct is retained
within the complement. Simple symmetric lenses have no such complement, so they
mimic the forgetful disjunctive lens of classical symmetric lenses.

Though classical symmetric lenses are more expressive, they introduce drawbacks
for synthesis: because each lens has a
custom complement, one can no longer specify the put functions through
input/output examples alone. One alternative would be to enrich specifications
with edit sequences; another would be to specify the structure of complements
explicitly (though the latter would be somewhat akin to specifying the internal
state of a program). In either case, the complexity of the specifications
increases.

\paragraph*{Symmetric Lenses vs.{} Simple Symmetric Lenses} To compare classical and simple
symmetric lenses, we define an $apply$ function on simple symmetric lenses as well. If
$apply(\Lens,\None,es) = es'$, then starting with no prior data, after edit
$es_i$, the lens \Lens generates $es_i'$ (the right format if $es_i =
\InLOf{x}$, and the left format if $es_i = \InROf{y}$). If
$apply(\Lens,\SomeOf{(x,y)},es) = es'$, then starting with data $x$
and $y$ on the left and right, respectively, after edit $es_i$, the lens \Lens
generates $es_i'$.

\begin{centermath}
\begin{tabular}[b]{l@{\qquad \; \qquad}l}
$
  \inferrule
  {
  }
  {
    apply(\Lens,xyo,[]) = []
  }
$
&
$
  \inferrule
  {
    \Lens.\CreateROf{x} = y\\
    apply(\Lens,\SomeOf{(x,y)},es) = es'
  }
  {
    apply(l,\None,\InLOf{x}::es) = \InROf{y}::es'
  }
$
\end{tabular}
\end{centermath}
\[
  \inferrule
  {
    \Lens.\CreateLOf{y} = x\\
    apply(\Lens,\SomeOf{(x,y)},es) = es'
  }
  {
    apply(l,\None,\InROf{y}::es) = \InLOf{x}::es'
  }
\]
\[
  \inferrule
  {
    \Lens.\PutROf{x'}{y}  = y'\\
    apply(\Lens,\SomeOf{(x',y')},es) = es'
  }
  {
    apply(l,\SomeOf{(x,y)},\InLOf{x'}::es) = \InROf{y'}::es'
  }
\]
\[
  \inferrule
  {
    \Lens.\PutLOf{y'}{x}  = x'\\
    apply(\Lens,\SomeOf{(x',y')},es) = es'
  }
  {
    apply(l,\SomeOf{(x,y)},\InROf{y'}::es) = \InLOf{x'}::es'
  }
\]

Next, we define \emph{forgetful symmetric
  lenses} to be symmetric lenses that satisfy the following additional laws.
\begin{equation}
  \tag{\ForgetfulRL}
  \begin{mathprooftree}
    \AxiomC{$\Lens.putr(x,c_1) = (\_,c_1')$}
    \def\extraVskip{.5pt}
    \noLine 
    \UnaryInfC{$\Lens.putr(x,c_2) = (\_,c_2')$}
    \AxiomC{$\Lens.putl(y,c_1') = (\_,c_1'')$}
    \def\extraVskip{.5pt}
    \noLine 
    \UnaryInfC{$\Lens.putl(y,c_2') = (\_,c_2'')$}
    \def\extraVskip{2pt}
    \singleLine
    \BinaryInfC{$c_1'' = c_2''$}
  \end{mathprooftree}
\end{equation}
\begin{equation}
  \tag{\ForgetfulLR}
  \begin{mathprooftree}
    \AxiomC{$\Lens.putl(y,c_1) = (\_,c_1')$}
    \def\extraVskip{.5pt}
    \noLine 
    \UnaryInfC{$\Lens.putl(y,c_2) = (\_,c_2')$}
    \AxiomC{$\Lens.putr(x,c_1') = (\_,c_1'')$}
    \def\extraVskip{.5pt}
    \noLine 
    \UnaryInfC{$\Lens.putr(x,c_2') = (\_,c_2'')$}
    \def\extraVskip{2pt}
    \singleLine
    \BinaryInfC{$c_1'' = c_2''$}
  \end{mathprooftree}
\end{equation}

Intuitively, these equations state that complements are uniquely determined by
the most recent input $x$ and $y$. Such lenses correspond exactly with simple
symmetric lenses, where all state is maintained by the $x$ and $y$ data.

\begin{theorem}
  Let $\Lens$ be a classical symmetric lens. The lens $\Lens$ is equivalent to a forgetful
  lens if, and only if, there exists a simple symmetric lens $\Lens'$ where
  $apply(\Lens,\Lens.init,es) = apply(\Lens',\None,es)$, for all put sequences $es$.
\end{theorem}

\paragraph*{Asymmetric Lenses}
Formally, an {\em asymmetric lens} $ \ell : S \Leftrightarrow V$ is a triple of functions $\ell.\get : S \longrightarrow V$, $\ell.\pput : V \longrightarrow S \longrightarrow S$ and $\ell.\create : V \longrightarrow S$ satisfying the following laws \cite{Focal2005-long}:
\begin{align*}
\ell.\get \; (\ell.\pput \; s \; v) &= v \tag{PUTGET}\\
\ell.\pput \; s \; (\ell.\get \; s) &= s \tag{GETPUT}\\
\ell.\get \; (\ell.\create \; v) &= v \tag{CREATEGET}
\end{align*}
Simple symmetric lenses are strictly more expressive than classical asymmetric
lenses.

\begin{theorem}
  Let $\Lens$ be an asymmetric lens. $\Lens$ is also a simple symmetric lens,
  where:
  \begin{center}
    \begin{tabular}{rcl@{\hspace*{3em}}rcl}
      $\Lens.\CreateLOf{y}$ & $=$ & $\Lens.create \App y$
      & $\Lens.\CreateROf{x}$ & $=$ & $\Lens.get \App x$\\
      $\Lens.\PutLOf{y}{x}$ & $=$ & $\Lens.put \App y \App x$
      & $\Lens.\PutROf{x}{y}$ & $=$ & $\Lens.get \App x$
    \end{tabular}
  \end{center}
\end{theorem}

\paragraph*{Other Lens Formulations}
Bidirectional programming has an extensive literature, with many extensions onto
basic lenses, like quotient lenses~\cite{quotientlenses} and matching
lenses~\cite{matchinglenses}. Readers can consult
\citet{DBLP:conf/icmt/CzarneckiFHLST09} for a survey of lenses
and lens-like structures.

\subsection{Data Transformation Synthesis}
Over the past decade, the programming languages community has explored the
synthesis of programs from a wide variety of angles. One of the key ideas is
typically to narrow the program search space by focusing on a specific domain,
and imposing constraints on syntax~\cite{sygus},
typing~\cite{augustsson-2004,gvero-pldi-2013,osera+:pldi15,feser-pldi-2015,scherer-icfp-2015,frankle+:popl16},
or both.

Automation of string transformations, in particular, has been the focus
of much prior attention.  For example, Gulwani's and others' work on FlashFill
generates one-way spreadsheet transformations from input/output
examples~\cite{flashfill,le-pldi-2014}.  On the one hand,
FlashFill is easier to use because one need not specify the type of
the data being transformed.  On the other hand, this type information
makes it possible to transform more complex formats. For example, Flash Fill
does not synthesize programs with nested loops~\cite{flashfill},
and hence is incapable of synthesizing the majority of our
benchmarks, even in one direction. A comparison of Optician (on
bijective benchmarks) to these synthesis tools is included in~\citet{optician}.

All pragmatic synthesis algorithms use heuristics of one kind or another. One of
the goals of the current paper is to try to ground those heuristics in a broader
theory, information theory, with the hope that this theory may inform future
design decisions and help us understand heuristics crafted in other tools and
possibly in other domains. For instance, we speculate that some of the
heuristics used in FlashFill (to take one well-documented example) may be
connected to some of the principles laid out here. For instance, Flash Fill
prioritizes the substring constructor over the constant string constructor when
ranking possible programs. Such a choice is consistent with our
information-theoretic viewpoint as the constant function throws away a great
deal of information about the source string being transformed. Likewise, Flash
Fill prefers ``wider'' character classes over ``narrower'' ones. Again such a
choice is implied by information theory---the wider class preserves more
information during translation. More broadly, we hope our information-theoretic
analysis provides a basis for understanding the heuristic choices made in
related work.


As discussed in the introduction,
the Optician tool~\cite{optician,maina+:quotient-synthesis} was a building block for
our work.  Optician synthesized bijective transformations~\cite{optician} and
bijections modulo quotients~\cite{maina+:quotient-synthesis}, but could
not synthesize more complex bidirectional transformations where one format contains
important information not present in the other---a common situation in the real
world.  From a technical perspective, the first key novelty in our work involves the definition, theory,
analysis, and implementation of a new class of simple symmetric lenses, designed
for synthesis.  The second key technical innovation involves the use of stochastic
regular expressions and information theory to guide the search for program transformations.
As mentioned earlier, we believe such information-theoretic techniques may have broad utility
in helping us understand how to formulate a search for a data transformation function.

While our tool uses types and examples to specify invertible transformations, other
tools have been shown to synthesize a backwards transformation from ordinary code designed
to implement the forwards transformation.  For instance, \citet{program-inversion-symbolic-transducers} show how to invert transformations
using symbolic finite automata, and
\citet{bidirectionalization-for-free} demonstrates how to
construct reverse transformations from forwards transformations
by exploiting parametricity properties.  These kinds of tools are very useful, but in different
circumstances from ours (namely, when one already has the code to implement one direction of
the transformation).

More broadly, information theory and probabilistic languages are common tools in
natural language understanding, machine translation, information retrieval, data
extraction and grammatical inference (see \citet{pereira:info-theory}, for an
introduction). Indeed, our work was inspired, in part, by the PADS format
inference tool~\cite{pads:synthesis}, which was in turn inspired by earlier work
on probabilistic grammatical inference and information
extraction~\cite{kushmerick-thesis,arasu:extracting}. PADS did not synthesize
data transformers, and we are not aware of the use of information-theoretic
principles in type-directed or syntax-guided synthesis of DSL programs. More
recently, the principles used in \citet{pads:synthesis} have been applied in
unsupervised learning~\cite{ULPS}. This work learns compact descriptions of a
single data set. Ideally, such descriptions are compact and information theory
is used as a measure of the compactness of the description learned. In contrast,
we are attempting to learn a translation from one data set to another. However,
there are many different candidate translations. To select amongst the candidate
translations, we choose the translation that preserves as much information from
source to target (and vice versa) as possible.





\section{Conclusion}
\label{sec:conc}
We described a synthesis algorithm for synthesizing synchronization
functions between data formats that may not be in bijective correspondence.
We identified a subset of symmetric lenses, simple symmetric lenses, 
develop new combinators for them, and show how to synthesize them from regular
expression specifications. In order to guide the search for ``likely'' lenses,
we introduce new search principles based on information theory and allow users
to control this search via relevance annotations. To evaluate the effectiveness
of these ideas, we designed and implemented a new tool for symmetric lens
synthesis and integrated it into the Boomerang lens programming framework. Our
experiments on 48 benchmarks demonstrate that we can synthesize complex lenses for
real-life formats in under 30 seconds.

\begin{acks}
  We thank our anonymous reviewers for their useful feedback and discussions and
  Nate Foster and Michael Greenberg for their help integrating Optician into
  Boomerang. This research has been supported in part by DARPA award
  FA8750-17-2-0028 and ONR 568751 (SynCrypt).
\end{acks}

\ifanon\else
\fi



\bibliography{local,bcp}


\newcommand{\SortNoop}[1]{}
\begin{thebibliography}{33}


\ifx \showCODEN    \undefined \def \showCODEN     #1{\unskip}     \fi
\ifx \showDOI      \undefined \def \showDOI       #1{#1}\fi
\ifx \showISBNx    \undefined \def \showISBNx     #1{\unskip}     \fi
\ifx \showISBNxiii \undefined \def \showISBNxiii  #1{\unskip}     \fi
\ifx \showISSN     \undefined \def \showISSN      #1{\unskip}     \fi
\ifx \showLCCN     \undefined \def \showLCCN      #1{\unskip}     \fi
\ifx \shownote     \undefined \def \shownote      #1{#1}          \fi
\ifx \showarticletitle \undefined \def \showarticletitle #1{#1}   \fi
\ifx \showURL      \undefined \def \showURL       {\relax}        \fi
\providecommand\bibfield[2]{#2}
\providecommand\bibinfo[2]{#2}
\providecommand\natexlab[1]{#1}
\providecommand\showeprint[2][]{arXiv:#2}

\bibitem[\protect\citeauthoryear{Alur, Bodik, Juniwal, Martin, Raghothaman,
  Seshia, Singh, Solar-Lezama, Torlak, and Udupa}{Alur et~al\mbox{.}}{2013}]%
        {sygus}
\bibfield{author}{\bibinfo{person}{Rajeev Alur}, \bibinfo{person}{Rastislav
  Bodik}, \bibinfo{person}{Garvit Juniwal}, \bibinfo{person}{Milo M.~K.
  Martin}, \bibinfo{person}{Mukund Raghothaman}, \bibinfo{person}{Sanjit~A.
  Seshia}, \bibinfo{person}{Rishabh Singh}, \bibinfo{person}{Armando
  Solar-Lezama}, \bibinfo{person}{Emina Torlak}, {and}
  \bibinfo{person}{Abhishek Udupa}.} \bibinfo{year}{2013}\natexlab{}.
\newblock \showarticletitle{Syntax-Guided Synthesis}. In
  \bibinfo{booktitle}{{\em Proceedings of the IEEE International Conference on
  Formal Methods in Computer-Aided Design (FMCAD)}}. \bibinfo{pages}{1--17}.
\newblock


\bibitem[\protect\citeauthoryear{Arasu and Garcia-Molina}{Arasu and
  Garcia-Molina}{2003}]%
        {arasu:extracting}
\bibfield{author}{\bibinfo{person}{Arvind Arasu} {and} \bibinfo{person}{Hector
  Garcia-Molina}.} \bibinfo{year}{2003}\natexlab{}.
\newblock \showarticletitle{Extracting Structured Data from Web Pages}. In
  \bibinfo{booktitle}{{\em Proceedings of the 2003 ACM SIGMOD International
  Conference on Management of Data}} {\em (\bibinfo{series}{SIGMOD '03})}.
  \bibinfo{pages}{337--348}.
\newblock
\showISBNx{1-58113-634-X}


\bibitem[\protect\citeauthoryear{Augustsson}{Augustsson}{2004}]%
        {augustsson-2004}
\bibfield{author}{\bibinfo{person}{Lennart Augustsson}.}
  \bibinfo{year}{2004}\natexlab{}.
\newblock \bibinfo{title}{[Haskell] Announcing Djinn, version 2004-12-11, a
  coding wizard}.
\newblock \bibinfo{howpublished}{Mailing List}.   (\bibinfo{year}{2004}).
\newblock
\newblock
\shownote{\url{http://www.haskell.org/pipermail/haskell/2005-December/017055.html}.}


\bibitem[\protect\citeauthoryear{Barbosa, Cretin, Foster, Greenberg, and
  Pierce}{Barbosa et~al\mbox{.}}{2010}]%
        {matchinglenses}
\bibfield{author}{\bibinfo{person}{Davi~M.J. Barbosa}, \bibinfo{person}{Julien
  Cretin}, \bibinfo{person}{Nate Foster}, \bibinfo{person}{Michael Greenberg},
  {and} \bibinfo{person}{Benjamin~C. Pierce}.} \bibinfo{year}{2010}\natexlab{}.
\newblock \showarticletitle{Matching Lenses: Alignment and View Update}. In
  \bibinfo{booktitle}{{\em Proceedings of the 15th ACM SIGPLAN International
  Conference on Functional Programming}} {\em (\bibinfo{series}{ICFP '10})}.
  \bibinfo{publisher}{ACM}, \bibinfo{address}{New York, NY, USA},
  \bibinfo{pages}{193--204}.
\newblock
\showISBNx{978-1-60558-794-3}
\showDOI{%
\url{https://doi.org/10.1145/1863543.1863572}}


\bibitem[\protect\citeauthoryear{Bohannon, Foster, Pierce, Pilkiewicz, and
  Schmitt}{Bohannon et~al\mbox{.}}{2008}]%
        {boomerang}
\bibfield{author}{\bibinfo{person}{Aaron Bohannon}, \bibinfo{person}{J.~Nathan
  Foster}, \bibinfo{person}{Benjamin~C. Pierce}, \bibinfo{person}{Alexandre
  Pilkiewicz}, {and} \bibinfo{person}{Alan Schmitt}.}
  \bibinfo{year}{2008}\natexlab{}.
\newblock \showarticletitle{Boomerang: Resourceful Lenses for String Data}. In
  \bibinfo{booktitle}{{\em Proceedings of the 35th Annual ACM SIGPLAN-SIGACT
  Symposium on Principles of Programming Languages}} {\em
  (\bibinfo{series}{POPL '08})}. \bibinfo{publisher}{ACM}.
\newblock


\bibitem[\protect\citeauthoryear{Bohannon, Vaughan, and Pierce}{Bohannon
  et~al\mbox{.}}{2006}]%
        {BohannonPierceVaughan}
\bibfield{author}{\bibinfo{person}{Aaron Bohannon}, \bibinfo{person}{Jeffrey~A.
  Vaughan}, {and} \bibinfo{person}{Benjamin~C. Pierce}.}
  \bibinfo{year}{2006}\natexlab{}.
\newblock \showarticletitle{Relational Lenses: {A} Language for Updateable
  Views}. In \bibinfo{booktitle}{{\em Principles of Database Systems (PODS)}}.
\newblock
\newblock
\shownote{Extended version available as University of Pennsylvania technical
  report MS-CIS-05-27.}


\bibitem[\protect\citeauthoryear{Calendars}{Calendars}{2016}]%
        {calendar-formats}
Calendars \bibinfo{year}{2016}\natexlab{}.
\newblock \bibinfo{title}{Calendars}.
\newblock
  \bibinfo{howpublished}{\url{http://fileformats.archiveteam.org/wiki/Calendars}}.
    (\bibinfo{year}{2016}).
\newblock


\bibitem[\protect\citeauthoryear{Carrasco, Forcada, and
  Santamar{\'i}a}{Carrasco et~al\mbox{.}}{1996}]%
        {stoch-rnn}
\bibfield{author}{\bibinfo{person}{Rafael~C. Carrasco},
  \bibinfo{person}{Mikel~L. Forcada}, {and} \bibinfo{person}{Laureano
  Santamar{\'i}a}.} \bibinfo{year}{1996}\natexlab{}.
\newblock \showarticletitle{Inferring stochastic regular grammars with
  recurrent neural networks}. In \bibinfo{booktitle}{{\em Grammatical
  Interference: Learning Syntax from Sentences}},
  \bibfield{editor}{\bibinfo{person}{Laurent Miclet} {and}
  \bibinfo{person}{Colin de~la Higuera}} (Eds.). \bibinfo{publisher}{Springer
  Berlin Heidelberg}, \bibinfo{address}{Berlin, Heidelberg},
  \bibinfo{pages}{274--281}.
\newblock
\showISBNx{978-3-540-70678-6}


\bibitem[\protect\citeauthoryear{Czarnecki, Foster, Hu, L{\"a}mmel, Sch{\"u}rr,
  and Terwilliger}{Czarnecki et~al\mbox{.}}{2009}]%
        {DBLP:conf/icmt/CzarneckiFHLST09}
\bibfield{author}{\bibinfo{person}{Krzysztof Czarnecki},
  \bibinfo{person}{J.~Nathan Foster}, \bibinfo{person}{Zhenjiang Hu},
  \bibinfo{person}{Ralf L{\"a}mmel}, \bibinfo{person}{Andy Sch{\"u}rr}, {and}
  \bibinfo{person}{James~F. Terwilliger}.} \bibinfo{year}{2009}\natexlab{}.
\newblock \showarticletitle{Bidirectional Transformations: A Cross-Discipline
  Perspective}. In \bibinfo{booktitle}{{\em ICMT}} {\em
  (\bibinfo{series}{Lecture Notes in Computer Science})},
  \bibfield{editor}{\bibinfo{person}{Richard~F. Paige}} (Ed.),
  Vol.~\bibinfo{volume}{5563}. \bibinfo{publisher}{Springer},
  \bibinfo{pages}{260--283}.
\newblock
\showISBNx{978-3-642-02407-8}


\bibitem[\protect\citeauthoryear{Ellis, Solar-Lezama, and Tenenbaum}{Ellis
  et~al\mbox{.}}{2015}]%
        {ULPS}
\bibfield{author}{\bibinfo{person}{Kevin Ellis}, \bibinfo{person}{Armando
  Solar-Lezama}, {and} \bibinfo{person}{Joshua~B. Tenenbaum}.}
  \bibinfo{year}{2015}\natexlab{}.
\newblock \showarticletitle{Unsupervised Learning by Program Synthesis}. In
  \bibinfo{booktitle}{{\em Proceedings of the 28th International Conference on
  Neural Information Processing Systems - Volume 1}} {\em
  (\bibinfo{series}{NIPS'15})}. \bibinfo{publisher}{MIT Press},
  \bibinfo{address}{Cambridge, MA, USA}, \bibinfo{pages}{973--981}.
\newblock
\showURL{%
\url{http://dl.acm.org/citation.cfm?id=2969239.2969348}}


\bibitem[\protect\citeauthoryear{Feser, Chaudhuri, and Dillig}{Feser
  et~al\mbox{.}}{2015}]%
        {feser-pldi-2015}
\bibfield{author}{\bibinfo{person}{John~K. Feser}, \bibinfo{person}{Swarat
  Chaudhuri}, {and} \bibinfo{person}{Isil Dillig}.}
  \bibinfo{year}{2015}\natexlab{}.
\newblock \showarticletitle{Synthesizing Data Structure Transformations from
  Input-output Examples}. In \bibinfo{booktitle}{{\em Proceedings of the 36th
  ACM SIGPLAN Conference on Programming Language Design and Implementation
  (PLDI)}}.
\newblock


\bibitem[\protect\citeauthoryear{Finance and Accounting}{Finance and
  Accounting}{2016}]%
        {tax-formats}
Finance and Accounting \bibinfo{year}{2016}\natexlab{}.
\newblock \bibinfo{title}{Finance and Accounting}.
\newblock
  \bibinfo{howpublished}{\url{http://fileformats.archiveteam.org/wiki/Finance_and_Accounting}}.
    (\bibinfo{year}{2016}).
\newblock


\bibitem[\protect\citeauthoryear{Fisher, Walker, Zhu, and White}{Fisher
  et~al\mbox{.}}{2008}]%
        {pads:synthesis}
\bibfield{author}{\bibinfo{person}{Kathleen Fisher}, \bibinfo{person}{David
  Walker}, \bibinfo{person}{Kenny~Q. Zhu}, {and} \bibinfo{person}{Peter
  White}.} \bibinfo{year}{2008}\natexlab{}.
\newblock \showarticletitle{From Dirt to Shovels: Fully Automatic Tool
  Generation From Ad Hoc Data}.
\newblock


\bibitem[\protect\citeauthoryear{Foster, Greenwald, Moore, Pierce, and
  Schmitt}{Foster et~al\mbox{.}}{2007}]%
        {Focal2005-long}
\bibfield{author}{\bibinfo{person}{J.~Nathan Foster},
  \bibinfo{person}{Michael~B. Greenwald}, \bibinfo{person}{Jonathan~T. Moore},
  \bibinfo{person}{Benjamin~C. Pierce}, {and} \bibinfo{person}{Alan Schmitt}.}
  \bibinfo{year}{2007}\natexlab{}.
\newblock \showarticletitle{Combinators for bidirectional tree transformations:
  {A} linguistic approach to the view-update problem}.
\newblock \bibinfo{journal}{{\em ACM Transactions on Programming Languages and
  Systems\/}} \bibinfo{volume}{29}, \bibinfo{number}{3} (\bibinfo{date}{May}
  \bibinfo{year}{2007}), \bibinfo{pages}{17}.
\newblock


\bibitem[\protect\citeauthoryear{Foster, Pilkiewicz, and Pierce}{Foster
  et~al\mbox{.}}{2008}]%
        {quotientlenses}
\bibfield{author}{\bibinfo{person}{J.~Nathan Foster},
  \bibinfo{person}{Alexandre Pilkiewicz}, {and} \bibinfo{person}{Benjamin~C.
  Pierce}.} \bibinfo{year}{2008}\natexlab{}.
\newblock \showarticletitle{Quotient Lenses}.
\newblock \bibinfo{journal}{{\em SIGPLAN Not.\/}} \bibinfo{volume}{43},
  \bibinfo{number}{9} (\bibinfo{date}{Sept.} \bibinfo{year}{2008}),
  \bibinfo{pages}{383--396}.
\newblock
\showISSN{0362-1340}
\showDOI{%
\url{https://doi.org/10.1145/1411203.1411257}}


\bibitem[\protect\citeauthoryear{Frankle, Osera, Walker, and Zdancewic}{Frankle
  et~al\mbox{.}}{2015}]%
        {frankle+:popl16}
\bibfield{author}{\bibinfo{person}{Jonathan Frankle},
  \bibinfo{person}{Peter-Michael Osera}, \bibinfo{person}{David Walker}, {and}
  \bibinfo{person}{Steve Zdancewic}.} \bibinfo{year}{2015}\natexlab{}.
\newblock \bibinfo{booktitle}{{\em Example-Directed Synthesis: A Type-Theoretic
  Interpretation (extended version)}}.
\newblock \bibinfo{type}{{T}echnical {R}eport} MS-CIS-15-12.
  \bibinfo{institution}{University of Pennsylvania}.
\newblock


\bibitem[\protect\citeauthoryear{Gulwani}{Gulwani}{2011}]%
        {flashfill}
\bibfield{author}{\bibinfo{person}{Sumit Gulwani}.}
  \bibinfo{year}{2011}\natexlab{}.
\newblock \showarticletitle{Automating String Processing in Spreadsheets Using
  Input-output Examples}. In \bibinfo{booktitle}{{\em Proceedings of the 38th
  Annual ACM SIGPLAN-SIGACT Symposium on Principles of Programming Languages}}
  {\em (\bibinfo{series}{POPL '11})}. \bibinfo{publisher}{ACM}.
\newblock


\bibitem[\protect\citeauthoryear{Gvero, Kuncak, Kuraj, and Piskac}{Gvero
  et~al\mbox{.}}{2013}]%
        {gvero-pldi-2013}
\bibfield{author}{\bibinfo{person}{Tihomir Gvero}, \bibinfo{person}{Viktor
  Kuncak}, \bibinfo{person}{Ivan Kuraj}, {and} \bibinfo{person}{Ruzica
  Piskac}.} \bibinfo{year}{2013}\natexlab{}.
\newblock \showarticletitle{Complete Completion Using Types and Weights}. In
  \bibinfo{booktitle}{{\em Proceedings of the 2013 {ACM} {SIGPLAN} {C}onference
  on {P}rogramming {L}anguage {D}esign and {I}mplementation {(PLDI)}}}.
\newblock


\bibitem[\protect\citeauthoryear{Hofmann, Pierce, and Wagner}{Hofmann
  et~al\mbox{.}}{2011}]%
        {symmetric-lenses}
\bibfield{author}{\bibinfo{person}{Martin Hofmann},
  \bibinfo{person}{Benjamin~C. Pierce}, {and} \bibinfo{person}{Daniel Wagner}.}
  \bibinfo{year}{2011}\natexlab{}.
\newblock \showarticletitle{Symmetric Lenses}. In \bibinfo{booktitle}{{\em
  {ACM} {SIGPLAN--SIGACT} {S}ymposium on {P}rinciples of {P}rogramming
  {L}anguages ({POPL}), Austin, Texas}}.
\newblock


\bibitem[\protect\citeauthoryear{Hu and D'Antoni}{Hu and D'Antoni}{2017}]%
        {program-inversion-symbolic-transducers}
\bibfield{author}{\bibinfo{person}{Qinheping Hu} {and} \bibinfo{person}{Loris
  D'Antoni}.} \bibinfo{year}{2017}\natexlab{}.
\newblock \showarticletitle{Automatic Program Inversion Using Symbolic
  Transducers}. In \bibinfo{booktitle}{{\em Proceedings of the 38th ACM SIGPLAN
  Conference on Programming Language Design and Implementation}} {\em
  (\bibinfo{series}{PLDI 2017})}. \bibinfo{publisher}{ACM},
  \bibinfo{address}{New York, NY, USA}, \bibinfo{pages}{376--389}.
\newblock
\showISBNx{978-1-4503-4988-8}
\showDOI{%
\url{https://doi.org/10.1145/3062341.3062345}}


\bibitem[\protect\citeauthoryear{Kushmerick}{Kushmerick}{1997}]%
        {kushmerick-thesis}
\bibfield{author}{\bibinfo{person}{Nicholas Kushmerick}.}
  \bibinfo{year}{1997}\natexlab{}.
\newblock {\em \bibinfo{title}{Wrapper Induction for Information Extraction}}.
\newblock \bibinfo{thesistype}{Ph.D. Dissertation}. \bibinfo{address}{Seattle,
  WA, USA}.
\newblock
\showISBNx{0-591-70843-4}


\bibitem[\protect\citeauthoryear{Le and Gulwani}{Le and Gulwani}{2014}]%
        {le-pldi-2014}
\bibfield{author}{\bibinfo{person}{Vu Le} {and} \bibinfo{person}{Sumit
  Gulwani}.} \bibinfo{year}{2014}\natexlab{}.
\newblock \showarticletitle{{FlashExtract}: {A} Framework for Data Extraction
  by Examples}. In \bibinfo{booktitle}{{\em Proceedings of the 35th {ACM}
  {SIGPLAN} Conference on Programming Language Design and Implementation}} {\em
  (\bibinfo{series}{PLDI '14})}. \bibinfo{publisher}{ACM}.
\newblock


\bibitem[\protect\citeauthoryear{Lehmann}{Lehmann}{1977}]%
        {starsemiring}
\bibfield{author}{\bibinfo{person}{Daniel Lehmann}.}
  \bibinfo{year}{1977}\natexlab{}.
\newblock \showarticletitle{Algebraic structures for transitive closure}.
\newblock \bibinfo{journal}{{\em Theoretical Computer Science\/}}
  \bibinfo{volume}{4} (\bibinfo{date}{02} \bibinfo{year}{1977}),
  \bibinfo{pages}{59--76}.
\newblock
\showDOI{%
\url{https://doi.org/10.1016/0304-3975(77)90056-1}}


\bibitem[\protect\citeauthoryear{Lutterkort}{Lutterkort}{2007}]%
        {augeas}
\bibfield{author}{\bibinfo{person}{David Lutterkort}.}
  \bibinfo{year}{2007}\natexlab{}.
\newblock \bibinfo{title}{Augeas: {A} {L}inux Configuration {API}}.
\newblock   (\bibinfo{date}{Feb.} \bibinfo{year}{2007}).
\newblock
\newblock
\shownote{Available from {\tt http://augeas.net/}.}


\bibitem[\protect\citeauthoryear{Maina, Miltner, Fisher, Pierce, Walker, and
  Zdancewic}{Maina et~al\mbox{.}}{2018}]%
        {maina+:quotient-synthesis}
\bibfield{author}{\bibinfo{person}{Solomon Maina}, \bibinfo{person}{Anders
  Miltner}, \bibinfo{person}{Kathleen Fisher}, \bibinfo{person}{Benjamin
  Pierce}, \bibinfo{person}{Dave Walker}, {and} \bibinfo{person}{Steve
  Zdancewic}.} \bibinfo{year}{2018}\natexlab{}.
\newblock \showarticletitle{Synthesizing Quotient Lenses}.
\newblock
\newblock
\shownote{To appear.}


\bibitem[\protect\citeauthoryear{Miltner, Fisher, Pierce, Walker, and
  Zdancewic}{Miltner et~al\mbox{.}}{2018}]%
        {optician}
\bibfield{author}{\bibinfo{person}{Anders Miltner}, \bibinfo{person}{Kathleen
  Fisher}, \bibinfo{person}{Benjamin~C. Pierce}, \bibinfo{person}{David
  Walker}, {and} \bibinfo{person}{Steve Zdancewic}.}
  \bibinfo{year}{2018}\natexlab{}.
\newblock \showarticletitle{Synthesizing Bijective Lenses}. In
  \bibinfo{booktitle}{{\em Proceedings of the 45th Annual ACM SIGPLAN-SIGACT
  Symposium on Principles of Programming Languages}} {\em
  (\bibinfo{series}{POPL 2018})}.
\newblock


\bibitem[\protect\citeauthoryear{Osera and Zdancewic}{Osera and
  Zdancewic}{2015}]%
        {osera+:pldi15}
\bibfield{author}{\bibinfo{person}{Peter-Michael Osera} {and}
  \bibinfo{person}{Steve Zdancewic}.} \bibinfo{year}{2015}\natexlab{}.
\newblock \showarticletitle{Type-and-example-directed program synthesis}. In
  \bibinfo{booktitle}{{\em Proceedings of the 36th ACM SIGPLAN Conference on
  Programming Language Design and Implementation}}. ACM.
\newblock


\bibitem[\protect\citeauthoryear{Pereira}{Pereira}{2000}]%
        {pereira:info-theory}
\bibfield{author}{\bibinfo{person}{Fernando Pereira}.}
  \bibinfo{year}{2000}\natexlab{}.
\newblock \showarticletitle{Formal grammar and information theory: Together
  again?}
\newblock \bibinfo{journal}{{\em Philosophical transactions of the royal
  society\/}}  \bibinfo{volume}{358} (\bibinfo{year}{2000}),
  \bibinfo{pages}{1239--1253}.
\newblock


\bibitem[\protect\citeauthoryear{Ross}{Ross}{2000}]%
        {stoch-def}
\bibfield{author}{\bibinfo{person}{Brian~J. Ross}.}
  \bibinfo{year}{2000}\natexlab{}.
\newblock \showarticletitle{Probabilistic Pattern Matching and the Evolution of
  Stochastic Regular Expressions}.
\newblock \bibinfo{journal}{{\em Applied Intelligence\/}} \bibinfo{volume}{13},
  \bibinfo{number}{3} (\bibinfo{date}{Nov.} \bibinfo{year}{2000}),
  \bibinfo{pages}{285--300}.
\newblock
\showISSN{0924-669X}
\showDOI{%
\url{https://doi.org/10.1023/A:1026524328760}}


\bibitem[\protect\citeauthoryear{Scherer and R\`{e}my}{Scherer and
  R\`{e}my}{2015}]%
        {scherer-icfp-2015}
\bibfield{author}{\bibinfo{person}{Gabriel Scherer} {and}
  \bibinfo{person}{Didier R\`{e}my}.} \bibinfo{year}{2015}\natexlab{}.
\newblock \showarticletitle{Which simple types have a unique inhabitant?}. In
  \bibinfo{booktitle}{{\em Proceedings of the 18th {ACM} {SIGPLAN}
  {I}nternational {C}onference on {F}unctional {P}rogramming (ICFP)}}.
\newblock


\bibitem[\protect\citeauthoryear{Sch{\"u}rr}{Sch{\"u}rr}{1995}]%
        {tgg}
\bibfield{author}{\bibinfo{person}{Andy Sch{\"u}rr}.}
  \bibinfo{year}{1995}\natexlab{}.
\newblock \showarticletitle{Specification of graph translators with triple
  graph grammars}. In \bibinfo{booktitle}{{\em Graph-Theoretic Concepts in
  Computer Science}}, \bibfield{editor}{\bibinfo{person}{Ernst~W. Mayr},
  \bibinfo{person}{Gunther Schmidt}, {and} \bibinfo{person}{Gottfried
  Tinhofer}} (Eds.). \bibinfo{publisher}{Springer Berlin Heidelberg},
  \bibinfo{address}{Berlin, Heidelberg}, \bibinfo{pages}{151--163}.
\newblock
\showISBNx{978-3-540-49183-5}


\bibitem[\protect\citeauthoryear{Shannon}{Shannon}{1948}]%
        {Shannon1948}
\bibfield{author}{\bibinfo{person}{Claude~Elwood Shannon}.}
  \bibinfo{year}{1948}\natexlab{}.
\newblock \showarticletitle{A Mathematical Theory of Communication}.
\newblock \bibinfo{journal}{{\em The Bell System Technical Journal\/}}
  \bibinfo{volume}{27}, \bibinfo{number}{3} (\bibinfo{date}{7}
  \bibinfo{year}{1948}), \bibinfo{pages}{379--423}.
\newblock
\showDOI{%
\url{https://doi.org/10.1002/j.1538-7305.1948.tb01338.x}}


\bibitem[\protect\citeauthoryear{Voigtl\"{a}nder}{Voigtl\"{a}nder}{2009}]%
        {bidirectionalization-for-free}
\bibfield{author}{\bibinfo{person}{Janis Voigtl\"{a}nder}.}
  \bibinfo{year}{2009}\natexlab{}.
\newblock \showarticletitle{Bidirectionalization for Free! (Pearl)}. In
  \bibinfo{booktitle}{{\em Proceedings of the 36th Annual ACM SIGPLAN-SIGACT
  Symposium on Principles of Programming Languages}} {\em
  (\bibinfo{series}{POPL '09})}. \bibinfo{publisher}{ACM},
  \bibinfo{address}{New York, NY, USA}, \bibinfo{pages}{165--176}.
\newblock
\showISBNx{978-1-60558-379-2}
\showDOI{%
\url{https://doi.org/10.1145/1480881.1480904}}


\end{thebibliography}

\ifappendices
\appendix

\onecolumn

\section{Symmetric Lenses full Detail}
\label{sec:appendixlenses}

Here we provide fullly detailed information on symmetric lenses -- full typing
rules and full semantics.

\begin{centermath}
\begin{array}[b]{l@{\qquad \; \qquad}l}
\begin{array}[b]{c}
\quad\\
  \inferrule*
  {
  }
  {
    \IdentityLensOf{\BRegex} \OfType \BRegex \Leftrightarrow \BRegex
  }
  \\
  \quad
  \end{array}
&
\begin{array}[b]{r@{\ }c@{\ }l}
    \CreateR{} \App s & = & s\\
    \CreateL{} \App s & = & s\\
    \PutR{} \App s \App t & = & s\\
    \PutL{} \App s \App t & = & s\\
  \end{array}
\end{array}
\end{centermath}
Note that the identity lens ignores the second argument in the put functions.
Because the two formats are fully synchronized, no knowledge of the prior data
is needed.

\begin{centermath}
\begin{array}[b]{l@{\qquad \; \qquad}l}
  \inferrule*
  {
    \String \in \LanguageOf{\BRegex}\\
    \StringAlt \in \LanguageOf{\BRegexAlt}
  }
  {
    \DisconnectOf{\BRegex}{\BRegexAlt}{\String}{\StringAlt}
    \OfType \BRegex \Leftrightarrow \BRegexAlt
  }
&
  \begin{array}{@{}r@{\ }c@{\ }l@{}}
    $\CreateR{} \App \String''''$ & = & $\StringAlt$\\
    $\CreateL{} \App \StringAlt'''''$ & = & $\String$\\
    $\PutR{} \App \String''' \App \StringAlt'$ & = & $\StringAlt'$\\
    $\PutL{} \App \StringAlt' \App \String'$ & = & $\String'$
  \end{array}
\end{array}
\end{centermath}

Just as the identity lens ignores the second argument in puts, disconnect lenses
ignore the first in both puts and creates.  The data is unsynchronized in these
two formats, information from one format doesn't impact the other.

\begin{centermath}
\begin{tabular}[b]{l@{\qquad \; \qquad}l}
$
  \inferrule*
  {
    \Lens_1 \OfType \BRegex_1 \Leftrightarrow \BRegexAlt_1\\
    \Lens_2 \OfType \BRegex_2 \Leftrightarrow \BRegexAlt_2
  }
  {
    \ConcatLensOf{\Lens_1}{\Lens_2} \OfType \BRegex_1 \Concat \BRegex_2
    \Leftrightarrow
    \BRegexAlt_1 \Concat \BRegexAlt_2
  }
$
&
$
  \inferrule*
  {
    \Lens_1 \OfType \BRegex_1 \Leftrightarrow \BRegexAlt_1\\
    \Lens_2 \OfType \BRegex_2 \Leftrightarrow \BRegexAlt_2
  }
  {
    \SwapLensOf{\Lens_1}{\Lens_2} \OfType \BRegex_1 \Concat \BRegex_2
    \Leftrightarrow
    \BRegexAlt_2 \Concat \BRegexAlt_1
  }
$
\end{tabular}
\end{centermath}
\begin{center}
  \begin{tabular}{@{}r@{\ }c@{\ }l@{}}
    $\CreateR{} \App \String_1 \String_2$ & = & $(\Lens_1.\CreateROf{\String_1})(\Lens_2.\CreateROf{\String_2})$\\
    $\CreateL{} \App \StringAlt_1 \StringAlt_2$ & = & $(\Lens_1.\CreateLOf{\StringAlt_1})(\Lens_2.\CreateLOf{\StringAlt_2})$\\
    $\PutR{} \App (\String_1\String_2) \App (\StringAlt_1\StringAlt_2)$ & = & $(\Lens_1.\PutROf{\String_1}{\StringAlt_1})(\Lens_2.\PutROf{\String_2}{\StringAlt_2})$\\
    $\PutL{} \App (\StringAlt_1\StringAlt_2) \App (\String_1\String_2)$ & = & $(\Lens_1.\PutLOf{\StringAlt_1}{\String_1})(\Lens_2.\PutLOf{\StringAlt_2}{\String_2})$\\
  \end{tabular}
\end{center}

Concat is similar to concatenation in existing string lens languages like
Boomerang.  For such terms, we do not provide the semantics, and merely refer readers to existing work. The swap combinator is similar to concat, though the second regular expression
is swapped.
\begin{center}
  \begin{tabular}{@{}r@{\ }c@{\ }l@{}}
    $\CreateR{} \App \String_1\String_2$ & = & $(\Lens_2.\CreateROf{\String_2})(\Lens_1.\CreateROf{\String_1})$\\
    $\CreateL{} \App \StringAlt_2\StringAlt_1$ & = & $(\Lens_1.\CreateLOf{\StringAlt_1})(\Lens_2.\CreateLOf{\StringAlt_2})$\\
    $\PutR{} \App (\String_1\String_2) \App (\StringAlt_2\StringAlt_1)$ & = & $(\Lens_2.\PutROf{\String_2}{\StringAlt_2})(\Lens_1.\PutROf{\String_1}{\StringAlt_1})$\\
    $\PutL{} \App (\StringAlt_2\StringAlt_1) \App (\String_1\String_2)$ & = & $(\Lens_1.\PutLOf{\StringAlt_1}{\String_1})(\Lens_2.\PutLOf{\StringAlt_2}{\String_2})$\\
  \end{tabular}
\end{center}
\begin{centermath}
\begin{tabular}[b]{l@{\qquad}l}
$
  \inferrule*
  {
    \Lens_1 \OfType \BRegex_1 \Leftrightarrow \BRegexAlt_1\\
    \Lens_2 \OfType \BRegex_2 \Leftrightarrow \BRegexAlt_2
  }
  {
    \OrLensOf{\Lens_1}{\Lens_2} \OfType
    \RegexOr{\BRegex_1}{\BRegex_2}
    \Leftrightarrow
    \RegexOr{\BRegexAlt_1}{\BRegexAlt_2}
  }
  $
&
  \begin{tabular}{@{}r@{\ }c@{\ }l@{}}
    $\CreateR{} \App \String$
    & =
    & $\begin{cases*}
      \Lens_1.\CreateROf{\String} & if $\String\in\LanguageOf{\BRegex_1}$\\
      \Lens_2.\CreateROf{\String} & if $\String\in\LanguageOf{\BRegex_2}$
      \end{cases*}$\\
    
    $\CreateL{} \App \StringAlt$
    & =
    & $\begin{cases*}
      \Lens_1.\CreateLOf{\StringAlt} & if $\StringAlt\in\LanguageOf{\BRegexAlt_1}$\\
      \Lens_2.\CreateLOf{\StringAlt} & if $\StringAlt\in\LanguageOf{\BRegexAlt_2}$
      \end{cases*}$\\
    
    $\PutR{} \App \String \App \StringAlt$
    & =
    & $\begin{cases*}
        \Lens_1.\PutROf{\String}{\StringAlt} & if $\String\in\LanguageOf{\BRegex_1} \BooleanAnd \StringAlt\in\LanguageOf{\BRegexAlt_1}$\\
        \Lens_2.\PutROf{\String}{\StringAlt} & if $\String\in\LanguageOf{\BRegex_2} \BooleanAnd \StringAlt\in\LanguageOf{\BRegexAlt_2}$\\
        \Lens_1.\CreateROf{\String} & if $\String\in\LanguageOf{\BRegex_1} \BooleanAnd \StringAlt\in\LanguageOf{\BRegexAlt_2}$\\
        \Lens_2.\CreateROf{\String} & if $\String\in\LanguageOf{\BRegex_2} \BooleanAnd \StringAlt\in\LanguageOf{\BRegexAlt_1}$
      \end{cases*}$\\
    
    $\PutL{} \App \StringAlt \App \String$
    & =
    & $\begin{cases*}
        \Lens_1.\PutLOf{\StringAlt}{\String} & if $\StringAlt\in\LanguageOf{\BRegexAlt_1} \BooleanAnd \String\in\LanguageOf{\BRegex_1}$\\
        \Lens_2.\PutLOf{\StringAlt}{\String} & if $\StringAlt\in\LanguageOf{\BRegexAlt_2} \BooleanAnd \String\in\LanguageOf{\BRegex_2}$\\
        \Lens_1.\CreateLOf{\StringAlt} & if $\StringAlt\in\LanguageOf{\BRegexAlt_1} \BooleanAnd \String\in\LanguageOf{\BRegex_2}$\\
        \Lens_2.\CreateLOf{\String} & if $\StringAlt\in\LanguageOf{\BRegexAlt_2} \BooleanAnd \String\in\LanguageOf{\BRegex_1}$
      \end{cases*}$\\
  \end{tabular}
\end{tabular}
\end{centermath}
The \OrLens lens deals with data that can come in one form or another. If the
data gets changed from one format to the other, information in the old format is
lost.
\begin{center}
  \begin{tabular}{@{}r@{\ }c@{\ }l@{}}
    $\CreateR{} \App \String_1\ldots\String_n$
    & =
    & $(\Lens.\CreateROf{\String_1})\ldots(\Lens.\CreateROf{\String_n})$\\
    
    $\CreateL{} \App \StringAlt_1\ldots\StringAlt_n$
    & =
    & $(\Lens.\CreateLOf{\StringAlt_1})\ldots(\Lens.\CreateLOf{\StringAlt_n})$\\
    
    $\PutR{} \App (\String_1\ldots\String_n) \App (\StringAlt_1\ldots\StringAlt_m)$
    & =
    & $\StringAlt_1'\ldots\StringAlt_n'$ where $\StringAlt_i' =
      \begin{cases*}
        \Lens.\PutROf{\String_i}{\StringAlt_i} & if $i \leq m$\\
        \Lens.\CreateROf{\String_i} & otherwise
      \end{cases*}$\\
    $\PutL{} \App (\StringAlt_1\ldots\StringAlt_m) \App (\String_1\ldots\String_n)$
    & =
    & $\String_1'\ldots\String_n'$ where $\String_i' =
      \begin{cases*}
        \Lens.\PutROf{\StringAlt_i}{\String_i} & if $i \leq n$\\
        \Lens.\CreateROf{\StringAlt_i} & otherwise
      \end{cases*}$
  \end{tabular}
\end{center}
\begin{centermath}
\begin{tabular}[b]{l@{\qquad}l}
$
  \centering
  \inferrule*
  {
    \Lens_1 \OfType \BRegex_1 \Leftrightarrow \BRegexAlt\\
    \Lens_2 \OfType \BRegex_2 \Leftrightarrow \BRegexAlt
  }
  {
    \MergeROf{\Lens_1}{\Lens_2} \OfType
    \RegexOr{\BRegex_1}{\BRegex_2}
    \Leftrightarrow
    \BRegexAlt
  }
$
&
  \begin{tabular}{@{}r@{\ }c@{\ }l@{}}
    $\CreateR{} \App \String$
    & =
    & $\begin{cases*}
      \Lens_1.\CreateROf{\String} & if $\String\in\LanguageOf{\BRegex_1}$\\
      \Lens_2.\CreateROf{\String} & if $\String\in\LanguageOf{\BRegex_2}$
      \end{cases*}$\\
    
    $\CreateL{} \App \StringAlt$
    & =
    & $\Lens_1.\CreateLOf{\StringAlt}$\\
    
    $\PutR{} \App \String \App \StringAlt$
    & =
    & $\begin{cases*}
      \Lens_1.\PutROf{\String}{\StringAlt} & if $\String\in\LanguageOf{\BRegex_1}$\\
      \Lens_2.\PutROf{\String}{\StringAlt} & if $\String\in\LanguageOf{\BRegex_2}$
    \end{cases*}$\\
    
    $\PutL{} \App \StringAlt \App \String$
    & =
    & $\begin{cases*}
        \Lens_1.\PutLOf{\StringAlt}{\String} & if $\String\in\LanguageOf{\BRegex_1}$\\
        \Lens_2.\PutLOf{\StringAlt}{\String} & if $\String\in\LanguageOf{\BRegex_2}$
      \end{cases*}$\\
  \end{tabular}
\end{tabular}
\end{centermath}

The \MergeR lens is interesting because it merges data where one data can be in
two formats, and one data has only one format. In previous
work~\cite{boomerang}, this was combined into \OrLens{}, where
\OrLens{} could have ambiguous types, but we find it more clear to have explicit
merge operators: it is easier to see what lens the synthesis algorithm is
creating.

\begin{centermath}
\begin{tabular}[b]{l@{\qquad}l}
$
  \inferrule*
  {
    \Lens_1 \OfType \BRegex \Leftrightarrow \BRegexAlt_1\\
    \Lens_2 \OfType \BRegex \Leftrightarrow \BRegexAlt_2
  }
  {
    \MergeLOf{\Lens_1}{\Lens_2} \OfType
    \BRegex
    \Leftrightarrow
    \RegexOr{\BRegexAlt_1}{\BRegexAlt_2}
  }
$
&
  \begin{tabular}{@{}r@{\ }c@{\ }l@{}}
    $\CreateR{} \App \String$
    & =
    & $\Lens_1.\CreateROf{\String}$\\
    
    $\CreateL{} \App \StringAlt$
    & =
    & $\begin{cases*}
      \Lens_1.\CreateLOf{\StringAlt} & if $\StringAlt\in\LanguageOf{\BRegexAlt_1}$\\
      \Lens_2.\CreateLOf{\StringAlt} & if $\StringAlt\in\LanguageOf{\BRegexAlt_2}$
      \end{cases*}$\\
    
    $\PutR{} \App \String \App \StringAlt$
    & =
    & $\begin{cases*}
      \Lens_1.\PutROf{\String}{\StringAlt} & if $\StringAlt\in\LanguageOf{\BRegexAlt_1}$\\
      \Lens_2.\PutROf{\String}{\StringAlt} & if $\StringAlt\in\LanguageOf{\BRegexAlt_2}$
    \end{cases*}$\\
    
    $\PutL{} \App \StringAlt \App \String$
    & =
    & $\begin{cases*}
        \Lens_1.\PutLOf{\StringAlt}{\String} & if $\StringAlt\in\LanguageOf{\BRegexAlt_1}$\\
        \Lens_2.\PutLOf{\StringAlt}{\String} & if $\StringAlt\in\LanguageOf{\BRegexAlt_2}$
      \end{cases*}$\\
  \end{tabular}
\end{tabular}
\end{centermath}
The \MergeL lens is symmetric to \MergeR.

\begin{centermath}
\begin{tabular}[b]{l@{\qquad}l}
  $
  \inferrule*
  {
    \Lens_1 \OfType \BRegex \Leftrightarrow \BRegexAlt\\
    \Lens_2 \OfType \BRegexAlt \Leftrightarrow \BRegexAltAlt\\
  }
  {
    \ComposeLensOf{\Lens_1}{\Lens_2} \OfType
    \BRegex \Leftrightarrow \BRegexAltAlt
  }
$
&
  \begin{tabular}{@{}r@{\ }c@{\ }l@{}}
    $\CreateR{} \App \String$ & = & $\Lens_2.\CreateROf{(\Lens_1.\CreateROf{\String})}$\\
    $\CreateL{} \App \StringAlt$ & = & $\Lens_1.\CreateLOf{(\Lens_2.\CreateLOf{\StringAlt})}$\\
    $\PutR{} \App \String \App \StringAltAlt$ & = & $\Lens_2.\PutROf{(\Lens_1.\PutROf{\String}{(\Lens_2.\CreateLOf{\StringAltAlt})})}{\StringAltAlt}$\\
    $\PutL{} \App \StringAltAlt \App \String$ & = & $\Lens_1.\PutLOf{(\Lens_2.\PutLOf{\StringAltAlt}{(\Lens_2.\CreateROf{\String})})}{\String}$
  \end{tabular}
\end{tabular}
\end{centermath}
Composing is interesting in the put functions. Because puts require intermediary
data, we recreate that intermediary data with creates.

\begin{centermath}
\begin{tabular}[b]{l@{\qquad}l}
$
\inferrule*
  {
    \Lens \OfType \BRegex \Leftrightarrow \BRegexAlt
  }
  {
    \IterateLensOf{\Lens} \OfType
    \StarOf{\BRegex}
    \Leftrightarrow
    \StarOf{\BRegexAlt}
  }
  $
  &
  $
  \inferrule*
  {
    \Lens \OfType \BRegex \Leftrightarrow \BRegexAlt
  }
  {
    \InvertOf{\Lens} \OfType \BRegexAlt \Leftrightarrow \BRegex
  }
  $
\end{tabular}
\end{centermath}
\begin{center}
  \begin{tabular}{@{}r@{\ }c@{\ }l@{}}
    $\CreateR{} \App \StringAlt$ & = & $\Lens.\CreateLOf{\StringAlt}$\\
    $\CreateL{} \App \String$ & = & $\Lens.\CreateROf{\String}$\\
    $\PutR{} \App \StringAlt \App \String$ & = & $\Lens.\PutLOf{\StringAlt}{\String}$\\
    $\PutL{} \App \String \App \StringAlt$ & = & $\Lens.\PutROf{\String}{\StringAlt}$
  \end{tabular}
\end{center}
The \IterateLens lens deals with iterated data, while inverting reverses the direction of a lens: creating on the right becomes creating on the left and vice versa, and putting on the right becomes putting on the left and vice versa. The invert combonator is particularly useful when chaining many compositions together, as it can be used to align the central types.
\[
  \centering
  \inferrule*
  {
    \Lens \OfType \BRegex \Leftrightarrow \BRegexAlt\\
    \BRegex \SSREquiv \BRegex'\\
    \BRegexAlt \SSREquiv \BRegexAlt'
  }
  {
    \Lens \OfType \BRegex' \Leftrightarrow \BRegexAlt'
  }
\]

Type equivalence enables a lens of type $S \Leftrightarrow T$ to be used as a
lens of type $S' \Leftrightarrow T'$ if $S$ equivalent to $S'$ and $T$ is
equivalent to $T'$. Type equivalence is useful both for addressing type
annotations, and for making well-typed compositions.

\section{Symmetric DNF Lenses: Syntax, Cost, Typing, and Semantics}
\label{sec:appendixdnf}

Here we give full details on symmetric DNF lenses, including syntax, cost,
typing, and semantics.

\paragraph*{Syntax}
Roughly speaking, symmetric DNF lenses are an n-ary union of symmetric sequence
lenses, which are, in turn, an n-ary sequence of atomic lenses. Atomic lenses
are iterations of symmetric DNF lenses. In the following syntax, the
nonterminals $i$, $j$, $p$, $q$, $r$, $c$, and $d$ all represent natural
numbers. The nonterminals $\String$ and $\StringAlt$ represent strings.
\begin{center}
  \begin{tabular}{@{}r@{\ }c@{}l@{}}
    \SAtomLens{} & \GEq{} & $\StarOf{\SDNFLens}$ \\
    \SSQLens{} & \GEq{} & $(\SSQLensOf{(i_1,j_1,\SAtomLens_1)\SeqLSep
                          \ldots\SeqLSep
                          (i_p,j_p,\SAtomLens_p)}
                          ,\ListOf{(\String_1,\Atom_1);\ldots;(\String_q,\Atom_q)}
                          ,\ListOf{(\StringAlt_1,\AtomAlt_1);\ldots;(\StringAlt_r,\AtomAlt_r)})$ \\
    \SDNFLens{} & \GEq{} & $(\SDNFLensOf{(i_1,j_1,\SSQLens_1)\DNFLSep
                           \ldots\DNFLSep
                           (i_p,j_p,\SSQLens_p)}
                           ,\ListOf{c_1;\ldots;c_q}
                           ,\ListOf{d_1;\ldots;d_r})$ \\
  \end{tabular}
\end{center}

\paragraph*{Information-Theoretic Cost Metric}
Much like we defined a syntactic means to find the expected number of bits to
recover string from a synchronized string in another format for simple symmetric
string lenses, we do the same for symmetric DNF lenses. Unlike simple symmetric
string lenses, DNF lenses do not have composition, so the entropy can be defined
over all terms. $\LEntropyOf{\DNFRegex \Given \SDNFLens, \DNFRegexAlt}$ is the
expected number of bits required to recover a string in $\DNFRegex$ from a
synchronized string in $\DNFRegexAlt$.

\[
  \begin{array}{rl}
    & \LEntropyOf{\DNFOf{\Sequence_1 \DNFSep \ldots \DNFSep \Sequence_n}\\
    & \hspace*{1.21em}\Given
      (\SDNFLensOf{(i_1,j_1,\SSQLens_1)\DNFLSep
      \ldots\DNFLSep
      (i_p,j_p,\SSQLens_p)}
      ,\ListOf{c_1;\ldots;c_q}
      ,\ListOf{d_1;\ldots;d_r})\\
    & \hspace*{1.21em},~\DNFOf{\SequenceAlt_1 \DNFSep \ldots \DNFSep \SequenceAlt_m}}\\
    = & \frac{\Sigma_{j=1}^m(H_j)}{m}\\
    & \text{where } H_j = (0,\Sigma_{\SetOf{k\SuchThat j_k = j}}
      \LEntropyOf{\Sequence_{i_k}\Given\SequenceLens_k,\SequenceAlt_j})
  \end{array}
\]

This entropy calculation for symmetric DNF lenses requires a similar entropy
calculation for symmetric sequence lenses and symmetric atom lenses.

\[
  \begin{array}{rl}
      & \REntropyOf{
        \SequenceOf{\String_0,\Atom_1,\ldots,\Atom_n,\String_n}\\
      & \hspace*{1em}\Given
        (\SSQLensOf{(i_1,j_1,\SAtomLens_1)\SeqLSep
        \ldots\SeqLSep
        (i_p,j_p,\SAtomLens_p)}
        ,\ListOf{(k_1,\String_1);\ldots;(k_q,\String_q)}
        ,\ListOf{(l_1,\StringAlt_1);\ldots;(l_r,\StringAlt_r)})\\
      & \hspace*{1.21em},~
        \SequenceOf{\StringAlt_0,\AtomAlt_1,\ldots,\AtomAlt_m,\StringAlt_m}}\\
      =
      & \Sigma_{x=1}^p\REntropyOf{\Atom_{i_x} \Given \AtomAlt_{i_x},\AtomLens_x} +
        \Sigma_{x=1}^q\REntropyOf{\Atom_{k_x}}
  \end{array}
\]

The entropy calculation for symmetric sequence lenses requires a similar entropy
calculation for symmetric atom lenses, which in turn relies on the entropy
calculation for symmetric DNF lenses.

\[\begin{array}{rl}
      & \REntropyOf{\PRegexStar{\DNFRegexAlt}{\ProbabilityAlt} \Given \PRegexStar{\DNFRegex}{\Probability},\StarOf{\SDNFLens}}\\
      =
      & \frac{\Probability}{1-\Probability}\REntropyOf{\DNFRegexAlt\Given\DNFRegex,\SDNFLens}
  \end{array}
\]
The expected number of bits required to recover a string in $\DNFRegexAlt$ from
a synchronized string in $\DNFRegex$, $\REntropyOf{\DNFRegexAlt \Given
  \SDNFLens, \DNFRegex}$ is defined symmetrically.

The typing judgment is a 3-ary relation over a single DNF lens, and two DNF
regular expressions. If $\SDNFLens \OfType \DNFRegex \Leftrightarrow
\DNFRegexAlt$, then the $\SDNFLens.\CreateR$, $\SDNFLens.\CreateL$,
$\SDNFLens.\PutR$, and $\SDNFLens.\PutL$ functions form a symmetric lens.

The typing judgement has 3 components. The first is the sublens components,
confirming that the sequence lenses the DNF lens is comprised of are all
well-typed. The second guarantees that if each sequence on the left has a
sequence lens that can be used for \CreateR{}s and \PutR{}s. The last guarantees
the same for sequences on the right, with \CreateL{} and \PutL{}.

\[
  \inferrule*
  {
    \SSQLens_1 \OfType \Sequence_{i_1} \Leftrightarrow \SequenceAlt_{j_1}\\
    \ldots\\
    \SSQLens_p \OfType \Sequence_{i_p} \Leftrightarrow \SequenceAlt_{j_p}\\\\
    i_{c_1} = 1\\
    \ldots\\
    i_{c_q} = q\\\\
    j_{d_1} = 1\\
    \ldots\\
    j_{d_r} = r
  }
  {
    (\SDNFLensOf{(i_1,j_1,\SSQLens_1)\DNFLSep
      \ldots\DNFLSep
      (i_p,j_p,\SSQLens_p)}
    ,\ListOf{c_1;\ldots;c_q}
    ,\ListOf{d_1;\ldots;d_r})
    \OfType\\\\
    \DNFOf{(\Sequence_1,\Probability_1) \DNFSep \ldots \DNFSep (\Sequence_q,\Probability_q)}
    \Leftrightarrow
    \DNFOf{(\SequenceAlt_1,\ProbabilityAlt_1) \DNFSep \ldots \DNFSep (\SequenceAlt_r,\ProbabilityAlt_q)}
  }
\]

The \CreateR{} function looks for the sequence the provided string matches. If
the string matches sequence $\Sequence_x$, then the lens will look in the create
list at position $x$ to find which sequence lens to use. Then, the sequence lens
transforms the given string using that sequence lens.

The \PutR{} function finds what pairs sequences the input source and view
strings match.  If there that pair of sequences have a sequence lens between
them, then the DNF lens merely performs that sequence lens on the provided
strings.  If there isn't a pair of sequence lenses between them, then \CreateR{}
is performed on the source, with the view forgotten.

The \CreateL{} and \PutL{} functions are defined symmetrically.

\begin{tabular}{@{}r@{\ }c@{\ }l@{}}
  $\CreateR{} \App s$ & = & $\SSQLens_{c_x}.\CreateR{} \App s$ if $s \in \Sequence_x$\\
  $\CreateL{} \App v$ & = & $\SSQLens_{d_y}.\CreateL{} \App v$ if $v \in \SequenceAlt_y$\\
  $\PutR{} \App s \App v$ & = &
                               $\begin{cases*}
                                 \SSQLens_x.\PutR{} \App s \App v & if $s \in \Sequence_{i_x}$ and $v \in \SequenceAlt_{j_x}$\\
                                 \CreateR{} \App s & if $\nexists x.$ $s \in \Sequence_{i_x}$ and $v \in \SequenceAlt_{j_x}$
                               \end{cases*}$\\
  $\PutL{} \App v \App s$ & = &
                               $\begin{cases*}
                                 \SSQLens_y.\PutL{} \App v \App s & if $v \in \SequenceAlt_{j_y}$ and $s \in \Sequence_{i_y}$\\
                                 \CreateL{} \App v & if $\nexists y.$ $v \in \SequenceAlt_{j_y}$ and $s \in \Sequence_{i_y}$
                               \end{cases*}$
\end{tabular}

\paragraph*{Symmetric Sequence Lenses: Typing and Semantics}
The typing judgment is a 3-ary relation over a single Sequence lens, and two
sequences. If $\SSQLens \OfType \Sequence \Leftrightarrow \SequenceAlt$, then
the $\SSQLens.\CreateR$, $\SSQLens.\CreateL$, $\SSQLens.\PutR$, and
$\SSQLens.\PutL$ functions form a symmetric lens.

The typing judgement has 4 components. The first is the sublens components,
confirming that the atom lenses the sequence lens is comprised of all are
well-typed. The second guarantees that if each string that will be used for
\CreateR{}s are members of the correct atoms.  The third guarantees
the same for strings and atoms on the right, with \CreateL.  The last guarantees
that each atom is mapped by at most one atom lens.

\[
  \inferrule* {
    \SAtomLens_1 \OfType \Atom_{i_1} \Leftrightarrow \AtomAlt_{j_1}\\
    \ldots\\
    \SAtomLens_p \OfType \Atom_{i_p} \Leftrightarrow \AtomAlt_{j_p}\\\\
    \String_{c_1} \in \Atom_{c_1}\\
    \ldots\\
    \String_{c_{q'}} \in \Atom_{c_{q'}}\\\\
    \StringAlt_{d_1} \in \AtomAlt_{d_1}\\
    \ldots\\
    \StringAlt_{d_{r'}} \in \AtomAlt_{d_{r'}}\\\\
    i_x = i_y \BooleanImplies x = y\\
    j_x = j_y \BooleanImplies x = y } {
    (\SSQLensOf{(i_1,j_1,\SAtomLens_1)\SeqLSep \ldots\SeqLSep
      (i_p,j_p,\SAtomLens_p)}
    ,\ListOf{(\String_{c_1},\Atom_{c_1});\ldots;(\String_{c_{q'}},\Atom_{c_{q'}})}
    ,\ListOf{(\StringAlt_{d_1},\Atom_{d_1});\ldots;(\StringAlt_{d_r},\Atom_{d_r})})\\
    \OfType \SequenceOf{\String_0'''' \SeqSep \Atom_1 \SeqSep \ldots \SeqSep
      \Atom_{q'} \SeqSep \String_{q}'} \Leftrightarrow \SequenceOf{\StringAlt_0'
      \SeqSep \AtomAlt_1 \SeqSep \ldots \SeqSep \AtomAlt_r \SeqSep
      \StringAlt_{r}'} }
\]

For each component of the string matching an atom, the \CreateR function looks
for the atom lens that maps on the atom. If there is such an atom lens,
\SAtomLens, then that the sequence lens puts the provided string into the
default string for the target atom. If there is no such atom lens, then the
sequence lens merely uses the default string.

For each component of the string matching an atom, the \PutR function looks
for the atom lens that maps on the atom. If there is such an atom lens,
\SAtomLens, then that the atom lens puts the provided string of the source atom into the
string of the target atom. If there is no such atom lens, then the
sequence lens merely recovers the target's string.

The \CreateL{} and \PutL{} functions are defined symmetrically.

\begin{tabular}{@{}r@{\ }c@{\ }l@{}}
  $\CreateR{} \App \String_0'\Concat \String_1'' \Concat \ldots \Concat \String_q'' \Concat \String_q'$
  & = 
  & $\StringAlt_0' \Concat \StringAlt_1'' \Concat \ldots \Concat
    \StringAlt_r'' \Concat \StringAlt_r'$\\
  & & where $\StringAlt_y'' =
    \begin{cases*}
      \SAtomLens_k.\PutROf{\String_{i_k}''}{\StringAlt_y} & if $j_k = y$\\
      \StringAlt_y & if $\nexists k. j_k = y$\\
    \end{cases*}$\\
  $\CreateL{} \App \StringAlt_0'\Concat \StringAlt_1'' \Concat \ldots \Concat \StringAlt_q''
  \Concat \StringAlt_q'$
  & = 
  & $\String_0' \Concat \String_1'' \Concat \ldots \Concat
    \String_r'' \Concat \String_r'$\\
  & & where $\String_x'' =
    \begin{cases*}
      \SAtomLens_k.\PutROf{\StringAlt_{j_k}''}{\String_x} & if $i_k = x$\\
      \String_x & if $\nexists k. i_k = x$\\
    \end{cases*}$\\
  $\PutR{} \App (\String_0'\Concat \String_1'' \Concat \ldots \Concat \String_q'' \Concat \String_q') \App (\StringAlt_0'\Concat \StringAlt_1'' \Concat \ldots \Concat \StringAlt_q'' \Concat \StringAlt_q')$
  & =
  & $\StringAlt_0'\Concat \StringAlt_1''' \Concat \ldots \Concat \StringAlt_q''' \Concat \StringAlt_q'$\\
  & & where
      $t_y''' =
      \begin{cases*}
        \SAtomLens_k.\PutROf{\String_{i_k}''}{\StringAlt_{y}''} & if $j_k = y$\\
        \StringAlt_y'' & if $\nexists k. j_k = y$\\
      \end{cases*}$\\
  $\PutL{} \App (\StringAlt_0'\Concat \StringAlt_1'' \Concat \ldots \Concat \StringAlt_q'') \Concat (\StringAlt_q' \App \String_0'\Concat \String_1'' \Concat \ldots \Concat \String_q'' \Concat \String_q')$
  & =
  & $\String_0'\Concat \String_1''' \Concat \ldots \Concat \String_r''' \Concat \String_r'$\\
  & & where
      $s_x''' =
      \begin{cases*}
        \SAtomLens_k.\PutLOf{\StringAlt_{j_k}''}{\String_{x}''} & if $i_k = x$\\
        \String_x'' & if $\nexists k. i_k = x$\\
      \end{cases*}$\\
\end{tabular}

\paragraph*{Symmetric Atom Lenses: Typing and Semantics}
The typing judgment is a 3-ary relation over a single atom lens, and two
atoms. If $\SAtomLens \OfType \Atom \Leftrightarrow \AtomAlt$, then
the $\SAtomLens.\CreateR$, $\SAtomLens.\CreateL$, $\SAtomLens.\PutR$, and
$\SAtomLens.\PutL$ functions form a symmetric lens.

The typing judgement just confirms that the DNF lens that comprises the
sequence lens is also well typed.

\[
  \inferrule*
  {
    \SDNFLens \OfType \DNFRegex \Leftrightarrow \DNFRegexAlt
  }
  {
    \StarOf{\SDNFLens}
    \OfType \PRegexStar{\DNFRegex}{\Probability}
    \Leftrightarrow \PRegexStar{\DNFRegexAlt}{\ProbabilityAlt}
  }
\]

For each component of the string matching an atom, the \CreateR function looks
for the atom lens that maps on the atom. If there is such an atom lens,
\SAtomLens, then that the sequence lens puts the provided string into the
default string for the target atom. If there is no such atom lens, then the
sequence lens merely uses the default string.

For each component of the string matching an atom, the \PutR function looks
for the atom lens that maps on the atom. If there is such an atom lens,
\SAtomLens, then that the atom lens puts the provided string of the source atom into the
string of the target atom. If there is no such atom lens, then the
sequence lens merely recovers the target's string.

The \CreateL{} and \PutL{} functions are defined symmetrically.

\begin{tabular}{@{}r@{\ }c@{\ }l@{}}
  $\CreateR{} \App \String_0 \Concat \ldots \Concat \String_n$
  & = 
  & $\StringAlt_1 \Concat \ldots \Concat \StringAlt_n$
  where $\StringAlt_i = \SDNFLens.\CreateR \App \String_i$\\
  $\CreateL{} \App \StringAlt_0 \Concat \ldots \Concat \StringAlt_m$
  & = 
  & $\String_1 \Concat \ldots \Concat \String_m$
  where $\String_i = \SDNFLens.\CreateL \App \StringAlt_i$\\
  $\PutR{} \App (\String_0 \Concat \ldots \Concat \String_n) \App
  (\StringAlt_1 \Concat \ldots \Concat \StringAlt_m)$
  & = 
  & $\StringAlt_1' \Concat \ldots \Concat \StringAlt_n'$
    where $\StringAlt_i' =
    \begin{cases*}
      \SDNFLens.\PutR \App \String_i \App \StringAlt_i & if $i \leq m$\\
      \SDNFLens.\CreateR \App \String_i & otherwise
    \end{cases*}$\\
  $\PutL{} \App (\StringAlt_1 \Concat \ldots \Concat \StringAlt_m) \App
  (\String_0 \Concat \ldots \Concat \String_n)$
  & = 
  & $\String_1' \Concat \ldots \Concat \String_m'$
    where $\String_i' =
    \begin{cases*}
      \SDNFLens.\PutL \App \StringAlt_i \App \String_i & if $i \leq n$\\
      \SDNFLens.\CreateL \App \StringAlt_i & otherwise
    \end{cases*}$\\
\end{tabular}

\section{Forgetful Symmetric Lenses}
\label{sec:appendixforget}

\begin{property}[Starting Forgetfulness RL]
  \label{prop:forget-rl}
  Let $\Lens$ be a forgetful symmetric lens.  If $(x_1',c_1') = \Lens.\PutL \App$
  $(y,\Snd \App (\Lens.\PutR \App (x,c_1)))$, and
  $(x_2',c_2') = \Lens.\PutL \App
  (y,\Snd \App (\Lens.\PutR \App (x,c_2)))$, then $x_1' = x_2'$.
\end{property}
\begin{proof}
  By \ForgetfulRL, $c_1' = c_2'$. We know $(y,c_1') = \Lens.\PutR \App
  (x_1',c_1')$ and $(y,c_1') = \Lens.\PutR \App (x_2',c_1')$ by \PutLR. By
  \PutRL, we know $(x_1' = \Lens.\PutL \App (y,c_1'))$ and $(x_2' = \Lens.\PutL
  \App (y,c_1')))$. Therefore, by transitivity of equality, $x_1' = x_2'$.
\end{proof}

\begin{property}[Starting Forgetfulness LR]
  \label{prop:forget-lr}
  Let $\Lens$ be a forgetful symmetric lens.  If $(y_1',c_1') = \Lens.\PutR \App$
  $(x,\Snd \App (\Lens.\PutL \App (y,c_1)))$, and
  $(x_2',c_2') = \Lens.\PutL \App
  (x,\Snd \App (\Lens.\PutR \App (y,c_2)))$,
  then $x_1' = x_2'$.
\end{property}
\begin{proof}
  Symmetric to Starting Forgetfulness RL.
\end{proof}

\begin{definition}[S]
  Let $\Lens$ be a forgetful lens.

  Consider the following four functions $S(\Lens)$, that we wish to satisfy the
  simple symmetric lens laws.

  \begin{centering}
    \begin{tabular}{@{}r@{\ }c@{\ }l@{}}
      $\CreateR{} \App s$
      & =
      & $\Fst \App (\Lens.\PutR{} \App (s,\Lens.init))$\\
      
      $\CreateL{} \App v$
      & =
      & $\Fst \App (\Lens.\PutL{} \App (v,\Lens.init))$\\
      
      $\PutR{} \App s \App v$
      & =
      & $\LetIn{(\_,c)}{\Lens.\PutL \App (v,\Lens.init)}$\\
      &
      & $\LetIn{(s',\_)}{\Lens.\PutR \App (s,c)}$\\
      &
      & $s'$\\
      
      $\PutL{} \App v \App s$
      & =
      & $\LetIn{(\_,c)}{\Lens.\PutR \App (s,\Lens.init)}$\\
      &
      & $\LetIn{(y',\_)}{\Lens.\PutL \App (v,c)}$\\
      &
      & $y'$\\
    \end{tabular}
  \end{centering}
\end{definition}

\begin{mylemma}
  If $\Lens$ is a forgetful symmetric lens, then $S(\Lens)$ is a simple symmetric
  lens.
\end{mylemma}
\begin{proof}
  
  \CreatePutRL{}:
  
  \begin{centering}
    \begin{tabular}{@{}r@{\ }c@{\ }l@{\ }l}
      $S(\Lens).\PutLOf{(S(\Lens).\CreateROf{x})}{x}$
      & =
      & $S(\Lens).\PutLOf{(\Fst \App (\Lens.\PutR{(x,\Lens.init)}))}{x}$
      & By unfolding definitions
      \\
      
      & =
      & $\LetIn{(\_,c)}{\Lens.\PutR \App (x,\Lens.init)}$\\
      &
      & $\LetIn{(y',\_)}{\Lens.\PutL \App (\Fst \App (\Lens.\PutR{(x,\Lens.init)}),c)}$\\
      &
      & $y'$
      & By unfolding definitions\\
      
      & =
      & $\LetIn{(y',\_)}{\Lens.\PutL \App (\Lens.\PutR \App (x,\Lens.init))}$\\
      &
      & $y'$
      & By tuple harmony \\
      
      & =
      & $x$
      & By \PutRL \\
    \end{tabular}
  \end{centering}

  \CreatePutLR{}:  Symmetric to \CreatePutRL{}

  \PutRL{}:

  \begin{centering}
    \begin{tabular}{@{}r@{\ }c@{\ }l@{\ }l}
      $S(\Lens).\PutLOf{(S(\Lens).\PutROf{x}{y})}{x}$
      & =
      & $\LetIn{(\_,c)}{\Lens.\PutL \App (y,\Lens.init)}$\\
      & & $\LetIn{(y',c')}{\Lens.\PutR \App (x,c)}$\\
      & & $S(\Lens).\PutLOf{y'}{x}$
      & By unfolding definitions
      \\
      
      & =
      & $\LetIn{(\_,c)}{\Lens.\PutL \App (y,\Lens.init)}$\\
      & & $\LetIn{(y',c')}{\Lens.\PutR \App (x,c)}$\\
      & & $\LetIn{(\_,c'')}{\Lens.\PutR \App (x,\Lens.init)}$\\
      & & $\LetIn{(x',c''')}{\Lens.\PutL \App (y',c'')}$\\
      & & $x'$
      & By unfolding definitions
    \end{tabular}
  \end{centering}

  At this point, we know from \PutRL that $(x,c') = \Lens.\PutL \App (y',c')$.
  By Property~\ref{prop:forget-rl}, this means that $x' = x$, as desired.

  \PutLR{}:  Symmetric to \PutRL{}
\end{proof}

\begin{definition}
  Fix a symmetric lens $\Lens$ beteween $X$ and $Y$. Consider the
  function, $\SingleApp_{\Lens} \OfType (X + Y) \times \Lens.C
  \rightarrow ((X + Y) \times \Lens.C)$, defined as:

  \begin{tabular}{@{}r@{\ }c@{\ }l@{\ }l}
    $\SingleApp_{\Lens}(\InLOf{x},c)$
    & =
    & $\LetIn{(y,c')}{\Lens.\PutRSymOf{(x,c)}}$\\
    &
    & $(\InROf{y},c')$\\
    
    $\SingleApp_{\Lens}(\InROf{y},c)$
    & =
    & $\LetIn{(x,c')}{\Lens.\PutLSymOf{(y,c)}}$\\
    &
    & $(\InLOf{x},\SomeOf{(x,y)})$
  \end{tabular}
\end{definition}

\begin{definition}
  Fix a simple symmetric lens $\Lens$ beteween $X$ and $Y$. Consider the
  function, $\SingleApp_{\Lens} \OfType ((X + Y) \times \OptionOf{(X \times Y)})
  \rightarrow ((X + Y) \times \OptionOf{(X \times Y)})$, defined as:

  \begin{tabular}{@{}r@{\ }c@{\ }l@{\ }l}
    $\SingleApp_{\Lens}(\InLOf{x},\None)$
    & =
    & $\LetIn{y}{\Lens.\CreateROf{x}}$\\
    &
    & $(\InROf{y},\SomeOf{(x,y)})$\\
    
    $\SingleApp_{\Lens}(\InROf{y},\None)$
    & =
    & $\LetIn{x}{\Lens.\CreateLOf{y}}$\\
    &
    & $(\InLOf{x},\SomeOf{(x,y)})$\\
    
    $\SingleApp_{\Lens}(\InLOf{x'},\SomeOf{(x,y)})$
    & =
    & $\LetIn{y'}{\Lens.\PutROf{x'}{y}}$\\
    &
    & $(\InROf{y'},\SomeOf{(x',y')})$\\
    
    $\SingleApp_{\Lens}(\InROf{y'},\SomeOf{(x,y)})$
    & =
    & $\LetIn{x'}{\Lens.\PutLOf{y'}{x}}$\\
    &
    & $(\InROf{x'},\SomeOf{(x',y')})$\\
  \end{tabular}
\end{definition}

\begin{definition}
  Fix a symmetric lens \Lens over $X$ and $Y$. We define the relation $R_\Lens$
  over $\Lens.C$ and $\OptionOf{(X \times Y)}$ as the largest relation such
  that:
  \begin{enumerate}
  \item $R_\Lens(c,None) \BooleanImplies c = \Lens.init$
  \item $R_\Lens(c,Some (x,y)) \BooleanImplies
    \Lens.\PutRSymOf{(x,c)} = (y,c) \BooleanAnd
    \Lens.\PutLSymOf{(y,c)} = (x,c)$
  \end{enumerate}
\end{definition}

\begin{mylemma}
  \label{lem:s-equiv}
  Let $\Lens$ be a symmetric lens.  Let $c \in \Lens.C$ be a complement, and
  $xyo \in \OptionOf{(X \times Y)}$.  If $R_{\Lens}(c,xyo)$, then
  $apply(\Lens,c,es) = apply(S(\Lens),xyo,es)$.
\end{mylemma}
\begin{proof}
  By induction on the derivation the application of $apply(S(\Lens),xyo,es)$
  \begin{case}[empty list]
    So by the case, $apply(S(\Lens),xyo,[]) = []$. Furthermore,
    $apply(\Lens,c,[]) = []$, as desired.
  \end{case}
  \begin{case}[first edit is a create right]
    So by the case, $apply(S(l),\None,\InLOf{x}::es) = \InROf{y}::es'$, where
    $S(\Lens).\CreateROf{x} = y$ and $apply(S(\Lens),\SomeOf{(x,y)},es) = es'$.

    As $R_{\Lens}(None,c)$, $c = \Lens.init$. So, performing $apply$ on
    $\Lens.init$, we get $apply(\Lens,\Lens.init,(\InLOf{x})::es) =
    (\InROf{y'})::es''$ where $\Lens.putr(x,c) = (y',c')$ and $apply(\Lens,c',es)
    = es'$.

    So, by definition, $S(\Lens).\CreateROf{x} = \Fst \App (\Lens.\PutRSymOf{(x,\Lens.init)})$, so $y = y'$.

    Furthermore, by \PutRL, $\Lens.\PutLSymOf{(y,c')} = (x,c')$, and by another
    application of \PutLR, $\Lens.\PutRSymOf{(x,c')} = (y,c')$.  This means
    that $R_{\Lens}(c',Some (x,y))$.

    So, by induction assumption, $apply(\Lens,c',es) = apply(S(\Lens),xyo,es)$, so
    $es' = es''$.  This means, $apply(S(l),\None,\InLOf{x}::es) =
    \InROf{y}::es'$ and $apply(\Lens,\Lens.init,(\InLOf{x})::es) =
    (\InROf{y})::es'$, so they are equal, as desired.
  \end{case}
  \begin{case}[first edit is a create left]
    Symmetric to previous case
  \end{case}
  \begin{case}[first edit is a put right]
    So by the case, $apply(S(\Lens),\SomeOf{(x,y)},\InLOf{x'}::es) = \InROf{y'}::es'$, where
    $S(\Lens).\PutROf{x}{y} = y'$ and $apply(S(\Lens),\SomeOf{(x',y')},es) = es'$.

    Performing $apply$ on $c$, we get $apply(\Lens,c,(\InLOf{x'})::es) =
    (\InROf{y'})::es''$ where $\Lens.\PutRSymOf{(x',c)} = (y',c')$ and
    $apply(\Lens,c',es) = es''$.

    So, by definition, $S(\Lens).\PutROf{x}{y} = \Fst \App (\Lens.\PutRSymOf{(x,c'')})$, where $c'' = \Snd \App (\Lens.\PutLSymOf{(y,\Lens.init)})$.

    By assumption, $R_{\Lens}(c,\SomeOf{(x,y)})$, so $\Lens.\PutLSymOf{(y,c)}
    = (x,c)$.  So, by Property~\ref{prop:forget-lr}, we know $y' = y''$.

    Furthermore, by \PutRL, $\Lens.\PutLSymOf{(y',c')} = (x',c')$, and by another
    application of \PutLR, $\Lens.\PutRSymOf{(x',c')} = (y',c')$.  This means
    that $R_{\Lens}(c',Some (x',y'))$.

    So, by induction assumption, $apply(\Lens,c',es) = apply(S(\Lens),xyo,es)$, so
    $es' = es''$.  This means, $apply(S(l),\SomeOf{(x,y)},\InLOf{x'}::es) =
    \InROf{y'}::es'$ and $apply(\Lens,c,(\InLOf{x'})::es) =
    (\InROf{y'})::es'$, so they are equal, as desired.
  \end{case}
  \begin{case}[first edit is a put left]
    Symmetric to previous case
  \end{case}
\end{proof}

\begin{definition}[F]
  Let $\Lens$ be a simple symmetric lens between $X$ and $Y$.

  Consider the following set $C$, distinguished element of that set, $init$, and
  pair of functions $\PutRSym$ and $\PutLSym$, that we wish to satisfy the
  symmetric lens laws, and that we also wish to be forgetful.

  \begin{centering}
    \begin{tabular}{@{}r@{\ }c@{\ }l@{}}
      $C$
      & =
      & $\OptionOf{(X \times Y)}$\\
      
      $init$
      & =
      & $\None$\\
      
      $\PutRSymOf{(x,c)}$
      & =
      & $(y,\SomeOf{(x,y)})$ where $y = \begin{cases*}
        \Lens.\CreateROf{x} & if $c = \None$\\
        \Lens.\PutROf{x}{y'} & if $c = \SomeOf{(x',y')}$\\
        \end{cases*}$\\
      
      $\PutLSymOf{(y,c)}$
      & =
      & $(x,\SomeOf{(x,y)})$ where $x = \begin{cases*}
        \Lens.\CreateLOf{y} & if $c = \None$\\
        \Lens.\PutLOf{y}{x'} & if $c = \Some{(x',y')}$\\
        \end{cases*}$\\
    \end{tabular}
  \end{centering}
\end{definition}

\begin{mylemma}
  \label{lem:f-sym}
  If $\Lens$ is a simple symmetric lens, then $F(\Lens)$ is a symmetric lens.
\end{mylemma}
\begin{proof}
  \PutRL: There are two cases, $c = \None$, and $c = \SomeOf{(x,y)}$.

  \begin{case}[c = \None]
    Let $(y',c') = F(\Lens).\PutRSymOf{(x',\None)}$. This means that $y' =
    \Lens.\CreateROf{x'}$, and $c' = \SomeOf{(x',y')}$.
    
    Now, consider $(x'',c'') = F(\Lens).\PutLSymOf{(y',\SomeOf{(x',y')})}$. By
    unfolding definitions, $x'' = \Lens.\PutLOf{y'}{x'}$. By \CreatePutRL, $x''
    = x'$, meaning $c'' = \SomeOf{(x',y')}$. This means $(x',c') =
    F(\Lens).\PutLSymOf{(y',\SomeOf{(x',y')})}$, as desired.
  \end{case}

  \begin{case}[c = \SomeOf{(x,y)}]
    Let $(y',c') = F(\Lens).\PutRSymOf{(x',\SomeOf{(x,y)})}$. This means that $y' =
    \Lens.\PutROf{x'}{y}$, and $c' = \SomeOf{(x',y')}$.
    
    Now, consider $(x'',c'') = F(\Lens).\PutLSymOf{(y',\SomeOf{(x',y')})}$. By
    unfolding definitions, $x'' = \Lens.\PutLOf{y'}{x'}$. By \PutRL, $x''
    = x'$, meaning $c'' = \SomeOf{(x',y')}$. This means $(x',c') =
    F(\Lens).\PutLSymOf{(y',\SomeOf{(x',y')})}$, as desired.
  \end{case}

  The second requirement, \PutLR, is symmetric.
\end{proof}

\begin{mylemma}
  If $\Lens$ is a simple symmetric lens, then $F(\Lens)$ is a forgetful symmetric lens.
\end{mylemma}
\begin{proof}
  By Lemma~\ref{lem:f-sym}, we know $F(\Lens)$ is symmetric, so we merely need
  to show it is forgetful.  We will tackle merely \ForgetfulRL, as the proof for
  \ForgetfulLR is symmetric.

  Let $c_1$ and $c_2$ be two arbitrary complements, and $x$ and $y$ two
  arbitrary values of $X$ and $Y$, respectively.
  
  We know that $c_1' = \Snd \App (F(\Lens).\PutRSymOf{(x,c_1)})$ and $c_2' =
  \Snd \App (F(\Lens).\PutRSymOf{(x,c_2)})$.  Now, by inversion on
  $F(\Lens).\PutRSym$, we know that both $c_1' = Some(x,y_1')$ and $c_2' =
  Some(x,y_2')$ for some values of $y_1$ and $y_2$ (though we don't actually
  care about the values of $y_1$ and $y_2$).

  Now by unfolding definitions we know, $\Snd \App
  (F(\Lens).\PutLSymOf{(y,\SomeOf{(x,y_1')})}) =
  \SomeOf(\Lens.\PutLOf{x}{y},y) = c_1''$. Similarly, we know $\Snd \App
  (F(\Lens).\PutLSymOf{(y,\SomeOf{(x,y_2')})}) =
  \SomeOf(\Lens.\PutLOf{x}{y},y) = c_2''$, so $c_1'' = c_2''$, as intended.
\end{proof}

\begin{mylemma}
\label{lem:f-equiv}
  If $\Lens$ be a simple symmetric lens, then $apply(\Lens,xyo,es) =
  apply(F(\Lens),xyo,es)$.
\end{mylemma}
\begin{proof}
  By induction on the derivation of $apply$ on $\Lens$!

  \begin{case}[empty list]
    So by the case, $apply(S(\Lens),xyo,[]) = []$. Furthermore,
    $apply(\Lens,xyo,[]) = []$, as desired.
  \end{case}

  \begin{case}[first edit is a create right]
    So by the case, $apply(\Lens,\None,\InLOf{x}::es) = \InROf{y}::es'$, where
    $\Lens.\CreateROf{x} = y$ and $apply(\Lens,\SomeOf{(x,y)},es) = es'$.

    Performing $apply$ on
    $\None$, we get $apply(F(\Lens),\None,(\InLOf{x})::es) =
    (\InROf{y'})::es''$ where $F(\Lens).putr(x,\None) = (y',c)$ and $apply(\Lens,c,es)
    = es'$.

    Unfolding definitions, $F(\Lens).putr(x,\None) = \Lens.\CreateROf{x} =
    (y,\SomeOf{(x,y)})$, so $c = \SomeOf{(x,y)}$ and $y = y'$

    So, by induction assumption, $apply(\Lens,\SomeOf{(x,y)},es) = apply(F(\Lens),\SomeOf{(x,y)},es)$, so
    $es' = es''$.  This means, $apply(\Lens,\None,\InLOf{x}::es) =
    \InROf{y}::es'$ and $apply(F(\Lens),\None,(\InLOf{x})::es) =
    (\InROf{y})::es'$, so they are equal, as desired.
  \end{case}

  \begin{case}[first edit is a create left]
    Symmetric to previous case
  \end{case}

  \begin{case}[first edit is a put right]
    So by the case, $apply(\Lens,\SomeOf{(x,y)},\InLOf{x'}::es) = \InROf{y'}::es'$, where
    $\Lens.\PutROf{x}{y} = y'$ and $apply(\Lens,\SomeOf{(x',y')},es) = es'$.

    Performing $apply$ on $F(\Lens)$, we get
    $apply(F(\Lens),\SomeOf{(x,y)},(\InLOf{x'})::es) = (\InROf{y'})::es''$ where
    $F(\Lens).\PutRSymOf{(x',\SomeOf{(x,y)})} = (y'',\SomeOf{(x',y'')})$ and
    $apply(\Lens,c',es) = es''$.
    
    So, by definition, $\Fst \App (F(\Lens).\PutRSymOf{(x',\SomeOf{(x,y)})})
    = \Lens.PutROf{x'}{y'}$, so also by definition, $c' = \SomeOf{(x',y')}$.

    So, by induction assumption, $apply(\Lens,\SomeOf{(x',y')},es) =
    apply(F(\Lens),\SomeOf{(x',y')},es)$, so $es' = es''$. This means,
    $apply(\Lens,\SomeOf{(x,y)},\InLOf{x'}::es) = \InROf{y'}::es'$ and
    $apply(F(\Lens),\SomeOf{(x,y)},(\InLOf{x'})::es) = (\InROf{y'})::es'$, so
    they are equal, as desired.
  \end{case}

  \begin{case}[first edit is a put left]
    Symmetric to previous case
  \end{case}
\end{proof}

\begin{theorem}
  Let $\Lens$ be a symmetric lens. The lens $\Lens$ is equivalent to a forgetful
  lens if, and only if, there exists a simple symmetric lens $\Lens'$ where
  $apply(\Lens,\Lens.init,es) = apply(\Lens',\None,es)$, for all put sequences
  $es$.
\end{theorem}

\begin{proof}
  \begin{case}[$\Rightarrow$]
    Let $\Lens$ be equivalent to a forgetful lens $\Lens'$.  Consider the
    simple symmetric lens, $S(\Lens')$.  By Lemma~\ref{lem:s-equiv},
    $apply(\Lens',\Lens'.init,es) = apply(S(\Lens'),\None,es)$.  As $\Lens$ is
    equivalent to $\Lens'$, $apply(\Lens,\Lens.init,es) =
    apply(\Lens',\Lens'.init,es)$.  So, by transitivity,
    $apply(\Lens,\Lens.init,es) = apply(S(\Lens),\None,es)$.
  \end{case}

  \begin{case}[$\Leftarrow$]
    Let $\Lens'$ be a simple symmetric lens where $apply(\Lens,\Lens.init,es) =
    apply(\Lens',\None,es)$, for all put sequences $es$.

    Consider $F(\Lens')$, a forgetful symmetric lens where $apply(\Lens',\None,es) =
    apply(F(\Lens'),\Lens'.init,es)$, for all put sequences $es$, as
    $\Lens'.init = \None$, by Lemma~\ref{lem:f-equiv}.  By transitivity, $apply(\Lens,\Lens.init,es) =
    apply(\Lens',\Lens'.init,es)$, so $\Lens$ and $\Lens'$ are equivalent.
  \end{case}
    
\end{proof}

\section{SREs and Information Theory}
\begin{theorem}
  \label{thm:semantics-correctness}
  If $\Regex \SSREquiv \RegexAlt$ then $\ProbabilityOf{\Regex}{\String} =
  \ProbabilityOf{\RegexAlt}{\String}$, for all strings $\String \in \LanguageOf{\Regex}$.
\end{theorem}
The following lemma is
\begin{proof}
\begin{enumerate}

\item
($S \; |_1 \; \varnothing \equiv^S S$)

For all $s \in \mathcal{L}(S \; | \; \varnothing)$, then $P_{S \; |_1 \; \varnothing} = 1 * P_S(s) = P_S(s)$.
\item
($S \cdot \varnothing \equiv^S \varnothing$)

For all $s \in \mathcal{L}(S \cdot \varnothing)$, the $P_{S \cdot \varnothing}(s) = \sum_{s_1 \cdot s_2}P_S(s_1) * P_{\varnothing}(s_2)$, but this sum is empty, hence $P_{S \cdot \varnothing}(s) = 0 = P_{\varnothing}(s)$.
\item
($\varnothing \cdot S \equiv^S \varnothing$)

Similar to the previous case.
\item
($(S \cdot S') \cdot S'' \equiv^S S \cdot (S' \cdot S'')$). Let $s \in \mathcal{L}(S \cdot S' \cdot S'')$

Then
\begin{align*}
P_{(S \cdot S') \cdot S''}(s) &= \sum_{s=s_4 \cdot s_3}P_{S \cdot S'}(s_4) * P_{S''}(s_3)\\
&= \sum_{s=s_4 \cdot s_3}\left(\sum_{s_4=s_1 \cdot s_2} P_{S}(s_1) * P_{S'}(s_2)\right)* P_{S''}(s_3)\\
&= \sum_{s=s_1\cdot s_2 \cdot s_3} P_{S}(s_1) * P_{S'}(s_2) * P_{S''}(s_3)\\
P_{S \cdot (S' \cdot S'')}(s) &= \sum_{s=s_1\cdot s_4} P_{S}(s) * \left(\sum_{s_4=s_2\cdot s_3}P_{S'}(s_2) * P_{S''}(s_3)\right)
\end{align*}
\item
($(S \; |_{p_1} \; S') \; |_{p_2} \; S'' \equiv^S S \; |_{p_1 * p_2} \; (S' \; |_{\frac{(1-p_1)*p_2}{1-p_1*p_2}} \; S'')$)

Let $s \in \mathcal{L}(S \; |\; S' \; | \; S'')$. Then
\begin{align*}
P_{(S \; |_{p_1} \; S') \; |_{p_2} \; S''}(s) &= p_2 * P_{S \; |_{p_1} \; S'}(s) + (1-p_2)P_{S''}(s)\\
&= (p_2 * (p_1 * P_S(s) + (1-p_1) * P_{S'}(s))) + (1-p_2)P_{S''}(s)\\
&= (p_2 * (p_1 * P_S(s) + P_{S'}(s)-p_1* P_{S'}(s))) + (1-p_2)P_{S''}(s)\\
&= p_2*p_1 * P_S(s) + p_2*P_{S'}(s)-p_2*p_1* P_{S'}(s) + (1-p_2)P_{S''}(s)\\
&= p_2*p_1 * P_S(s) + (1-p_1) * p_2 * P_{S'}(s)+ (1-p_2)P_{S''}(s)\\
S \; |_{p_1 * p_2} \; (S' \; |_{\frac{(1-p_1)*p_2}{1-p_1*p_2}} \; S'') &= (p_1p_2)  P_S(s) + (1-p_1p_2)  \left(\frac{(1-p_1)p_2}{1-p_1p_2}  P_{S'}(s) + \left(1-\frac{(1-p_1)p_2}{1-p_1p_2}\right)P_{S''}(s)\right)
\end{align*}
\item
($S \; |_p \; T \equiv^S T \; |_{1-p} \; S$)

For all $s \in \mathcal{L}(S \; | \; T)$, we have
$$P_{S \; |_p \; T}(s) = p*P_S(s) + (1-p)*P_T(s) = (1-p) * P_T(s) + (1-(1-p))*P_s(s) = P_{T \; |_{1-p} \; S}(s)$$
\item
($S \cdot (S' \; |_p \; S'' ) \equiv^S (S \cdot S') \; |_p \; (S \cdot S'')$)

For all $s \in \mathcal{L}(S \cdot (S' \sep S''))$ we have
\begin{align*}
P_{S \cdot (S' \; |_p \; S'' )} &= \sum_{s_1 \cdot s_2 = s}P_S(s_1) * P_{S' \; |_p \; S''}(s)\\
&= \sum_{s_1 \cdot s_2 = s}P_S(s_1) * (p * P_{S'}(s_2) + (1-p) * P_{S''}(s_2))\\
&= \sum_{s_1 \cdot s_2 = s}p * P_S(s_1) * P_{S'}(s_2) + \sum_{s_1 \cdot s_2 = s}(1-p) * P_S(s_1) * P_{S''}(s_2)\\
P_{(S \cdot S') \; |_p \; (S \cdot S'')}(s) &= p*P_{S \cdot S'}(s) + (1-p)*P_{S \cdot S''}(s)
\end{align*}
\item
$((S' \; |_p \; S'') \cdot S \equiv^S (S' \cdot S) \; |_p \; (S'' \cdot S))$

For all $s \in \mathcal{L}((S' \; | \; S'') \cdot S)$,  we have
\begin{align*}
P_{(S' \; |_p \; S'') \cdot S}(s) &= \sum_{s_1 \cdot s_2 = s}P_{S' \; |_p \; S''}(s_1) * P_S(s_2)\\
&= \sum_{s_1 \cdot s_2 = s}(p * P_{S'}(s_1) + (1-p) * P_{S''}(s_1))* P_S(s_2)\\
&= \sum_{s_1 \cdot s_2=s}p*P_{S'}(s_1) * P_S(s_2) + \sum_{s_1 \cdot s_2 = s}(1-p) * P_{S''}(s_1) * P_S(s_2)\\
P_{(S' \cdot S) \; |_p \; (S'' \cdot S)}&= p * P_{S' \cdot S}(s) + (1-p) * P_{S'' \cdot S}(s)
\end{align*}
\item
($\varepsilon \cdot S \equiv^S S$)

For all $s \in \mathcal{L}(S)$, we have
$$P_{\epsilon \cdot S}(s) = \sum_{s_1 \cdot s_2 = s}P_{\epsilon}(s_1) * P_S(s_2) = P_{\epsilon}(\epsilon) * P_S(s_2) = P_S(s)$$
since $s_2 = s$ in the computation above.
\item
($S \cdot \epsilon \equiv^S S$)

Similar to the previous case.
\item
($S^{*p} \equiv^S \epsilon \; |_{(1-p)} \; (S^{*p} \cdot S)$)

Let $s \in \mathcal{L}(S^*)$, in which case we have
\begin{align*}
P_{\epsilon \; |_{1-p} \; (S^{*p} \cdot S)}(s) &= (1-p) * P_{\epsilon}(s) + p * P_{S^{*p} \cdot S}(s)\\
&= (1-p) * P_{\epsilon}(s) + p * \sum_{a \cdot b = s}\left(\sum_n \sum_{s_1 \ldots \cdot s_n=a} p^n (1-p) \prod_{i=1}^n P_S(s_i)\right) * P_S(b)\\
&= (1-p) * P_{\epsilon}(s) + \sum_{a \cdot b = s}\left(\sum_n \sum_{s_1 \cdot \ldots \cdot s_n=s} p^{n+1} (1-p) \prod_{i=1}^n P_S(s_i) * P_S(b)\right)\\
&= (1-p) * P_{\epsilon}(s) + \sum_n \sum_{s_1 \cdot \ldots \cdot s_{n+1}=s} p^{n+1} (1-p) \prod_{i=1}^{n+1} P_S(s_i)\\
&= (1-p) * P_{\epsilon}(s) + \left(\sum_n \sum_{s_1 \cdot \ldots \cdot s_n = s} p^n * (1-p) * \prod_{i=1}^n P_S(s_i)\right) - (1-p) * (\mathbbm{1}_{s = \varepsilon})
\end{align*}
Since
$$
P_{S^{*p}}(s) = \sum_n \sum_{s_1 \cdot \ldots \cdot s_n = s} p^n * (1-p) * \prod_{i=1}^n P_S(s_i)
$$
then if $s = \epsilon$, the terms $(1-p) * P_{\epsilon}(s)$ and $(1-p) * \mathbbm{1}_{s=\varepsilon}$ cancel each other out. Otherwise, if $s \neq \epsilon$, then $(1-p) * P_{\epsilon}(s) = 0 = (1-p) * (\mathbbm{1}_{s = \varepsilon})$ from which the result follows.
\item
($S^{*p} \equiv^S \epsilon \; |_{(1-p)} \; (S \cdot S^{*p})$)

Similar to the previous case.
\end{enumerate}
\end{proof}
\begin{mylemma}[Equivalence of \ConcatSequence{} and \Concat{}]
  If $\LanguageOf{\Regex}=\LanguageOf{\Sequence}$,
  and $\LanguageOf{\RegexAlt}=\LanguageOf{\SequenceAlt}$,
  then $\LanguageOf{\RegexConcat{\Regex}{\RegexAlt}}=\LanguageOf{\ConcatSequenceOf{\Sequence}{\SequenceAlt}}$.
\end{mylemma}
\begin{proof}
  Let $\Sequence=\SequenceOf{\String_0\SeqSep\Atom_1\SeqSep\ldots
    \SeqSep\Atom_n\SeqSep\String_n}$, and
  let $\SequenceAlt=[\StringAlt_0\SeqSep\AtomAlt_1\SeqSep\ldots
  \SeqSep\AtomAlt_m\SeqSep\StringAlt_m]$. Then
  \begin{align*}  
    \LanguageOf{\ConcatSequenceOf{\Sequence}{\SequenceAlt}} & = 
                                                              \LanguageOf{\SequenceOf{\String_0\SeqSep\Atom_1\SeqSep\ldots
                                                              \SeqSep\Atom_n\SeqSep\String_n\Concat\StringAlt_0\SeqSep{}
                                                              \AtomAlt_1\SeqSep\ldots\SeqSep\AtomAlt_m\SeqSep\StringAlt_m}} \\
                                                            & = 
                                                              \{\String_0\Concat\String_1'\Concat\ldots\Concat\String_n'\Concat\String_n
                                                              \Concat\StringAlt_0\Concat\StringAlt_1'\Concat\ldots
                                                              \Concat\StringAlt_m'\Concat\StringAlt_m\SuchThat{} \String_i'\in\LanguageOf{\Atom_i} \BooleanAnd{}
                                                              \StringAlt_i'\in\LanguageOf{\AtomAlt_i}\}\\
                                                            & = 
                                                              \{\String\Concat\StringAlt{} \SuchThat{} \String\in\LanguageOf{\Sequence}
                                                              \BooleanAnd{} \StringAlt\in\LanguageOf{\SequenceAlt}\}\\
                                                            & =
                                                              \{\String\Concat\StringAlt{} \SuchThat{} \String\in\LanguageOf{\Regex}
                                                              \BooleanAnd{} \StringAlt\in\LanguageOf{\RegexAlt}\}\\
                                                            & =
                                                              \LanguageOf{\RegexConcat{\Regex}{\RegexAlt}}
\end{align*}
\end{proof}

\begin{mylemma}[Equivalence of \ConcatDNF{} and \Concat{}]
  \label{lem:cdnfeq}
  If $\LanguageOf{\Regex}=\LanguageOf{\DNFRegex}$,
  and $\LanguageOf{\RegexAlt}=\LanguageOf{\DNFRegexAlt}$,
  then $\LanguageOf{\RegexConcat{\Regex}{\RegexAlt}}=
  \LanguageOf{\ConcatDNFOf{\DNFRegex}{\DNFRegexAlt}}$.
\end{mylemma}
\begin{proof}
  Let $\DNFRegex=\DNFOf{\Sequence_0\DNFSep\ldots\DNFSep\Sequence_n}$, and
  let $\DNFRegexAlt=\DNFOf{\SequenceAlt_0\DNFSep\ldots\DNFSep\SequenceAlt_m}$
  \begin{align*}
    \LanguageOf{\ConcatDNFOf{\DNFRegex}{\DNFRegexAlt}} & = 
                                                         \LanguageOf{\DNFOf{\ConcatSequenceOf{\Sequence_i}{\SequenceAlt_j}
                                                         \text{ for $i\in\RangeIncInc{1}{n}$, $j\in\RangeIncInc{1}{m}$}}} \\
                                                       & = 
                                                         \{\String\SuchThat \String\in\ConcatSequenceOf{\Sequence_i}{\SequenceAlt_j} \text{ where $i\in\RangeIncInc{1}{n}$, $j\in\RangeIncInc{1}{m}$}\}\\
                                                       & = 
                                                         \{\String\Concat\StringAlt{} \SuchThat{} \String\in\LanguageOf{\Sequence_i}
                                                         \BooleanAnd{} \StringAlt\in\LanguageOf{\SequenceAlt_j}\} \text{ where $i\in\RangeIncInc{1}{n}$, $j\in\RangeIncInc{1}{m}$}\}\\
                                                       & =
                                                         \{\String\Concat\StringAlt{} \SuchThat{} \String\in\LanguageOf{\DNFRegex}
                                                         \BooleanAnd{} \StringAlt\in\LanguageOf{\DNFRegexAlt}\}\\
                                                       & =
                                                         \{\String\Concat\StringAlt{} \SuchThat{} \String\in\LanguageOf{\Regex}
                                                         \BooleanAnd{} \StringAlt\in\LanguageOf{\RegexAlt}\}\\
                                                       & =
                                                         \LanguageOf{\RegexConcat{\Regex}{\RegexAlt}}
\end{align*}
\end{proof}

\begin{mylemma}[Equivalence of $\Atom$ and $\AtomToDNFOf{\Atom}$]
  \label{lem:atomtodnfeq}
  $\LanguageOf{\Atom} = \LanguageOf{\AtomToDNFOf{\Atom}}$
\end{mylemma}
\begin{proof}
\begin{align*}
  \LanguageOf{\AtomToDNFOf{\Atom}} &=
  \LanguageOf{\DNFOf{\SequenceOf{\EmptyString \SeqSep \Atom \SeqSep \EmptyString}}}\\
  & =
  \SetOf{\String \SuchThat \String \in
    \LanguageOf{\SequenceOf{\EmptyString \SeqSep \Atom \SeqSep \EmptyString}}} \\
  & =
  \SetOf{\EmptyString\Concat\String\Concat\EmptyString \SuchThat \String \in
    \LanguageOf{\Atom}} \\
   & = \SetOf{\String \SuchThat \String \in
    \LanguageOf{\Atom}} \\
    &= \LanguageOf{\Atom}
    \end{align*}
\end{proof}

\begin{mylemma}[Equivalence of \OrDNF{} and \Or{}]
  \label{lem:odnfeq}
  If $\LanguageOf{\Regex}=\LanguageOf{\DNFRegex}$,
  and $\LanguageOf{\RegexAlt}=\LanguageOf{\DNFRegexAlt}$,
  then $\LanguageOf{\RegexOr{\Regex}{\RegexAlt}}=
  \LanguageOf{\OrDNFOf{\DNFRegex}{\DNFRegexAlt}}$.
\end{mylemma}
\begin{proof}
  Let $\DNFRegex=\DNFOf{\Sequence_0\DNFSep\ldots\DNFSep\Sequence_n}$, and
  let $\DNFRegexAlt=\DNFOf{\SequenceAlt_0\DNFSep\ldots\DNFSep\SequenceAlt_m}$
  
  \begin{align*}
    \mathcal{L}(DS \oplus DT)& = 
                                                     \LanguageOf{\DNFOf{\Sequence_0\DNFSep\ldots\DNFSep\Sequence_n\DNFSep
                                                     \SequenceAlt_1\DNFSep\ldots\DNFSep\SequenceAlt_m}}\\
                                                   & = 
                                                     \{\String\SuchThat{} \String\in\Sequence_i\vee\String\in\SequenceAlt_j \text{ where $i\in\RangeIncInc{1}{n}$, $j\in\RangeIncInc{1}{m}$}\}\\
                                                   & = 
                                                     \{\String{} \SuchThat{} \String\in\LanguageOf{\DNFRegex}
                                                     \BooleanOr{} \String\in\LanguageOf{\DNFRegexAlt}\}\\
                                                   & =
                                                     \{\String \SuchThat{} \String\in\LanguageOf{\Regex}
                                                     \BooleanOr{} \String\in\LanguageOf{\RegexAlt}\}\\
                                                   & =
                                                     \LanguageOf{\RegexOr{\Regex}{\RegexAlt}}
  \end{align*}
\end{proof}

\begin{theorem}\label{ConversionPreservesSemantics}
  For all regular expressions \Regex{},
  $\LanguageOf{\ToDNFRegexOf{\Regex}}=\LanguageOf{\Regex{}}$.
\end{theorem}
\begin{proof}
  By structural induction.

  Let $\Regex=\String$.
  $\LanguageOf{\ToDNFRegex(\String)}=\LanguageOf{\DNFOf{\SequenceOf{\String}}}=
  \{\String\}=\LanguageOf{\String}$

  Let $\Regex=\emptyset$.
  $\LanguageOf{\ToDNFRegex(\emptyset)}=\LanguageOf{\DNFOf{}} =
  \{\} = \LanguageOf{\emptyset}$.

  Let $\Regex=\StarOf{\Regex'}$.
  By induction assumption, $\LanguageOf{\ToDNFRegex(\Regex')}=
  \LanguageOf{\Regex'}$.
  \begin{align*}
    \LanguageOf{\ToDNFRegex(\StarOf{\DNFRegex'})} & =
                                                    \LanguageOf{\DNFOf{\SequenceOf{\StarOf{\ToDNFRegex(\Regex')}}}}\\
                                                  & =
                                                    \{\String\SuchThat\String\in
                                                    \LanguageOf{\SequenceOf{\StarOf{\ToDNFRegex(\Regex')}}}\}\\
                                                  & = 
                                                    \{\String\SuchThat{} \String\in\LanguageOf{\StarOf{\ToDNFRegex(\Regex')}}\}\\
                                                  & =
                                                    \{\String_1\Concat\ldots\Concat\String_n\SuchThat{}
                                                    n\in\Nats \BooleanAnd\String_i\in\LanguageOf{\ToDNFRegex(\Regex')}\}\\
                                                  & =
                                                    \{\String_1\Concat\ldots\Concat\String_n\SuchThat{}
                                                    n\in\Nats\BooleanAnd\String_i\in\LanguageOf{\Regex'}\}\\
                                                  & = \LanguageOf{\StarOf{\Regex'}}
  \end{align*}

  Let $\Regex=\RegexConcat{\Regex_1}{\Regex_2}$.
  By induction assumption,
  $\LanguageOf{\ToDNFRegex(\Regex_1)}=\LanguageOf{\Regex_1}$, and
  $\LanguageOf{\ToDNFRegex(\Regex_2)}=\LanguageOf{\Regex_2}$.
  $\ToDNFRegex(\RegexConcat{\Regex_1}{\Regex_2})=
  \ConcatDNFOf{\ToDNFRegex(\Regex_1)}{\ToDNFRegex(\Regex_2)}$.
  By Lemma~\ref{lem:cdnfeq},
  $\RegexConcat{\Regex_1}{\Regex_2}=
  \ConcatDNFOf{\ToDNFRegex(\Regex_1)}{\ToDNFRegex(\Regex_2)}$.

  Let $\Regex=\RegexOr{\Regex_1}{\Regex_2}$.
  By induction assumption,
  $\LanguageOf{\ToDNFRegex(\Regex_1)}=\LanguageOf{\Regex_1}$, and
  $\LanguageOf{\ToDNFRegex(\Regex_2)}=\LanguageOf{\Regex_2}$.

  $\Downarrow(S \; | \; T) = (\Downarrow S) \oplus (\Downarrow T)$. By Lemma~\ref{lem:odnfeq},
  $\mathcal{L}(S \; | \; T) = \mathcal{L}((\Downarrow S) \oplus (\Downarrow T))$.
  
\end{proof}
\begin{theorem*}
  $\ProbabilityOf{\Regex}{\String} = \ProbabilityOf{\ToDNFRegex
    \Regex}{\String}$ and $\LanguageOf{\Regex} =
  \LanguageOf{\ToDNFRegexOf{\Regex}}$
\end{theorem*}
\begin{proof}
By \cref{ConversionPreservesSemantics}, $\LanguageOf{\Regex} =
  \LanguageOf{\ToDNFRegexOf{\Regex}}$, so we will prove the result $\ProbabilityOf{\Regex}{\String} = \ProbabilityOf{\ToDNFRegex
    \Regex}{\String}$.
\begin{enumerate}
\item
$(S = s)$

By definition, $\Downarrow s = \langle ([s],1) \rangle$ and $P_{\langle ([s], 1) \rangle}(s') = 
\begin{cases}
1 & \text{if } s = s'\\
0 & \text{otherwise}
\end{cases}$ hence $P_{\langle ([s],1) \rangle}(s') = P_s(s')$.
\item
$(S = \varnothing)$

By definition, $\Downarrow \varnothing = \langle \rangle$ and $P_{\langle \rangle}(s) = 0 = P_{\varnothing}(s)$
\item
$(S = S^{*p})$

By definition, $\Downarrow (S^{*p}) = \langle [\epsilon \cdot (\Downarrow S)^{*p} \cdot \epsilon] \rangle $. By the induction hypothesis, $P_{(\Downarrow S)} = P_{S}$, hence
\begin{align*}
P_{[\epsilon \cdot (\Downarrow S)^{*p}] \cdot \epsilon}(s) &= P_{(\Downarrow S)^{*p}}(s)\\
&= \sum_n \sum_{s_1 \cdot \ldots \cdot s_n=s} p^n * (1-p) * \prod_{i=1}^n P_{(\Downarrow S)}(s_i)\\
&= \sum_n \sum_{s_1 \cdot \ldots \cdot s_n=s} p^n * (1-p) * \prod_{i=1}^n P_{S}(s_i)\\
&= P_{S^{*p}}(s)
\end{align*}
\item
$(S = S_1 \cdot S_2)$

By definition, $\Downarrow (S_1 \cdot S_2) = (\Downarrow S_1) \odot (\Downarrow S_2)$. By the induction hypothesis, $P_{S_1} = P_{(\Downarrow S_1)}$, and $P_{S_2} = P_{(\Downarrow S_2)}$. Assume that $\Downarrow S_1 = \langle (SQ_1,p_1) \; | \; \ldots \; | \; (SQ_m, p_n)\rangle$ and $\Downarrow S_2 = \langle (TQ_1,q_1) \; | \; \ldots \; | \; (TQ_n, q_m) \rangle$. Then
\begin{align*}
P_{S_1}(s) &= \sum_{i=1}^n p_i P_{SQ_i}(s)\\
P_{S_1}(s) &= \sum_{j=1}^m q_i P_{TQ_j}(s)\\
\end{align*}
hence
\begin{align*}
P_{S_1 \cdot S_2}(s) &= \sum_{s_1 \cdot s_2 = s} P_{S_1}(s_1)* P_{S_2}(s_2)\\
&= \sum_{s_1 \cdot s_2 = s} \left(\sum_{i=1}^n p_i P_{SQ_i}(s_1)\right)* \left(\sum_{j=1}^m q_i P_{TQ_j}(s_2)\right)\\
&= \sum_{s_1 \cdot s_2 = s} \sum_{i=1}^n \sum_{j=1}^m (p_i * q_j) * (P_{SQ_i}(s_1) *  P_{TQ_j}(s_2))\\
&= \sum_{i=1}^n \sum_{j=1}^m (p_i * q_j)\sum_{s_1 \cdot s_2 = s} (P_{SQ_i}(s_1) *  P_{TQ_j}(s_2))
\end{align*}
By definition,
$\ConcatDNFOf{\DNFOf{(\Sequence_1,\Probability_1)\DNFSep\ldots\DNFSep(\Sequence_n,\Probability_n)}}{\DNFOf{(\SequenceAlt_1,\ProbabilityAlt_1)\DNFSep\ldots\DNFSep(\SequenceAlt_m,\ProbabilityAlt_m)}}=$\\
      $\DNFLeft (\ConcatSequenceOf{\Sequence_1}{\SequenceAlt_1},\Probability_1*\ProbabilityAlt_1)\DNFSep \ldots
      \DNFSep
      (\ConcatSequenceOf{\Sequence_1}{\SequenceAlt_m},\Probability_1*\ProbabilityAlt_m)\DNFSep
      \ldots$\\
      $\DNFSep
      (\ConcatSequenceOf{\Sequence_n}{\SequenceAlt_1},\Probability_n*\ProbabilityAlt_1)\DNFSep
      \ldots \DNFSep
      (\ConcatSequenceOf{\Sequence_n}{\SequenceAlt_m},\Probability_n * \ProbabilityAlt_m) \DNFRight$
      
where
$$\ConcatSequenceOf{[\String_0\SeqSep\Atom_1\SeqSep\ldots\SeqSep\Atom_n\SeqSep\String_n]}{[\StringAlt_0\SeqSep\AtomAlt_1\SeqSep\ldots\SeqSep\AtomAlt_m\SeqSep\StringAlt_m]}=
  [\String_0\SeqSep\Atom_1\SeqSep\ldots\SeqSep\Atom_n\SeqSep\String_n\Concat\StringAlt_0\SeqSep\AtomAlt_1\SeqSep\ldots\SeqSep\AtomAlt_m\SeqSep\StringAlt_m]$$
Hence
$$P_{\langle (SQ_1, p_1) \; | \; \ldots \; | \;(SQ_n,p_n)\rangle \odot \langle (TQ_1, q_1) \; | \; \ldots \; | \;(TQ_m,q_m)\rangle}(s) = \sum_{i=1}^n \sum_{j=1}^m(p_i * q_j) * P_{SQ_i \odot_{SQ} TQ_j}(s)$$
Since $P_{[s_0 \cdot A_1 \cdot \ldots \cdot A_n \cdot s_n]}(s) = \sum_{s_0 \cdot s'_1 \cdot \ldots \cdot s'_n \cdot s_n = s} \prod_{i=1}^n P_{A_i}(s'i)$, then
\begin{align*}
P_{[s_0 \cdot A_1 \cdot \ldots \cdot A_n \cdot s_n] \odot_{SQ} [t_0 \cdot B_1 \cdot \ldots \cdot B_m \cdot t_m]}(s) &= \sum_{s_0 \cdot s'_1 \cdot \ldots \cdot s'_n \cdot s_n \cdot t_0 \cdot t'_1 \cdot \ldots t'_m \cdot t_m =s } \prod_{i=1}^n P_{A_i}(s'_i) \prod_{j=1}^m P_{B_j}(t'_j)\\
&= \sum_{a \cdot b = s} \sum_{s_0 \cdot s'_1 \cdot \ldots \cdot s'_n \cdot s_n = a} \sum_{t_0 \cdot t'_1 \cdot \ldots \cdot t'_m \cdot t_m = b} \prod_{i=1}^n P_{A_i}(s'_i) \prod_{j=1}^m P_{B_j}(t'_j)\\
& \sum_{a \cdot b = s} \left(\sum_{s_0 \cdot s'_1 \cdot \ldots \cdot s'_n \cdot s_n = a} \prod_{i=1}^n P_{A_i}(s'_i)\right) \left(\sum_{t_0 \cdot t'_1 \cdot \ldots \cdot t'_m \cdot t_m = b} \prod_{j=1}^m P_{B_j}(t'_j)\right)\\
&= \sum_{a \cdot b = s}P_{[s_0 \cdot A_1 \cdot \ldots \cdot A_n \cdot s_n]}(a) * P_{[t_0 \cdot B_1 \cdot \ldots \cdot B_m \cdot t_m]}(b)
\end{align*}
In other words, $P_{SQ \odot_{SQ} TQ}(s) = \sum_{a \cdot b = s}P_{SQ}(a) * P_{TQ}(b)$. Hence, 
\begin{align*}
P_{S_1 \cdot S_2}(s) &= \sum_{i=1}^n \sum_{j=1}^m (p_i * q_j)\sum_{s_1 \cdot s_2 = s} (P_{SQ_i}(s_1) *  P_{TQ_j}(s_2))\\
&= \sum_{i=1}^n \sum_{j=1}^m (p_i * q_j)\sum_{s_1 \cdot s_2 = s} P_{SQ_i \odot_{SQ} TQ_j}(s)\\
&= P_{\langle (SQ_1, p_1) \; | \; \ldots \; | \;(SQ_n,p_n)\rangle \odot \langle (TQ_1, q_1) \; | \; \ldots \; | \;(TQ_m,q_m)\rangle}(s)\\
&= P_{(\Downarrow S_1) \odot (\Downarrow S_2)}(s)
\end{align*}
which is what we wanted to show.
\item
$(S = S_1 \; |_p \; S_2)$

By definition $\Downarrow (S_1 \; |_p \; S_2) = (\Downarrow S_1) \oplus_p (\Downarrow S_2)$. By the induction hypothesis $P_{S_1} = P_{\Downarrow S_1}$ and $P_{S_2} = P_{\Downarrow S_2}$. Assume that $\Downarrow S_1 = \langle (SQ_1,p_1) \; | \; \ldots \; | \; (SQ_m, p_n)\rangle$ and $\Downarrow S_2 = \langle (TQ_1,q_1) \; | \; \ldots \; | \; (TQ_n, q_m) \rangle$. By definition, 

  $\OrDNFOf{\DNFOf{(\Sequence_1,\Probability_1)\DNFSep\ldots\DNFSep(\Sequence_n,\Probability_n)}}{\DNFOf{(\SequenceAlt_1,\ProbabilityAlt_1)\DNFSep\ldots\DNFSep(\SequenceAlt_m,\ProbabilityAlt_m)}}{\Probability} =$\\
  $\DNFOf{(\Sequence_1,\Probability_1*\Probability)\DNFSep\ldots\DNFSep(\Sequence_n,\Probability_n*\Probability)\DNFSep(\SequenceAlt_1,\ProbabilityAlt_1*(1-\Probability))\DNFSep\ldots\DNFSep(\SequenceAlt_m,\ProbabilityAlt_m*(1-\Probability))}$
  hence
  \begin{align*}
  P_{(\Downarrow S_1) \; |_p \; (\Downarrow S_2)}(s) &= P_{\langle (SQ_1,p_1) \; | \; \ldots \; | \; (SQ_n, p_n)\rangle \oplus_p \langle (TQ1, q_1) \; | \; \ldots \; | \; (TQ_m, q_m)\rangle}(s) \\
  &= P_{\langle (SQ_1,p_1) \; | \; \ldots \; | \; (SQ_n, p_n) \; | \; (TQ1, q_1) \; | \; \ldots \; | \; (TQ_m, q_m)\rangle}(s)\\
  &= \sum_{i=1}^m (p * p_i) P_{SQ_i}(s) + \sum_{j=1}^m ((1-p) * q_j) P_{TQ_j}(s)\\
  &= p * \sum_{i=1}^m p_i P_{SQ_i}(s) + (1-p)* \sum_{j=1}^m q_j * P_{TQ_j}(s)\\
&= p * P_{\langle (SQ_1,p_1) \; | \; \ldots \; | \; (SQ_n, p_n)\rangle}(s) + (1-p) * P_{\langle (TQ1, q_1) \; | \; \ldots \; | \; (TQ_m, q_m)\rangle}(s)\\
&= p * P_{\Downarrow S_1}(s) + (1-p) * P_{\Downarrow S_2}(s) 
  \end{align*}
  
  Since $P_{S_1} = P_{\Downarrow S_1}$ and $P_{S_1} = P_{\Downarrow S_1}$, then 
  \begin{align*}
  P_{S_1}(s) &= \sum_{i=1}^n p_i * P_{SQ_i}(s)\\
  P_{S_2}(s) &= \sum_{j=1}^m q_m * P_{TQ_j}(s)
  \end{align*}
  hence
  \begin{align*}
  P_{S_1 \; |_p \; S_2}(s) &= p * (P_{S_1}(s)) + (1-p) * (P_{S_2}(s))\\
  &= p \left(\sum_{i=1}^n p_i * P_{SQ_i}(s)\right) + (1-p)\left(\sum_{j=1}^m q_m * P_{TQ_j}(s)\right)\\
  &= p * P_{\Downarrow S_1}(s) + (1-p) * P_{\Downarrow S_2}(s)\\
  &= P_{(\Downarrow S_1) \; |_p \; (\Downarrow S_2)}(s)
  \end{align*}
  which is what we wanted to show. This completes the proof.
\end{enumerate}
\end{proof}
\begin{theorem}
  \label{CorrectEntropy}
  If $\Regex$ is unambiguous and contains no empty subcomponents,
  $\EntropyOf{\Regex}$ is the entropy of $\mathcal{L}(\Regex)$.
\end{theorem}
\begin{proof}
Given a discrete set $X$ and a probability distribution $P$ on $X$, we define the entropy $H(X)$ of $X$ by
$$H(X) = -\sum_{x \in X}P(x)\log_2{P(x)}$$
We prove the theorem by induction on $S$.
\begin{enumerate}
\item
($S = s$)

Then $\mathbb{H}(s) = 0 = H(\mathcal{L}(s)) = H(\{s\})$
\item
$(S = S^{*p})$

Since $S^{*p}$ is unambiguous, then $S$ is also anambiguous, and for each $s \in \mathcal{L}(S^{*p})$, there exist unique $t_1, \ldots, t_n \in \mathcal{L}(S)$ such that $s = t^s_1 \cdot \ldots \cdot t^s_{n_s}$. By the induction hypothesis, $H(\mathcal{L}(S)) = \mathbb{H}(S)$, thus
\begin{align*}
H(\mathcal{L}(S^{*p})) &= - \sum_{s \in \mathcal{L}(S^{*p})} P_{S^{*p}}(s) \log_2 P_{S^{*p}}(s)\\
&= -(1-p)\sum_n p^n \sum_{\substack{(s_1, \ldots, s_n) \\ s_i \in \mathcal{L}(S)}} \prod_{i=1}^n P_S(s_i)\log_2 \left(p^n * (1 - p) * \prod_{i=1}^n P_S(s_i)\right)\\
&= -(1-p)\sum_n p^n \sum_{\substack{(s_1, \ldots, s_n) \\ s_i \in \mathcal{L}(S)}} \prod_{i=1}^n P_S(s_i)\left(n \log_2 p  + \log_2 (1 - p) + \log_2 \left(\prod_{i=1}^n P_S(s_i)\right)\right)\\
&= -(1-p)\sum_n p^n \sum_{\substack{(s_1, \ldots, s_n) \\ s_i \in \mathcal{L}(S)}} \prod_{i=1}^n P_S(s_i)\left(n \log_2 p  + \log_2 (1 - p) + \log_2 \left(\prod_{i=1}^n P_S(s_i)\right)\right)\\
&= -(1-p) * \log_2 p \sum_n n * p^n \sum_{\substack{(s_1, \ldots, s_n) \\ s_i \in \mathcal{L}(S)}} \prod_{i=1}^n P_S(s_i)
-(1-p) * \log_2 (1 - p) \sum_n p^n \sum_{\substack{(s_1, \ldots, s_n) \\ s_i \in \mathcal{L}(S)}} \prod_{i=1}^n P_S(s_i) \\
&-(1-p)\sum_n p^n \sum_{\substack{(s_1, \ldots, s_n) \\ s_i \in \mathcal{L}(S)}} \prod_{i=1}^n P_S(s_i)\log_2 \left(\prod_{i=1}^n P_S(s_i)\right)\\
&= -\log_2 p * \frac{p}{(1-p)} -\log_2 (1 - p) 
-(1-p)\sum_n p^n H(\mathcal{L}(S^n))\\
&= -\log_2 p * \frac{p}{(1-p)} -\log_2 (1 - p) 
-(1-p)\sum_n p^n n H(\mathcal{L}(S))\\
&= -\log_2 p * \frac{p}{(1-p)} -\log_2 (1 - p) 
-\frac{p}{(1-p)} * H(\mathcal{L}(S))\\
&= \frac{p}{(1-p)}(H(\mathcal{L}(S)) - \log_2 p) -\log_2 (1 - p)\\
&= \frac{p}{(1-p)}(\mathbb{H}(S) - \log_2 p) -\log_2 (1 - p)
\end{align*}
\item
($S = S_1 \cdot S_2$)

Since $S$ is unambigous, then $S_1$ and $S_2$ are also unambiguous, and for each $s \in \mathcal{L}(S_1 \cdot S_2)$ there exists a unique $s_1 \in \mathcal{L}(S_1)$ and a unique $s_2 \in \mathcal{L}(S_2)$ such that $s = s_1 \cdot s_2$. By the induction hypothesis, $H(\mathcal{L}(S_1)) = \mathbb{H}(S_1)$ and $H(\mathcal{L}(S_2)) = \mathbb{H}(S_2)$, hence
\begin{align*}
H(\mathcal{L}(S_1 \cdot S_2))(s) &= -\sum_{s_1 \in S_1} \sum_{s_2 \in S_2}P_{S_1}(s_1)*P_{s_2}(s_2) \log_2(P_{S_1}(s_1)*P_{S_2}(s_2))\\
&= -\sum_{s_1 \in S_1} P_{S_1}(s_1)\sum_{s_2 \in S_2}P_{s_2}(s_2) (\log_2(P_{S_1}(s_1) + \log_2 P_{S_2}(s_2)))\\
&= -\sum_{s_1 \in S_1} P_{S_1}(s_1) \sum_{s_2 \in S_2}P_{S_2}(s_2)\log_2(P_{S_1}(s_1)) + \sum_{s_1 \in S_1} P_{S_1}(s_1)\sum_{s_2 \in S_2} P_{S_2}(s_2)\log_2 P_{S_2}(s_2)\\
&= -\sum_{s_1 \in S_1} P_{S_1}(s_1) \log_2(P_{S_1}(s_1)) + \sum_{s_2 \in S_2} P_{S_2}(s_2)\log_2 P_{S_2}(s_2)\\
&= H(\mathcal{L}(S_1)) + H(\mathcal{L}(S_2))\\
&= \mathbb{H}(S_1) + \mathbb{H}(S_2)
\end{align*}
\item
($S = S_1 \; |_p \; S_2$)

Since $S$ is unambiguous, it must be that $S_1$ and $S_2$ are also unambiguous, and
every $s \in \mathcal{L}(S_1 \; |_p \; S_2)$ is either in $\mathcal{L}(S_1)$
or in $\mathcal{L}(S_2)$. By the induction hypothesis, $H(\mathcal{L}(S_1))
= \mathbb{H}(S_1)$ and $H(\mathcal{L}(S_2)) =
\mathbb{H}(S_2)$. Then
\begin{align*}
H(\mathcal{L}(S_1 \; |_p \; S_2)) &= -\sum_{s \in \mathcal{L}(S_1)}(p * P_{S_1}(s))\log_2 (p * P_{S_1}(s)) - \sum_{s \in \mathcal{L}(S_2)}((1-p) * P_{S_2}(s))\log_2 ((1-p) * P_{S_2}(s))\\
&=- p\sum_{s \in \mathcal{L}(S_1)}P_{S_1}(s)(\log_2 p + \log_2 P_{S_1}(s)) - (1-p)\sum_{s \in \mathcal{L}(S_2)}P_{S_2}(s)(\log_2 (1-p) + \log_2 P_{S_2}(s))\\
&= p * (\log_2 p +  H(\mathcal{L}(S_1))) + (1-p) * ( \log_2 (1-p) + H(\mathcal{L}(S_2)))\\
&= p * (\log_2 p +  \mathbb{H}(S_1)) + (1-p) * ( \log_2 (1-p) + \mathbb{H}(S_2))\\
\end{align*}
\end{enumerate}
\end{proof}
\begin{theorem*}
  Let $\Lens$ be an asymmetric lens. $\Lens$ is also a simple symmetric lens,
  where
  \begin{center}
    \begin{tabular}{rcl}
      $\Lens.\CreateROf{x}$ & $=$ & $\Lens.get \App x$\\
      $\Lens.\CreateLOf{y}$ & $=$ & $\Lens.create \App y$\\
      $\Lens.\PutROf{x}{y}$ & $=$ & $\Lens.get \App x$\\
      $\Lens.\PutLOf{y}{x}$ & $=$ & $\Lens.put \App y \App x$
    \end{tabular}
  \end{center}
\end{theorem*}
\begin{proof}
Let $\ell$ be an asymmetric lens. Then $\ell$ satisfies the following laws:
\begin{align*}
\ell.\get \; (\ell.\pput \; v \; s) &= v \tag{PUTGET}\\
\ell.\pput \; (\ell.\get \; s) \; s&= s \tag{GETPUT}\\
\ell.\get \; (\ell.\pput \; v \; s) &= v \tag{CREATGET}
\end{align*}
Define the simple symmetric lens $\lceil \ell \rceil$ by 
\begin{align*}
\lceil \ell \rceil.\CreateROf{s} &= \ell.\get \; s\\
\lceil \ell \rceil.\CreateLOf{v} &= \ell.\create \; v\\
\lceil \ell \rceil.\PutROf{s}{v} &= \ell.\get \; s\\
\lceil \ell \rceil.\PutLOf{v}{s} &= \ell.\pput \; v \; s\\
\end{align*}
We now show that $\lceil \ell \rceil$ satisfies the simple symmetric lens laws:
\begin{align*}
\lceil \ell \rceil.\PutLOf{(\lceil \ell \rceil.\CreateROf{s})}{s} &= \ell.\pput \; (\ell.\get \; s) \; s = s\\
\lceil \ell \rceil.\PutROf{(\lceil \ell \rceil.\CreateLOf{v})}{v} &= \ell.\get \; (\ell.\create \; v) = v\\
\lceil \ell \rceil.\PutLOf{(\lceil \ell \rceil.\PutROf{s}{v})}{s} &= \ell.\pput \; (\ell.\get) \; s = s\\
\lceil \ell \rceil.\PutROf{(\lceil \ell \rceil.\PutLOf{v}{s})}{v} &= \ell.\get \; (\ell.\pput \; v \; s) = v
\end{align*}
\end{proof}

\begin{theorem*}
  Let $\Lens \OfType \Regex \Leftrightarrow \RegexAlt$, where $\Lens$ does not
  include composition, $\Regex$ and $\RegexAlt$ and unambiguous, and neither
  $\Regex$ nor $\RegexAlt$ contain any empty subcomponents.
  \begin{enumerate}
  \item $\REntropyOf{\RegexAlt \Given \Lens, \Regex}$ is an upper bound for the expected
    information content of $\SetOf{t \SuchThat t \in \LanguageOf{\RegexAlt}}$,
    given $\SetOf{s \SuchThat s \in \LanguageOf{\Regex} \BooleanAnd
      \Lens.\PutROf{s}{t} = t}$
  \item $\LEntropyOf{\Regex \Given \Lens, \RegexAlt}$ is an upper bound for the expected
    information content of $\SetOf{s \SuchThat s \in \LanguageOf{\Regex}}$,
    given $\SetOf{t \SuchThat t \in \LanguageOf{\RegexAlt} \BooleanAnd
      \Lens.\PutLOf{t}{s} = s}$
  \end{enumerate}
\end{theorem*}
\begin{proof}
Let $H(T \; | \; \ell, S)$ be the expected information content of $\{t \; | \; t \in \mathcal{L}(T)\}$ given $\{s \; | \; s \in \mathcal{L}(S) \wedge \ell.\PutROf{s}{t} = t\}$. Given $s \in S$, let $V^{\rightarrow}_{\ell}(s) = \{t \in \mathcal{L}(T) \; | \; \ell.\PutROf{s}{t} = t\}$. We want to show that
$$H(T \; | \; \ell, S) = \sum_{s \in \mathcal{L}(S)}P_{\mathcal{L}(S)}(s)H(V^{\rightarrow}_{\ell}(s)) \leq \mathbb{H}^{\rightarrow}(T \; | \; \ell, S)$$
(The proofs for $\mathbb{H}^{\leftarrow}(S \; | \; \ell, T)$ will follow the same pattern). We proceed by induction over $\ell$. 
\begin{enumerate}
\item
($\ell = \IdentityLensOf{S}$)

Observe that $V^{\rightarrow}_{\IdentityLensOf{S}}(s) = \{t \in \mathcal{L}(S) \; | \; \IdentityLensOf{S}.\ell.\PutROf{s}{t} = t\} = \{t \in \mathcal{L}(S) \; | \; t = s\} = \{s\}$.

Consequently, $H(V^{\rightarrow}_{\IdentityLensOf{S}}(s)) = H(\{s\}) = 0$, from which it follows that 
$$\sum_{s \in \mathcal{L}(S)} P_{\mathcal{L}(S)}(s) * H(V^{\rightarrow}_{\IdentityLensOf{S}}(s)) = 0 = \mathbb{H}^{\rightarrow}(S \; | \; \IdentityLensOf{S}, T)$$
\item
($\ell = \DisconnectOf{\Regex}{\RegexAlt}{\String}{\StringAlt}$)

Observe that 
$$V^{\rightarrow}_{\DisconnectOf{\Regex}{\RegexAlt}{\String}{\StringAlt}}(s) = \{t \in \mathcal{L}(T) \; | \; \DisconnectOf{\Regex}{\RegexAlt}{\String}{\StringAlt}.\PutROf{s}{t} = t\} = \{t \in \mathcal{L}(T) \; | \; t = t\} = \mathcal{L}(T)$$
Consequently
\begin{align*}
\sum_{s \in \mathcal{L}(S)}P_{S}(s)H(V^{\rightarrow}_{\DisconnectOf{\Regex}{\RegexAlt}{\String}{\StringAlt}}(s)) &= \sum_{s \in \mathcal{L}(S)}P_{S}(s)H(\mathcal{L}(T))\\
&= H(\mathcal{L}(T))\\
&= \mathbb{H}(T)\\
&= \mathbb{H}^{\rightarrow}(T \; | \; \DisconnectOf{\Regex}{\RegexAlt}{\String}{\StringAlt}, S)
\end{align*}
\item
($\ell = \IterateLensOf{\Lens} : S^{*p} \Leftrightarrow T^{*q}$)

Observe that 
\begin{align*}
V^{\rightarrow}_{\IterateLensOf{\Lens}}(s_1 \cdot \ldots \cdot s_n) &= \{t_1 \cdot \ldots \cdot t_m \; | \; \IterateLensOf{\Lens}.\PutROf{(s_1 \cdot \ldots \cdot s_n) }{(t_1 \cdot \ldots \cdot t_m)} = t_1 \cdot \ldots \cdot t_m\}\\
&= \{t_1 \cdot \ldots \cdot t_n \; | \; \ell.\PutROf{s_i}{t_i} = t_i \text{ for }1 \leq i \leq n\}\\
&= V^{\rightarrow}_{\ell}(s_1) \cdot \ldots \cdot V^{\rightarrow}_{\ell}(s_n)
\end{align*}
Consequently
\begin{align*}
\sum_{s \in \mathcal{L}(S^{*p})}P_{S^{*p}}(s) H(V^{\rightarrow}_{\IterateLensOf{\Lens}}(s))
&= (1-p)\sum_n p^n \sum_{\substack{(s_1, \ldots, s_n)\\s_i \in \mathcal{L}(S)}} \prod_{i=1}^n P_S(s_i) H(V^{\rightarrow}_{\ell}(s_1) \cdot \ldots \cdot V^{\rightarrow}_{\ell}(s_n))\\
&= (1-p)\sum_n p^n \sum_{\substack{(s_1, \ldots, s_n)\\s_i \in \mathcal{L}(S)}} \prod_{i=1}^n P_S(s_i) \left(\sum_{j=1}^n H(V^{\rightarrow}_{\ell}(s_j)) \right)\\
&= (1-p)\sum_n p^n \sum_{\substack{(s_1, \ldots s_{n-1}) \\ s_i \in \mathcal{L}(S)}} \prod_{i=1}^{n-1}P_S(s_i) \sum_{s_n} P_S(s_n) \left(\sum_{j=1}^n H(V^{\rightarrow}_{\ell}(s_j)) \right)
\end{align*}
We focus on the term $\sum_{s_n} P_S(s_n) \sum_{j=1}^n H(V^{\rightarrow}_{\ell}(s_j))$:
\begin{align*}
\sum_{s_n} P_S(s_n) \sum_{j=1}^n H(V^{\rightarrow}_{\ell}(s_j)) &= \sum_{s_n} P_S(s_n)H(V^{\rightarrow}_{\ell}(s_n)) + P_S(s_n) \sum_{j=1}^{n-1} H(V^{\rightarrow}_{\ell}(s_j))\\
&= H(T \; | \; \ell, S) + \sum_{j=1}^{n-1}H(V^{\rightarrow}_{\ell}(s_j))
\end{align*}
Consequently,
\begin{align*}
\sum_{s \in \mathcal{L}(S^{*p})}P_{S^{*p}}(s) H(V^{\rightarrow}_{\IterateLensOf{\Lens}}(s))
&= (1-p)\sum_n p^n \sum_{\substack{(s_1, \ldots s_{n-1}) \\ s_i \in \mathcal{L}(S)}} \prod_{i=1}^{n-1}P_S(s_i) \left( H(S \; | \; \ell, T) + \sum_{j=1}^{n-1}H(V^{\rightarrow}_{\ell}(s_j))\right)\\
&= (1-p)\sum_n p^n \sum_{\substack{(s_1, \ldots, s_{n-1})\\s_i \in \mathcal{L}(S)}} \prod_{i=1}^{n-1}P_S(s_i) H(S \; | \; \ell, T)\\
&+(1-p)\sum_n p^n \sum_{\substack{(s_1, \ldots, s_{n-1}) \\ s_i \in \mathcal{L}(S)}} \prod_{i=1}^{n-1}P_S(s_i) \left(\sum_{j=1}^{n-1}H(V^{\rightarrow}_{\ell}(s_j))\right)\\
&= H(S \; | \; \ell, T) * (1-p)\sum_n p^n \\
&+(1-p)\sum_n p^n \sum_{\substack{(s_1, \ldots, s_{n-1}) \\ s_i \in \mathcal{L}(S)}} \prod_{i=1}^{n-1}P_S(s_i) \left(\sum_{j=1}^{n-1}H(V^{\rightarrow}_{\ell}(s_j))\right)
\end{align*}
Unrolling the term 
$$\sum_{\substack{(s_1, \ldots, s_{n-1}) \\ s_i \in \mathcal{L}(S)}} \prod_{i=1}^{n-1}P_S(s_i) \left(\sum_{j=1}^{n-1}H(V^{\rightarrow}_{\ell}(s_j))\right)$$ 
$(n-1)$ more times gives
\begin{align*}
H(T^{*q} \; | \; V^{\rightarrow}_{\IterateLensOf{\Lens}}(s), S^{*p}) &=
\sum_{s \in \mathcal{L}(S^{*p})}P_{S^{*p}}(s) H(V^{\rightarrow}_{\IterateLensOf{\Lens}}(s))\\
&= H(T \; | \; \ell, S) * (1-p)\sum_n n * p^n \\
&= \frac{p}{1-p} * H(T \; | \; \ell, S)\\
&\leq \frac{p}{1-p} * \mathbb{H}^{\rightarrow}(T \; | \; \ell, S)\\
&= \mathbb{H}(T^{*q} \; | \; V^{\rightarrow}_{\IterateLensOf{\Lens}}(s), S^{*p})
\end{align*}

which is what we wanted to show.
\item
($S = \ConcatLensOf{\Lens_1}{\Lens_2}$)

Observe that 
\begin{align*}
V^{\rightarrow}_{ \ConcatLensOf{\Lens_1}{\Lens_2}}(s_1 \cdot s_2) &= \{t_1 \cdot t_2 \in \mathcal{L}(T_1 \cdot T_2) \; | \;  \ConcatLensOf{\Lens_1}{\Lens_2}.\PutROf{(s_1 \cdot s_2)}{(t_1 \cdot t_2)} = t_1 \cdot t_2\}\\
&= \{t_1 \cdot t_2 \in \mathcal{L}(T_1 \cdot T_2) \; | \;  \ell_1.\PutROf{s_1}{t_1 } = t_1  \text{ and } \ell_2.\PutROf{t_2}{t_2} = t_2\}\\
&= V^{\rightarrow}_{\ell_1}(s_1) \cdot V^{\rightarrow}_{\ell_2}(s_2)
\end{align*}
Consequently,
\begin{align*}
&H(T_1 \cdot T_2 \; | \; \ConcatLensOf{\Lens_1}{\Lens_2}, S_1 \cdot S_2)\\
& =\sum_{s \in \mathcal{L}(S_1 \cdot S_2)}P_{S_1 \cdot S_2}(s) H(V^{\rightarrow}_{\ConcatLensOf{\Lens_1}{\Lens_2}}(s))\\
&= \sum_{s_1 \in \mathcal{L}(s_1)} P_{S_1}(s_1)\sum_{s_2 \in \mathcal{L}(s_2)}P_{S_2}(s_2)(H(V^{\rightarrow}_{\ell_1}(s_1)) + H(V^{\rightarrow}_{\ell_2}(s_2)))\\
&= \sum_{s_1 \in \mathcal{L}(s_1)} P_{S_1}(s_1)\sum_{s_2 \in \mathcal{L}(s_2)}P_{S_2}(s_2)H(V^{\rightarrow}_{\ell_1}(s_1)) + \sum_{s_1 \in \mathcal{L}(s_1)} P_{S_1}(s_1)\sum_{s_2 \in \mathcal{L}(s_2)}P_{S_2}(s_2)H(V^{\rightarrow}_{\ell_2}(s_2))\\
&= H(T_1 \; | \; \ell_1, S_1) + H(T_2 \; | \; \ell_2, S_2)\\
&\leq \mathbb{H}^{\rightarrow}(T_1 \; | \; \ell_1, S_1) + \mathbb{H}^{\rightarrow}(T_2 \; | \; \ell_2, S_2)\\
&= \mathbb{H}(T_1 \cdot T_2 \; | \; \ConcatLensOf{\Lens_1}{\Lens_2}, S_1 \cdot S_2)
\end{align*}
\item
($S = \SwapLensOf{\Lens_1}{\Lens_2}$)

Similar to the previous case.
\item
($S = \OrLensOf{\Lens_1}{\Lens_2}$)

Observe that 
$$
V^{\rightarrow}_{ \OrLensOf{\Lens_1}{\Lens_2}}(s) = \{t \in \mathcal{L}(T_1 \; |_q \; T_2) \; | \;  \OrLensOf{\Lens_1}{\Lens_2}.\ell.\PutROf{s}{t} = t\}
$$
If $s \in \mathcal{L}(S_1)$, then
$$
V^{\rightarrow}_{ \OrLensOf{\Lens_1}{\Lens_2}}(s) = \{t \in \mathcal{L}(T_1 \; |_q \; T_2) \; | \;  \ell_1.\PutROf{s}{t} = t\} = V^{\rightarrow}_{\ell_1}(s)
$$
while if $s \in \mathcal{L}(S_2)$, then
$$
V^{\rightarrow}_{ \OrLensOf{\Lens_1}{\Lens_2}}(s) = \{t \in \mathcal{L}(T_1 \; |_q \; T_2) \; | \;  \ell_2.\PutROf{s}{t} = t\} = V^{\rightarrow}_{\ell_2}(s)
$$
Consequently,
\begin{align*}
H(T_1 \; |_q \; T_2 \; | \; \OrLensOf{\Lens_1}{\Lens_2}, S_1 \; |_p \; S_2) &= \sum_{s \in \mathcal{L}(S_1 \; |_p \; S_2)}P_{S_1 \; |_p \; S_2}(s) H(V^{\rightarrow}_{\OrLensOf{\Lens_1}{\Lens_2}}(s))\\
&= p * \sum_{s \in \mathcal{L}(S_1)}P_{S_1}(s) H(V^{\rightarrow}_{\ell_1}(s)) + (1-p) * \sum_{s \in \mathcal{L}(S_2)}P_{S_2}(s) H(V^{\rightarrow}_{\ell_2}(s))\\
&= p * (H(T_1 \; | \; \ell_1,S_1)) + (1-p) * (H(T_2 \; | \; \ell_2,S_2))\\
&\leq p * (\mathbb{H}^{\rightarrow}(T_1 \; | \; \ell_1,S_1)) + (1-p) * (\mathbb{H}^{\rightarrow}(T_2 \; | \; \ell_2,S_2))\\
&= \mathbb{H}(T_1 \; |_q \; T_2 \; | \; \OrLensOf{\Lens_1}{\Lens_2}, S_1 \; |_p \; S_2) 
\end{align*}
\item
($S = \MergeROf{\Lens_1}{\Lens_2}$)

Observe that 
$$
V^{\rightarrow}_{ \MergeROf{\Lens_1}{\Lens_2}}(s) = \{t \in \mathcal{L}(T) \; | \;  \MergeROf{\Lens_1}{\Lens_2}.\PutROf{s}{t} = t\}
$$
If $s \in \mathcal{L}(S_1)$, then
$$
V^{\rightarrow}_{ \MergeROf{\Lens_1}{\Lens_2}}(s) = \{t \in \mathcal{L}(T) \; | \;  \ell_1.\PutROf{s}{t} = t\} = V^{\rightarrow}_{\ell_1}(s)
$$
while if $s \in \mathcal{L}(S_2)$, then
$$
V^{\rightarrow}_{ \MergeROf{\Lens_1}{\Lens_2}}(s) = \{t \in \mathcal{L}(T) \; | \;  \ell_2.\PutROf{s}{t} = t\} = V^{\rightarrow}_{\ell_2}(s)
$$
Consequently,
\begin{align*}
H(T \; | \; \MergeROf{\Lens_1}{\Lens_2}, S_1 \; |_p \; S_2) &= \sum_{s \in \mathcal{L}(S_1 \; |_p \; S_2)}P_{S_1 \; |_p \; S_2}(s) H(V^{\rightarrow}_{\MergeROf{\Lens_1}{\Lens_2}}(s))\\
&= p * \sum_{s \in \mathcal{L}(S_1)}P_{S_1}(s) H(V^{\rightarrow}_{\ell_1}(s)) + (1-p) * \sum_{s \in \mathcal{L}(S_2)}P_{S_2}(s) H(V^{\rightarrow}_{\ell_2}(s))\\
&= p * (H(T_1 \; | \; \ell_1, S_1) ) + (1-p) * (H(T_2 \; | \; \ell_2, S_2) )\\
&\leq p * (\mathbb{H}^{\rightarrow}(T_1 \; | \; \ell_1, S_1) ) + (1-p) * (\mathbb{H}^{\rightarrow}(T_2 \; | \; \ell_2, S_2) )\\
&= \mathbb{H}(T \; | \; \MergeROf{\Lens_1}{\Lens_2}, S_1 \; |_p \; S_2) 
\end{align*}
\item
($S = \MergeLOf{\Lens_1}{\Lens_2}$)

Observe that 
\begin{align*}
V^{\rightarrow}_{ \MergeLOf{\Lens_1}{\Lens_2}}(s) &= \{t \in \mathcal{L}(T_1 \; |_q \; T_2) \; | \;  \MergeLOf{\Lens_1}{\Lens_2}.\PutROf{s}{t} = t\}\\
&= \{t \in \mathcal{L}(T_1 \; |_q \; T_2) \; | \;  \ell_1.\PutROf{s}{t} = t \text{ or } \ell_2.\PutROf{s}{t} = t\}\\
&= V^{\rightarrow}_{\ell_1}(s) \cup V^{\rightarrow}_{\ell_2}(s)
\end{align*}
To compute $H(V^{\rightarrow}_{ \MergeLOf{\Lens_1}{\Lens_2}}(s))$, we use the formula 
$$H(P_1, P_2) = H(m(P_1), m(P_2)) + m(P_1)H\left(\frac{P_1}{m(P_1)}\right) + m(P_2)H\left(\frac{P_2}{m(P_2)}\right)$$
if the probability space $P$ is the sum of measure spaces $P_1$ and $P_2$, and $m(P_i)$ is the measure of $P_i$:
\begin{align*}
&H(V^{\rightarrow}_{ \MergeLOf{\Lens_1}{\Lens_2}}(s)) =\\
&H \left(\frac{q * P_{T_1}V^{\rightarrow}_{\ell_1}(s)}{q * P_{T_1}V^{\rightarrow}_{\ell_1}(s) + (1-q) * P_{T_2}V^{\rightarrow}_{\ell_2}(s)}, \frac{(1-q) * P_{T_2}V^{\rightarrow}_{\ell_2}(s)}{q * P_{T_1}V^{\rightarrow}_{\ell_1}(s) + (1-q) * P_{T_2}V^{\rightarrow}_{\ell_2}(s)}\right)\\
&+ \frac{q * P_{T_1}V^{\rightarrow}_{\ell_1}(s)}{q * P_{T_1}V^{\rightarrow}_{\ell_1}(s) + (1-q) * P_{T_2}V^{\rightarrow}_{\ell_2}(s)} H(V^{\rightarrow}_{\ell_1}(s)) + \frac{(1-q) * P_{T_2}V^{\rightarrow}_{\ell_2}(s)}{q * P_{T_1}V^{\rightarrow}_{\ell_1}(s) + (1-q) * P_{T_2}V^{\rightarrow}_{\ell_2}(s)} H(V^{\rightarrow}_{\ell_2}(s))
\end{align*}
Consequently,
\begin{align*}
&H(T_1 \; |_q \; T_2 \; | \; \MergeLOf{\Lens_1}{\Lens_2}, S) = \\
&\sum_{s \in \mathcal{L}(S)}H \left(\frac{q * P_{T_1}V^{\rightarrow}_{\ell_1}(s)}{q * P_{T_1}V^{\rightarrow}_{\ell_1}(s) + (1-q) * P_{T_2}V^{\rightarrow}_{\ell_2}(s)}, \frac{(1-q) * P_{T_2}V^{\rightarrow}_{\ell_2}(s)}{q * P_{T_1}V^{\rightarrow}_{\ell_1}(s) + (1-q) * P_{T_2}V^{\rightarrow}_{\ell_2}(s)}\right)P_S(s)\\
&+ \sum_{s \in \mathcal{L}(S)}\left(\frac{q * P_{T_1}V^{\rightarrow}_{\ell_1}(s)}{q * P_{T_1}V^{\rightarrow}_{\ell_1}(s) + (1-q) * P_{T_2}V^{\rightarrow}_{\ell_2}(s)}\right)P_S(s)H(V^{\rightarrow}_{\ell_1}(s))\\
&+ \sum_{s \in \mathcal{L}(S)}\left(\frac{(1-q) * P_{T_2}V^{\rightarrow}_{\ell_2}(s)}{q * P_{T_1}V^{\rightarrow}_{\ell_1}(s) + (1-q) * P_{T_2}V^{\rightarrow}_{\ell_2}(s)}\right)P_S(s)H(V^{\rightarrow}_{\ell_2}(s))
\end{align*}
Since
\begin{align*}
H \left(\frac{q * P_{T_1}V^{\rightarrow}_{\ell_1}(s)}{q * P_{T_1}V^{\rightarrow}_{\ell_1}(s) + (1-q) * P_{T_2}V^{\rightarrow}_{\ell_2}(s)}, \frac{(1-q) * P_{T_2}V^{\rightarrow}_{\ell_2}(s)}{q * P_{T_1}V^{\rightarrow}_{\ell_1}(s) + (1-q) * P_{T_2}V^{\rightarrow}_{\ell_2}(s)}\right) &\leq 1\\
\left(\frac{q * P_{T_1}V^{\rightarrow}_{\ell_1}(s)}{q * P_{T_1}V^{\rightarrow}_{\ell_1}(s) + (1-q) * P_{T_2}V^{\rightarrow}_{\ell_2}(s)}\right) &\leq 1\\
\left(\frac{(1-q) * P_{T_2}V^{\rightarrow}_{\ell_2}(s)}{q * P_{T_1}V^{\rightarrow}_{\ell_1}(s) + (1-q) * P_{T_2}V^{\rightarrow}_{\ell_2}(s)}\right) &\leq 1
\end{align*}
Then
\begin{align*}
H(T_1 \; |_q \; T_2 \; | \; \MergeLOf{\Lens_1}{\Lens_2}, S) & \leq \sum_{s \in \mathcal{L}(s)}P_S(s) + \sum_{s \in \mathcal{L}(s)}P_S(s)H(V^{\rightarrow}_{\ell_1}(s)) + \sum_{s \in \mathcal{L}(s)}P_S(s)H(V^{\rightarrow}_{\ell_2})(s)\\
&= H(T_1 \; | \; \ell_1, S) + H(T_2 \; | \; \ell_2, S)\\
&\leq 1 + \mathbb{H}^{\rightarrow}(T_1 \; | \; \ell_1, S) + \mathbb{H}^{\rightarrow}(T_2 \; | \; \ell_2, S)
\end{align*}
\item
($S = \REntropyOf{\Regex \Given \InvertOf{\Lens}, \RegexAlt}$)

Observe that $$V^{\rightarrow}_{\InvertOf{\Lens}}(t) = \{s \in \mathcal{L}(S) \; | \; \InvertOf{\Lens}.\PutROf{s}{t} = t\} = \{s \in \mathcal{L}(S) \; | \; \ell.\PutLOf{t}{s} = s\} = V^{\leftarrow}_{\ell}(t)$$
Consequently
$$
H(S \; | \; \InvertOf{\Lens}, T) = \sum_{t \in \mathcal{L}(T)}P_T(t) H(V^{\rightarrow}_{\InvertOf{\Lens}}(t))
= \sum_{t \in \mathcal{L}(T)}P_T(t) H(V^{\leftarrow}_{\ell}(t))
\leq \mathbb{H}^{\leftarrow}(S\; | \; \ell, T)
$$
\end{enumerate}

\end{proof}
\begin{theorem}
If $DS$ is unambiguous, then $\EntropyOf{\DNFRegex}$ is the entropy of $P_{\DNFRegex}$.
\end{theorem}
\begin{proof}
By mutual induction over $DS$, $SQ$ and $A$ (the computations of the entropies follow the same pattern as the respective computations in \cref{CorrectEntropy}).
\begin{enumerate}
\item
($A = DS^{*p}$)

By the induction hypothesis $H(DS) = \mathbb{H}(DS)$, from which it follows that
$$H(DS^{*p}) = \frac{p}{1-p}(H(DS) - \log_2 p) - \log_2 (1-p) = \mathbb{H}(DS^{*p})$$
\item
($SQ = [s_0 \cdot A_1 \cdot \ldots \cdot A_n \cdot s_n]$)

By the induction hypothesis, hence by unambiguity $H(A_i) = \mathbb{H}(A_i)$, thus
$$H([s_0 \cdot A_1 \cdot \ldots \cdot A_n \cdot s_n]) = \sum_{i=1}^n H(A_i) = \sum_{i=1}^n \mathbb{H}(A_i) = \mathbb{H}([s_0 \cdot A_1 \cdot \ldots \cdot A_n \cdot s_n])$$
\item
$(DS = \langle (SQ_1, p_1) \; | \; \ldots \; | \; (SQ_n, p_n)\rangle)$

By the induction hypothesis, $H(SQ_i, p_i) = \mathbb{H}(SQ_i, p_i)$. Using the formula
$$H(P) = H(m(P_1), \ldots m(P_n)) + \sum_{i=1}^n m(P_i)H(P_i)$$
if $P$ is a measure space satisfying $P = P_1 + \ldots + P_n$, then
\begin{align*}
H(\langle (SQ_1, p_1) \; | \; \ldots \; | \; (SQ_n, p_n)\rangle)
&= H(p_1, \ldots, p_n) + \sum_{i=1}^n p_i H(SQ_i, p_i)\\
&= \sum_{i=1}^n p_i \mathbb{H}(SQ_i, p_i) + p_i * \log_2 p_i\\
&= \mathbb{H}(\langle (SQ_1, p_1) \; | \; \ldots \; | \; (SQ_n, p_n)\rangle)
\end{align*}
\end{enumerate}
\end{proof}

\fi
\end{document}